\documentclass[preprint,12pt,authoryear]{elsarticle}

\usepackage{setspace}                %
\usepackage{fancyhdr}                %
\usepackage{times}                   
\usepackage{mathptmx}                

\usepackage{amsmath,amssymb,amsfonts}
\usepackage{amsbsy,latexsym}          
\usepackage{amsthm}                  %
\usepackage{mathrsfs}                
\usepackage{cancel}                  
\usepackage{tfrupee}                 %
\usepackage[nointegrals]{wasysym}   %
\usepackage{adjustbox}      

\usepackage{multirow}                %
\usepackage{booktabs}                
\usepackage{tabulary}                %
\usepackage{array}                   %
\usepackage{tabularx}                %
\usepackage{colortbl}                
\usepackage{xcolor}                  
\usepackage{makecell}                %

\usepackage{graphicx}                %
\usepackage{float}                   
\usepackage{morefloats}              %
\usepackage{floatflt}                %
\usepackage{subcaption}              
\usepackage[graphicx]{realboxes}     %
\usepackage{multirow}       

\usepackage{algorithm}               %
\usepackage{algorithmicx}            %
\usepackage{algpseudocode}           
\usepackage{listings}                %

\usepackage{hyperref}                %
\usepackage{cleveref}                
\usepackage{tocbibind}               

\usepackage{textcomp}                
\usepackage{pifont}                  
\usepackage{manyfoot}                %
\usepackage{url}                     

\usepackage{amssymb}
\usepackage{amsmath}

\theoremstyle{thmstyleone}%
\newtheorem{Assumption}{Assumption}
\newtheorem{theorem}{Theorem}
\newtheorem{Lemma}{Lemma}[section]%
\theoremstyle{thmstyletwo}%

\theoremstyle{thmstylethree}%
\newtheorem{definition}{Definition}[section]%
\journal{The Journal of Finance and Data Science}

\begin{document}

\begin{frontmatter}



\title{SBCA: Cross-Modal BERT-driven Actor-Critic for Multi-Asset Portfolio Optimization} 

\author[a]{Jinfeng Pan\corref{cor1}} 
\ead{ppjinfeng} 

\author[b]{Jiahao Chen} 
\ead{changahou@stu2024.jnu.edu.cn} 

\cortext[cor1]{Corresponding author: Jinfeng Pan}

\affiliation[a]{organization={School of Accounting, Guangdong University of Finance},%
	addressline={Longdong Road},
	city={Guangzhou},
	postcode={510521},
	state={Guangdong},
	country={China}}
\affiliation[b]{organization={School of Economics, Jinan University},
	addressline={Huangpu Avenue}, 
	city={Guangzhou}, 
	postcode={510000}, 
	state={Guangdong}, 
	country={China}} 

\begin{abstract}
	
Portfolio optimization is constrained by linear assumptions and insufficient integration of multi-modal information in traditional models. This paper proposes a cross-modal BERT-driven Actor-Critic framework (SBCA) for multi-asset portfolio optimization to address the deficiencies of existing deep reinforcement learning (DRL) methods in fusing price data and financial text sentiment, as well as lacking practical trading constraints. The framework adopts a cross-modal gated fusion mechanism to adaptively integrate price time-series features and text semantic features, embeds downside risk and turnover penalty constraints into the reward function, and constructs a complete empirical system for validation. Experiments on 11-year U.S. stock multi-asset datasets show that SBCA outperforms equal weight, buy-and-hold and market benchmark strategies in portfolio value, annual return, Sharpe ratio and maximum drawdown. Ablation studies verify the complementary enhancement of Actor-Critic mechanism and cross-modal fusion module. Cost sensitivity analysis confirms the model’s robustness under varying transaction costs. SBCA provides an effective and interpretable end-to-end solution for dynamic quantitative portfolio decision-making.

\end{abstract}
\begin{keyword}
Portfolio optimization\sep 
Cross-modal learning\sep 
Deep reinforcement learning\sep 
BERT\sep 
Quantitative investment

\end{keyword}

\end{frontmatter}



\section{Introduction}\label{sec:introduction}

With the rapid development of quantitative investment and artificial intelligence, multi‑asset dynamic portfolio allocation has become a key direction for improving returns and controlling risks in real financial markets. The growing volume of financial news, research reports, and public opinion texts provides rich incremental information for market prediction, yet such unstructured data remains underutilized in mainstream decision models.

Portfolio optimization is a core research topic in the interdisciplinary field of finance, econometrics, and artificial intelligence. Financial asset returns exhibit complex characteristics such as nonlinearity, non-stationarity, and modal dependence, which are affected by multiple dynamic factors. Traditional portfolio optimization models have obvious limitations due to their reliance on linear assumptions and stationary distribution conditions, and they cannot effectively integrate unstructured information such as textual public opinion, failing to meet the practical demands of quantitative investment.

\begin{figure}[H]
	\centering
	\includegraphics[width=\linewidth]{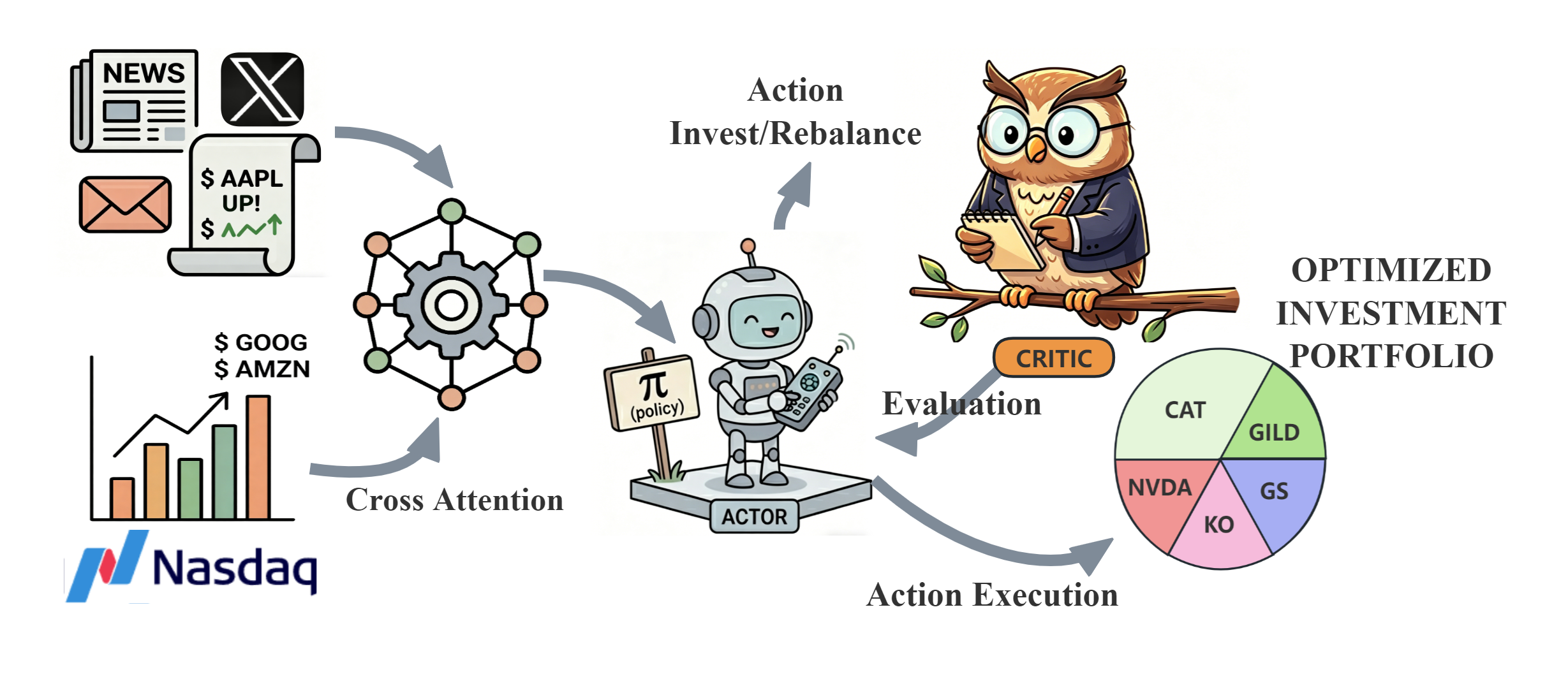}
	\caption{SBCA Task}
	\label{fig:taskfig}
\end{figure}

Deep reinforcement learning (DRL) provides an effective paradigm for dynamic portfolio optimization with its advantages in sequential decision-making and dynamic environment adaptation. Meanwhile, pre-trained language models represented by BERT can extract valuable semantic and sentiment features from unstructured financial texts. However, existing research still suffers from several critical limitations. Multi-modal fusion between price sequences and textual sentiment is mostly superficial and fails to capture adaptive interactive relations across modalities. Most existing methods overlook real-world trading constraints including downside risk and portfolio turnover, which significantly weakens their practical applicability. In addition, empirical validation is often insufficient and lacks systematic ablation tests and cost sensitivity analysis to verify model robustness and generalization.

To address these issues, this paper proposes a cross-modal BERT-driven Actor-Critic framework (SBCA) for multi-asset portfolio optimization. The framework introduces a cross-modal gated fusion module to achieve adaptive integration of price time-series features and BERT-extracted textual semantic features. It further embeds downside risk penalties and portfolio turnover constraints into the reward function to align model training with realistic trading environments. Moreover, it establishes a complete and rigorous empirical pipeline including ablation studies and cost sensitivity analysis to ensure model reliability and reproducibility. The framework adopts a cross-modal gated fusion mechanism to adaptively integrate price time-series features and text semantic features, embeds downside risk and turnover penalty constraints into the reward function, and constructs a complete empirical system for validation. The main contributions of this paper are summarized in three aspects:

\begin{itemize}
	\item A novel cross-modal gated fusion approach is developed to integrate price and textual semantic features, enabling interpretable heterogeneous information fusion for financial applications.
	\item Reinforcement learning-based optimization is enhanced with downside risk and turnover constraints to balance portfolio performance and improve model training reliability.
	\item A comprehensive empirical framework is established to validate the model’s effectiveness and provide reproducible support for its practical deployment in multi-asset portfolio optimization.
\end{itemize}

The subsequent structure of this paper is as follows: \Cref{sec:related_work} systematically reviews relevant research on portfolio optimization and deep reinforcement learning in the financial field; \Cref{sec:theoretical_foundation} elaborates the theoretical foundation and model structure of the SBCA framework; \Cref{sec:experiments} introduces data sources, experimental design, and empirical results; \Cref{sec:conclution} summarizes research conclusions, limitations, and future prospects.

\section{Related Work}\label{sec:related_work}

\subsection{Portfolio Optimization and Deep Reinforcement Learning}

Traditional portfolio optimization began with the mean-variance theory proposed by Markowitz \cite{markowitz1952portfolio}, which laid the fundamental theoretical foundation for portfolio risk diversification. However, this classic model suffers from obvious defects, including severe estimation errors and unrealistic linear and stationary assumptions \cite{demiguel2009optimal}, which make it difficult to adapt to the complex and volatile real financial market. Later, scholars extended the mean-variance model by incorporating transaction costs \cite{magill1976portfolio}, downside risk measures, and online learning rules, trying to improve its practical applicability. 

With the rapid development of artificial intelligence, deep reinforcement learning has reshaped the research paradigm of dynamic portfolio management by virtue of its model-free, end-to-end sequential decision-making capabilities. Early DRL frameworks for portfolio optimization focused on solving the problem of continuous action spaces and designed portfolio vector memory mechanisms to support adaptive asset allocation \cite{jiang2017deep}. On this basis, model-based DRL methods further improved the stability of portfolio strategies by integrating environment prediction and data augmentation technologies \cite{yu2019model}, making the decision-making process more in line with the dynamic changes of the market. 

In recent years, Actor–Critic variants such as DDPG, PPO, and SAC have become the dominant frameworks in DRL-based portfolio optimization, due to their excellent performance in handling continuous portfolio weighting problems. A series of improved methods have been proposed to address the practical constraints of real trading: cost-sensitive DRL models \cite{zhang2020cost} explicitly consider transaction costs in the training objective, DeepTrader \cite{wang2021deeptrader} optimizes the decision-making process through multi-dimensional market state perception, and dynamic rebalancing RL models \cite{lim2022dynamic} adaptively adjust portfolio weights according to market structure mutations, all of which have achieved significant improvements in risk-adjusted returns. In addition, platforms such as FinRL and FinRL-Meta have standardized the empirical evaluation of DRL-based portfolio strategies \cite{liu2021finrl,liu2022finrlmeta}, providing a unified benchmark for academic research. Despite these remarkable progresses, most existing DRL models still do not systematically use financial text information or conduct in-depth fusion of multi-modal data, which becomes a key bottleneck restricting their performance improvement.

\subsection{Cross-Modal Learning and Financial Language Models}

With the rise of pre-trained language models, financial language models have developed rapidly, among which BERT and its domain-adapted versions (such as FinBERT) have become the core tools for extracting information from unstructured financial texts \cite{araci2019finbert,yang2020finbert,liu2020finbert}. These models can effectively capture semantic information and sentiment tendencies from financial news, research reports, and social media content, providing valuable decision-making information for financial markets. Specifically, they have been widely applied in stock movement prediction \cite{chen2021stock,wei2020stock,li2020applying} and investor sentiment analysis \cite{shen2024financial,cicekyurt2026enhancing}, and have achieved good results in improving the accuracy of market prediction.

Cross-modal learning, as a frontier research direction in artificial intelligence, aims to fuse structured price data with unstructured textual information, so as to comprehensively capture the multi-dimensional driving factors of financial markets. In the field of portfolio optimization, scholars have made initial attempts in cross-modal fusion: Multimodal Transformers and cross-modal attention mechanisms have been used to combine price patterns and semantic signals, effectively improving the accuracy of market forecasting \cite{pei2025cross,gohari2024modality}. Some studies have incorporated textual sentiment features into the state space of reinforcement learning \cite{nawathe2024multimodal,mantshimuli2025sentiment}, while others have fused large language models (LLMs) with policy gradient methods to enhance the interpretability of portfolio strategies \cite{xiong2025flagtrader}.

However, transaction costs are an important factor affecting strategy returns, and their fluctuations will directly erode the profits of portfolio strategies \cite{choudhary2025risk}. Existing cross-modal DRL studies rarely conduct systematic sensitivity analysis on transaction costs, which makes it difficult to verify the robustness of strategies in different market environments. In addition, the generalization ability of portfolio strategies is also a key concern in academic research and practical deployment \cite{sood2026deep}, but existing models lack a unified and rigorous empirical testing system, which makes it difficult to fully verify their generalization and robustness. This paper will focus on these issues and conduct in-depth empirical research to provide reliable theoretical and practical support for the application of cross-modal DRL in portfolio optimization.

\section{Theoretical Foundation}\label{sec:theoretical_foundation}

\subsection{Notation Convention}
\label{subsec:notation}
This paper adopts a unified notation system, and all symbols maintain consistent meanings throughout the paper unless otherwise specified.

Denote the set of real numbers as $\mathbb{R}$, the set of non-negative real numbers as $\mathbb{R}_+$, and the set of positive integers as $\mathbb{N}^*$. For any $n \in \mathbb{N}^*$, $\mathbb{R}^n$ denotes the $n$-dimensional Euclidean space, $\mathbf{1} \in \mathbb{R}^n$ denotes the all-ones vector, $\|\cdot\|_1$ denotes the $L_1$ norm, and $\odot$ denotes the Hadamard product (element-wise multiplication). For a random variable $X$ defined on a measurable space, $\mathbb{E}[X]$ denotes its mathematical expectation with respect to the probability measure.

All stochastic processes in this paper are defined on a complete probability space $(\Omega, \mathcal{F}, \mathbb{P})$, where $\{\mathcal{F}_t\}_{t \in \mathcal{T}}$ is a filtration satisfying the usual conditions on $(\Omega, \mathcal{F})$, the discrete-time index set is $\mathcal{T} = \{0,1,\dots,T\}$, $T \in \mathbb{N}^*$ is the terminal time of the investment horizon, and all stochastic processes are adapted to this filtration.

In financial market modeling, let the number of tradable risky assets be $N$. For any $i=1,2,\dots,N$ and $t \in \mathcal{T}$, the closing price of the $i$-th asset at time $t$ is $p_{t,i}$, the single-period simple return is $y_{t,i} = p_{t,i}/p_{t-1,i}$, the single-period log return is $r_{t,i} = \ln y_{t,i}$, and the asset return vector at time $t$ is $\mathbf{y}_t \in \mathbb{R}_+^N$; the portfolio weight vector at time $t$ is $\mathbf{w}_t \in \mathbb{R}^N$, the turnover rate between adjacent trading days is $TO_t$, the one-way trading commission rate is $c$, the unit net value of the portfolio at time $t$ is $V_t$, and the investor's utility function is $U(\cdot)$.

For feature space notations, the length of the price time series observation window is $W$, the price feature space is $\mathcal{X}_p = \mathbb{R}^{WN}$, the text feature space is $\mathcal{X}_s = \mathbb{R}^N$, and the shared latent space for features is $\mathcal{H} = \mathbb{R}^H$; the encoding mappings for price and text features are denoted as $f_p$ and $f_s$ respectively, and the cross-modal gated fusion mapping is denoted as $f_{fusion}$; the downside risk penalty coefficient and turnover penalty coefficient are $\lambda_{risk}$ and $\lambda_{turnover}$ respectively, with $\lambda_{risk}, \lambda_{turnover} > 0$.

In the Markov Decision Process (MDP) framework, the discrete-time MDP is denoted as $\mathcal{M}=(\mathcal{S}, \mathcal{A}, P, R, \gamma)$, where $\mathcal{S}$ is the non-empty measurable state space, $\mathcal{A}$ is the non-empty measurable action space, $P$ is the state transition kernel, $R$ is the bounded measurable reward function, and $\gamma \in [0,1]$ is the discount factor; the decision policy is denoted as $\pi$, and the cumulative discounted return at time $t$ is denoted as $G_t$; given a policy $\pi$, the state value function, action value function, and advantage function are denoted as $V^\pi$, $Q^\pi$, and $A^\pi$ respectively, and the single-step temporal difference error is denoted as $\delta_t$.

For model optimization notations, the parameters of the Actor policy network and Critic value network are $\theta$ and $\phi$ respectively; the policy optimization objective function is $\mathcal{J}(\theta)$, and the loss functions of the Actor and Critic networks are $\mathcal{L}_{Actor}$ and $\mathcal{L}_{Critic}$ respectively.

\subsection{Assumptions and Definitions}
\label{subsec:assumptions_definitions}
To conduct rigorous theoretical analysis and model construction, this paper proposes the following basic assumptions:

\begin{Assumption}
	The financial market has no frictional costs other than trading commissions, no short-selling mechanism, all assets have infinite liquidity, and trading behavior does not affect asset prices.
\end{Assumption}

\begin{Assumption}
	Market participants can costlessly obtain all historical information contained in $\mathcal{F}_t$ at time $t$, portfolio decisions are based solely on historically available information, and there is no look-ahead bias.
\end{Assumption}

\begin{Assumption}
	The dynamic evolution of asset prices satisfies the Markov property, i.e., given the current market state, the conditional distribution of future prices is independent of historical states.
\end{Assumption}

\begin{Assumption}
	Investors are risk-averse, their preferences are characterized by a Constant Relative Risk Aversion (CRRA) utility function, and the decision objective is to maximize the expected utility of terminal wealth.
\end{Assumption}

\begin{Assumption}[Monotonicity of Penalty Terms]
	For any two admissible portfolio strategies $\pi_1$ and $\pi_2$, if $\mathbb{E}_{\pi_1}[\ln V_T] > \mathbb{E}_{\pi_2}[\ln V_T]$, then the expected cumulative penalty term of strategy $\pi_1$ is no higher than that of strategy $\pi_2$, i.e., $\mathbb{E}_{\pi_1}\left[\sum_{t=1}^T \mathcal{P}_t\right] \leq \mathbb{E}_{\pi_2}\left[\sum_{t=1}^T \mathcal{P}_t\right]$.
	\label{assump:penalty_monotonicity}
\end{Assumption}

\begin{Assumption}[Regularity of Feature Mapping]
	The Layer Normalization (LayerNorm) module in the feature encoding mapping includes learnable gain parameter $\gamma$ and bias parameter $\beta$, with its standard form:
	$$
	\text{LayerNorm}(x) = \gamma \cdot \frac{x - \mathbb{E}[x]}{\sqrt{\text{Var}(x) + \epsilon}} + \beta
	$$
	where $\epsilon>0$ is a constant for numerical stability. There exist learnable parameters $\gamma, \beta$ such that $\text{LayerNorm}(x) = x$ holds for any $x$ in the domain.
	\label{assump:layernorm_identity}
\end{Assumption}

Based on the above assumptions, the core definitions of this paper are given as follows:

\begin{definition}[Price Process]
	The price process $\{p_{t,i}\}_{t \in \mathcal{T}}$ of the $i$-th risky asset is adapted to the filtration $\{\mathcal{F}_t\}$, and $p_{t,i} > 0$ for any $t \in \mathcal{T}$.
\end{definition}

\begin{definition}[Admissible Portfolio Strategy]
	A sequence $\{\mathbf{w}_t\}_{t=0}^{T-1}$ is called an admissible portfolio strategy if, for any $t \in \{0,1,\dots,T-1\}$, $\mathbf{w}_t$ is $\mathcal{F}_t$-measurable and satisfies the long-only constraint and budget constraint:
	\begin{equation}
		\mathbf{1}^\top \mathbf{w}_t = 1, \quad w_{t,i} \geq 0, \quad \forall i=1,2,\dots,N
		\label{eq:admissible_policy}
	\end{equation}
\end{definition}

\begin{definition}[Turnover and Net Value Process]
	Given an admissible portfolio strategy $\{\mathbf{w}_t\}$, the turnover rate at time $t$ is defined as:
	\begin{equation}
		TO_t = \frac{1}{2}\|\mathbf{w}_t - \mathbf{w}_{t-1}\|_1
		\label{eq:turnover}
	\end{equation}
	where the initial position weight is set as $\mathbf{w}_{-1}=\mathbf{1}/N$ by convention; the portfolio net value process $\{V_t\}_{t \in \mathcal{T}}$ satisfies the initial condition $V_0=1$, and for any $t \geq 1$, has the recurrence relation:
	\begin{equation}
		V_t = V_{t-1} \cdot \left( \mathbf{w}_{t-1}^\top \mathbf{y}_t - c \cdot TO_t \right)
		\label{eq:VT_recursion}
	\end{equation}
\end{definition}

\begin{definition}[Discrete-Time Markov Decision Process]
	A 5-tuple $\mathcal{M}=(\mathcal{S}, \mathcal{A}, P, R, \gamma)$ is called a discrete-time Markov Decision Process, where: $\mathcal{S}$ is the non-empty measurable state space, and the state $s_t$ at time $t$ includes the time series features of asset prices and the sentiment features of news texts; $\mathcal{A}$ is the non-empty measurable action space, and the action $a_t$ is the admissible portfolio weight vector; $P: \mathcal{S} \times \mathcal{A} \to \mathcal{P}(\mathcal{S})$ is the state transition kernel, where $\mathcal{P}(\mathcal{S})$ is the space of probability measures on $\mathcal{S}$; $R: \mathcal{S} \times \mathcal{A} \times \mathcal{S} \to \mathbb{R}$ is the bounded measurable reward function; $\gamma \in [0,1]$ is the discount factor.
\end{definition}

\begin{definition}[Policy and Cumulative Discounted Return]
	A mapping $\pi: \mathcal{S} \to \mathcal{P}(\mathcal{A})$ is called a stochastic policy, where $\pi(a|s)$ is the conditional probability of selecting action $a$ in state $s$; if $\pi(\cdot|s)$ is a degenerate point distribution for any $s \in \mathcal{S}$, then $\pi$ is called a deterministic policy. Given a policy $\pi$, the cumulative discounted return at time $t$ is defined as:
	\begin{equation}
		G_t = \sum_{k=0}^{T-t} \gamma^k R(s_{t+k},a_{t+k},s_{t+k+1})
		\label{eq:cumulative_return}
	\end{equation}
	where $a_{t+k} \sim \pi(\cdot|s_{t+k})$ and $s_{t+k+1} \sim P(\cdot|s_{t+k},a_{t+k})$.
\end{definition}

\begin{definition}[Value Function and Advantage Function]
	Given a policy $\pi$, the state value function $V^\pi: \mathcal{S} \to \mathbb{R}$ is defined as:
	\begin{equation}
		V^\pi(s) = \mathbb{E}_\pi\left[ G_t | s_t = s \right]
		\label{eq:state_value}
	\end{equation}
	the action value function $Q^\pi: \mathcal{S} \times \mathcal{A} \to \mathbb{R}$ is defined as:
	\begin{equation}
		Q^\pi(s,a) = \mathbb{E}_\pi\left[ G_t | s_t = s, a_t = a \right]
		\label{eq:action_value}
	\end{equation}
	the advantage function $A^\pi: \mathcal{S} \times \mathcal{A} \to \mathbb{R}$ is defined as:
	\begin{equation}
		A^\pi(s,a) = Q^\pi(s,a) - V^\pi(s)
		\label{eq:advantage}
	\end{equation}
\end{definition}

\begin{definition}[Temporal Difference Error]
	Given a policy $\pi$ and a parameterized state value function $V_\phi(s)$, the single-step temporal difference error is defined as:
	\begin{equation}
		\delta_t = R(s_t,a_t,s_{t+1}) + \gamma V_\phi(s_{t+1}) - V_\phi(s_t)
		\label{eq:td_error}
	\end{equation}
\end{definition}

\begin{definition}[Cross-Modal Feature Encoding and Fusion Mapping]
	The price feature encoding mapping $f_p: \mathcal{X}_p \to \mathcal{H}$ is defined as:
	\begin{equation}
		f_p(\mathbf{p}) = \text{ReLU}\left( \text{LayerNorm}\left( \mathbf{W}_p \mathbf{p} + \mathbf{b}_p \right) \right)
		\label{eq:price_encoding}
	\end{equation}
	the text feature encoding mapping $f_s: \mathcal{X}_s \to \mathcal{H}$ is defined as:
	\begin{equation}
		f_s(\mathbf{s}) = \text{ReLU}\left( \text{LayerNorm}\left( \mathbf{W}_s \mathbf{s} + \mathbf{b}_s \right) \right)
		\label{eq:text_encoding}
	\end{equation}
	the cross-modal gated fusion mapping $f_{fusion}: \mathcal{H} \times \mathcal{H} \to \mathcal{H}$ is defined as:
	\begin{equation}
		f_{fusion}(\mathbf{f}_p, \mathbf{f}_s) = \text{ReLU}\left( \mathbf{W}_f \left( \mathbf{f}_p \odot (1 + \mathbf{g}) + \mathbf{f}_s \right) + \mathbf{b}_f \right)
		\label{eq:fusion_mapping}
	\end{equation}
	where the gating weight $\mathbf{g} = \text{Tanh}\left( \mathbf{W}_g \mathbf{f}_s + \mathbf{b}_g \right)$, and $\mathbf{W}_p,\mathbf{b}_p,\mathbf{W}_s,\mathbf{b}_s,\mathbf{W}_f,\mathbf{b}_f,\mathbf{W}_g,\mathbf{b}_g$ are learnable parameters.
\end{definition}

\begin{definition}[Risk-Sensitive Reward Function]
	The risk-sensitive reward function designed in this paper is defined as:
	\begin{equation}
		r_t = \ln \left( \mathbf{w}_{t-1}^\top \mathbf{y}_t - c \cdot TO_t \right) - \lambda_{risk} \cdot \left[ \max\left( - \ln \left( \mathbf{w}_{t-1}^\top \mathbf{y}_t - c \cdot TO_t \right), 0 \right) \right]^2 - \lambda_{turnover} \cdot TO_t
		\label{eq:reward_def}
	\end{equation}
	where the single-period penalty term is denoted as:
	\begin{equation}
		\mathcal{P}_t = \lambda_{risk} \cdot \left[ \max\left( - \ln \left( \mathbf{w}_{t-1}^\top \mathbf{y}_t - c \cdot TO_t \right), 0 \right) \right]^2 + \lambda_{turnover} \cdot TO_t
		\label{eq:penalty_def}
	\end{equation}
\end{definition}

\subsection{Lemmas}\label{subsec:lemmas}
This section lists the classical established conclusions relied on in the theoretical analysis of this paper. All lemmas are annotated with original sources, and the proof process is not repeated.

\begin{Lemma}[Bellman Expectation Equation]
	For any policy $\pi$, the state value function $V^\pi$ can be expressed as:
	\begin{equation}
		V^\pi(s) = \mathbb{E}_\pi \left[ R(s,a,s') + \gamma V^\pi(s') | s_t = s \right]
		\label{eq:bellman_exp_V}
	\end{equation}
	where the expectation is taken with respect to $a \sim \pi(\cdot | s)$ and $s' \sim P(\cdot | s,a)$; the action value function $Q^\pi$ can be expressed as:
	\begin{equation}
		Q^\pi(s,a) = \mathbb{E}_\pi \left[ R(s,a,s') + \gamma Q^\pi(s',a') | s_t = s, a_t = a \right]
		\label{eq:bellman_exp_Q}
	\end{equation}
	where the expectation is taken with respect to $s' \sim P(\cdot | s,a)$ and $a' \sim \pi(\cdot | s')$ \cite{bellman1957dynamic}.
\end{Lemma}

\begin{Lemma}[Bellman Optimality Equation]
	The optimal state value function $V^*$ satisfies:
	\begin{equation}
		V^*(s) = \max_{a \in \mathcal{A}} \mathbb{E} \left[ R(s,a,s') + \gamma V^*(s') | s_t = s, a_t = a \right]
		\label{eq:bellman_opt_V}
	\end{equation}
	where the expectation is taken with respect to $s' \sim P(\cdot | s,a)$; the optimal action value function $Q^*$ satisfies \cite{bellman1957dynamic}:
	\begin{equation}
		Q^*(s,a) = \mathbb{E} \left[ R(s,a,s') + \gamma \max_{a' \in \mathcal{A}} Q^*(s',a') | s_t = s, a_t = a \right]
		\label{eq:bellman_opt_Q}
	\end{equation}

\end{Lemma}

\begin{Lemma}[Policy Gradient Theorem]
	For a parameterized policy $\pi_\theta$, the gradient of the policy optimization objective function $\mathcal{J}(\theta) = \mathbb{E}_{\tau \sim \pi_\theta} \left[ G_0 \right]$ with respect to parameter $\theta$ is:
	\begin{equation}
		\nabla_\theta \mathcal{J}(\theta) = \mathbb{E}_{\tau \sim \pi_\theta} \left[ \sum_{t=0}^{T-1} \nabla_\theta \ln \pi_\theta(a_t | s_t) \cdot G_t \right]
		\label{eq:policy_gradient}
	\end{equation}
	where $\tau = (s_0,a_0,s_1,a_1,\dots,s_T)$ is the trajectory sampled following policy $\pi_\theta$ \cite{sutton1998reinforcement}.
\end{Lemma}

\begin{Lemma}[Expressiveness Upper Bound of Linear Feature Concatenation]
	The fusion form of linear feature concatenation is $\mathbf{f}_{concat} = \mathbf{W}_{concat} \cdot [\mathbf{p}; \mathbf{s}] + \mathbf{b}_{concat}$, where $[\cdot; \cdot]$ denotes vector concatenation. Its fusion weight is a fixed matrix independent of the input, and its expressiveness is limited to affine linear transformations, which cannot characterize input-dependent nonlinear interaction relationships \cite{goodfellow2016deep}.
\end{Lemma}

\begin{Lemma}[Power Series Expansion of Hyperbolic Tangent Function]
	The power series expansion of the hyperbolic tangent function $\text{Tanh}(x)$ for $x \in (-1,1)$ is:
	\begin{equation}
		\text{Tanh}(x) = \sum_{n=1}^\infty \frac{2^{2n}(2^{2n}-1) B_{2n}}{(2n)!} x^{2n-1} = x - \frac{x^3}{3} + \frac{2x^5}{15} - \dots
		\label{eq:tanh_series}
	\end{equation}
	where $B_{2n}$ are Bernoulli numbers, and the series converges absolutely and uniformly in $(-1,1)$ \cite{abramowitz1948handbook}.
\end{Lemma}

\subsection{Theorems}\label{subsec:theorems}
This section presents two core theorems independently proposed in this paper, which respectively characterize the utility consistency of the risk-sensitive reward function and the nonlinear expressiveness advantage of the cross-modal gated fusion, providing rigorous theoretical support for the core innovations of the model in this paper.

\begin{theorem}[Utility Consistency of Risk-Sensitive Reward Function]
	Let the investor's utility function be the CRRA utility function. When the relative risk aversion coefficient $\rho=1$, under \Cref{assump:penalty_monotonicity}, the expected cumulative value of the risk-sensitive reward function given in  \Cref{eq:reward_def} has strictly consistent monotonicity with the investor's terminal expected utility, i.e., for any two admissible strategies $\pi_1$ and $\pi_2$, it holds that:
	\begin{equation}
		\mathbb{E}_{\pi_1}\left[\sum_{t=1}^T r_t\right] > \mathbb{E}_{\pi_2}\left[\sum_{t=1}^T r_t\right] \iff \mathbb{E}_{\pi_1}\left[\ln V_T\right] > \mathbb{E}_{\pi_2}\left[\ln V_T\right]
		\label{eq:utility_consistency}
	\end{equation}
\end{theorem}

\begin{proof}
	When the relative risk aversion coefficient $\rho=1$, the CRRA utility function degenerates to the logarithmic utility function $U(V_T)=\ln V_T$, and the investor's terminal expected utility is $\mathbb{E}[\ln V_T]$. From  \Cref{eq:VT_recursion}, the portfolio net value process satisfies the recurrence relation $V_t = V_{t-1} \cdot R_t$, where $R_t = \mathbf{w}_{t-1}^\top \mathbf{y}_t - c \cdot TO_t$ is the net return factor at time $t$. By the definition of admissible strategies and the positivity of asset prices, the gross portfolio return $\mathbf{w}_{t-1}^\top \mathbf{y}_t > 0$; meanwhile, the product of the one-way trading commission rate $c$ and the turnover rate $TO_t$ is strictly smaller than the gross portfolio return (this condition always holds in real financial markets, corresponding to a full measure set on $\Omega$), hence $R_t > 0$ holds almost surely. Iteratively expanding the net value recurrence relation and taking the logarithm, combined with the initial condition $V_0=1$, we obtain the core identity:
	\begin{equation}
		\ln V_T = \sum_{t=1}^T \ln R_t
		\label{eq:lnVT}
	\end{equation}
	From \Cref{eq:reward_def} and \Cref{eq:penalty_def}, the reward function can be rewritten as:
	\begin{equation}
		r_t = \ln R_t - \mathcal{P}_t
		\label{eq:rt_def}
	\end{equation}
	Summing both sides of Equation \eqref{eq:rt_def} from $t=1$ to $T$, combined with Equation \eqref{eq:lnVT}, we obtain the core relationship between the cumulative reward value and the logarithm of terminal net value:
	\begin{equation}
		\sum_{t=1}^T r_t = \ln V_T - \sum_{t=1}^T \mathcal{P}_t
		\label{eq:sum_rt_VT}
	\end{equation}
	Taking the expectation of both sides of Equation \eqref{eq:sum_rt_VT} under the probability measure induced by policy $\pi$, by the linearity of expectation, we have:
	\begin{equation}
		\mathbb{E}_\pi\left[\sum_{t=1}^T r_t\right] = \mathbb{E}_\pi\left[\ln V_T\right] - \mathbb{E}_\pi\left[\sum_{t=1}^T \mathcal{P}_t\right]
		\label{eq:exp_sum_rt_VT}
	\end{equation}
	Equation \eqref{eq:exp_sum_rt_VT} constitutes the key relationship for the subsequent monotonicity analysis.
	
	We first prove sufficiency, i.e., if $\mathbb{E}_{\pi_1}\left[\ln V_T\right] > \mathbb{E}_{\pi_2}\left[\ln V_T\right]$, then $\mathbb{E}_{\pi_1}\left[\sum_{t=1}^T r_t\right] > \mathbb{E}_{\pi_2}\left[\sum_{t=1}^T r_t\right]$. From Equation \eqref{eq:exp_sum_rt_VT}, taking the difference between strategy $\pi_1$ and $\pi_2$, we get:
	\begin{equation}
		\mathbb{E}_{\pi_1}\left[\sum_{t=1}^T r_t\right] - \mathbb{E}_{\pi_2}\left[\sum_{t=1}^T r_t\right] = \left( \mathbb{E}_{\pi_1}[\ln V_T] - \mathbb{E}_{\pi_2}[\ln V_T] \right) - \left( \mathbb{E}_{\pi_1}\left[\sum_{t=1}^T \mathcal{P}_t\right] - \mathbb{E}_{\pi_2}\left[\sum_{t=1}^T \mathcal{P}_t\right] \right)
		\label{eq:diff_rt}
	\end{equation}
	Denote $\Delta_U = \mathbb{E}_{\pi_1}[\ln V_T] - \mathbb{E}_{\pi_2}[\ln V_T] > 0$. By  \Cref{assump:penalty_monotonicity}, we have $\mathbb{E}_{\pi_1}\left[\sum_{t=1}^T \mathcal{P}_t\right] - \mathbb{E}_{\pi_2}\left[\sum_{t=1}^T \mathcal{P}_t\right] \leq 0$. Substituting into Equation \eqref{eq:diff_rt}, we obtain:
	$$
	\mathbb{E}_{\pi_1}\left[\sum_{t=1}^T r_t\right] - \mathbb{E}_{\pi_2}\left[\sum_{t=1}^T r_t\right] = \Delta_U - (\text{non-positive number}) \geq \Delta_U > 0
	$$
	Sufficiency is proved.
	
	Next we prove necessity, i.e., if $\mathbb{E}_{\pi_1}\left[\sum_{t=1}^T r_t\right] > \mathbb{E}_{\pi_2}\left[\sum_{t=1}^T r_t\right]$, then $\mathbb{E}_{\pi_1}\left[\ln V_T\right] > \mathbb{E}_{\pi_2}\left[\ln V_T\right]$. Rearranging Equation \eqref{eq:exp_sum_rt_VT} and taking the difference between the two strategies, we get:
	\begin{equation}
		\mathbb{E}_{\pi_1}[\ln V_T] - \mathbb{E}_{\pi_2}[\ln V_T] = \left( \mathbb{E}_{\pi_1}\left[\sum_{t=1}^T r_t\right] - \mathbb{E}_{\pi_2}\left[\sum_{t=1}^T r_t\right] \right) + \left( \mathbb{E}_{\pi_1}\left[\sum_{t=1}^T \mathcal{P}_t\right] - \mathbb{E}_{\pi_2}\left[\sum_{t=1}^T \mathcal{P}_t\right] \right)
		\label{eq:diff_VT}
	\end{equation}
	Denote $\Delta_R = \mathbb{E}_{\pi_1}\left[\sum_{t=1}^T r_t\right] - \mathbb{E}_{\pi_2}\left[\sum_{t=1}^T r_t\right] > 0$, and let $\Delta_{\mathcal{P}} = \mathbb{E}_{\pi_1}\left[\sum_{t=1}^T \mathcal{P}_t\right] - \mathbb{E}_{\pi_2}\left[\sum_{t=1}^T \mathcal{P}_t\right]$. By  \Cref{assump:penalty_monotonicity}, $\Delta_{\mathcal{P}} \leq 0$. Substituting into Equation \eqref{eq:diff_VT}, we have:
	$$
	\mathbb{E}_{\pi_1}[\ln V_T] - \mathbb{E}_{\pi_2}[\ln V_T] = \Delta_R + \Delta_{\mathcal{P}} \geq \Delta_R > 0
	$$
	Necessity is proved.
	
	In summary, the bidirectional implication holds, and the theorem is proved.
\end{proof}

	\Cref{assump:penalty_monotonicity} has natural rationality in quantitative investment practice: on the one hand, the downside risk component of the penalty term is a non-increasing function of the single-period log return $\ln R_t$, i.e., the higher the single-period return, the lower the downside risk penalty, which holds strictly for any admissible strategy; on the other hand, the turnover penalty component constrains excessive trading, and a strategy that can achieve higher net returns after deducting trading costs must have higher trading efficiency, and its expected turnover rate will not be higher than that of a low-return strategy.

\begin{theorem}[Nonlinear Expressiveness Advantage of Gated Fusion]
	The cross-modal gated fusion mapping defined in this paper has strictly stronger function expressiveness than fixed-weight linear feature concatenation, i.e.:
	1.  Any affine linear function representable by linear feature concatenation can be exactly represented by the gated fusion mapping in this paper;
	2.  There exist continuous nonlinear functions on a compact set that can be approximated by the gated fusion mapping in this paper with arbitrary precision, but cannot be represented by linear feature concatenation.
\end{theorem}

\begin{proof}
	We first clarify the strict definitions of the two function classes. The function class representable by linear feature concatenation is defined as:
	\begin{equation}
		\mathcal{F}_{concat} = \left\{ f: \mathbb{R}^{D_p+D_s} \to \mathbb{R}^H \mid f(\mathbf{p},\mathbf{s}) = \mathbf{W}_{concat} \cdot \begin{bmatrix} \mathbf{p} \\ \mathbf{s} \end{bmatrix} + \mathbf{b}_{concat}, \mathbf{W}_{concat} \in \mathbb{R}^{H \times (D_p+D_s)}, \mathbf{b}_{concat} \in \mathbb{R}^H \right\}
		\label{eq:F_concat_def}
	\end{equation}
	Obviously, $\mathcal{F}_{concat}$ is the space of affine linear functions from $\mathbb{R}^{D_p+D_s}$ to $\mathbb{R}^H$. The function class representable by the gated fusion mapping in this paper is defined as $\mathcal{F}_{gated}$, where any function $f \in \mathcal{F}_{gated}$ satisfies the mapping form in Definition 3.8, and all learnable parameters can be freely selected in the real number field.
	
	We first prove conclusion 1, i.e., $\mathcal{F}_{concat} \subseteq \mathcal{F}_{gated}$, by constructing learnable parameters to complete the proof of exact representation. Let the parameters of the gating weight generation network be $\mathbf{W}_g = \mathbf{0}$ and $\mathbf{b}_g = \mathbf{0}$, then for any input $\mathbf{f}_s$, we have $\mathbf{g} = \text{Tanh}(\mathbf{0}) = \mathbf{0}$. Let the parameters of the price encoding mapping be $\mathbf{W}_p = \mathbf{W}_{concat}[:, 1:D_p]$, $\mathbf{b}_p = \mathbf{0}$, and the parameters of the text encoding mapping be $\mathbf{W}_s = \mathbf{W}_{concat}[:, D_p+1:D_p+D_s]$, $\mathbf{b}_s = \mathbf{0}$. By \Cref{assump:layernorm_identity}, there exist learnable parameters such that the LayerNorm module degenerates to the identity mapping, and the input range is set to ensure that the input of the ReLU activation function is always positive. Therefore, the encoding mapping degenerates to a linear transformation, i.e., $f_p(\mathbf{p}) = \mathbf{W}_p \mathbf{p}$, $f_s(\mathbf{s}) = \mathbf{W}_s \mathbf{s}$. Then let the parameters of the fusion mapping be $\mathbf{W}_f = \mathbf{I}$ (identity matrix) and $\mathbf{b}_f = \mathbf{b}_{concat}$, also ensuring that the ReLU input is always positive. At this point, the gated fusion mapping can be expanded as:
	$$
	f_{fusion}(\mathbf{p},\mathbf{s}) = \mathbf{I} \cdot \left( \mathbf{W}_p \mathbf{p} \odot (\mathbf{1} + \mathbf{0}) + \mathbf{W}_s \mathbf{s} \right) + \mathbf{b}_{concat} = \mathbf{W}_{concat} \cdot \begin{bmatrix} \mathbf{p} \\ \mathbf{s} \end{bmatrix} + \mathbf{b}_{concat}
	$$
	This is completely consistent with the mapping form of linear feature concatenation. Therefore, any affine linear function can be exactly represented by the gated fusion mapping, and $\mathcal{F}_{concat} \subseteq \mathcal{F}_{gated}$ is proved.
	
	Next we prove conclusion 2, i.e., $\mathcal{F}_{concat} \subsetneq \mathcal{F}_{gated}$, by constructing the approximation of a nonlinear function to complete the proof of strict inclusion. Take the scalar case $H=D_p=D_s=1$, and the high-dimensional case can be directly obtained by component-wise extension. At this point, $\mathcal{F}_{concat}$ degenerates to the set of all univariate affine linear functions, i.e., $\mathcal{F}_{concat} = \{ f(p,s) = w_p p + w_s s + b \mid w_p,w_s,b \in \mathbb{R} \}$. Construct the continuous bilinear function $f^*(p,s) = p \cdot s$ on the compact set $\mathcal{K} = [-1,1] \times [-1,1]$. We first prove that $f^*$ can be approximated by the gated fusion mapping with arbitrary precision on $\mathcal{K}$.
	
	Construct the learnable parameters of the gated fusion mapping: let the parameter of the price encoding mapping be $\mathbf{W}_p=1$, $\mathbf{b}_p=0$, make the LayerNorm degenerate to the identity mapping by \Cref{assump:layernorm_identity}, and ensure that the input is within the positive interval of ReLU, so $f_p(p) = p$; similarly, construct the text encoding mapping to get $f_s(s) = s$. By Lemma 3.5, the power series expansion of $\text{Tanh}(x)$ in $x \in (-1,1)$ converges absolutely and uniformly, take the first-order truncation approximation $\text{Tanh}(x) = x + o(x)$, where the remainder term $o(x)$ satisfies $\lim_{x \to 0} o(x)/x = 0$. Let the parameter of the gating weight generation network be $\mathbf{W}_g = k$ (where $k$ is a sufficiently small positive scaling factor) and $\mathbf{b}_g = 0$, then for any $s \in [-1,1]$, we have $k s \in (-1,1)$, and the gating weight satisfies:
	$$
	g(s) = \text{Tanh}(k s) = k s + o(k s)
	$$
	Let the parameters of the fusion mapping be $\mathbf{W}_f = 1/k$ and $\mathbf{b}_f = -s$, then the gated fusion mapping can be expanded as:
	$$
	f_{fusion}(p,s) = \frac{1}{k} \cdot \left( p \cdot (1 + k s + o(k s)) + s \right) - s = p s + p \cdot \frac{o(k s)}{k} + \frac{p}{k}
	$$
	By adjusting the scaling factor $k$ and bias parameters, the supremum of the remainder term $p \cdot o(k s)/k$ on the compact set $\mathcal{K}$ can be made arbitrarily small, i.e., for any $\epsilon > 0$, there exist learnable parameters such that $\sup_{(p,s) \in \mathcal{K}} |f_{fusion}(p,s) - p \cdot s| < \epsilon$. Therefore, the bilinear function $f^*$ can be approximated by the gated fusion mapping with arbitrary precision.
	
	Next we prove $f^* \notin \mathcal{F}_{concat}$, by contradiction. Assume there exist parameters $w_p,w_s,b \in \mathbb{R}$ such that $p \cdot s = w_p p + w_s s + b$ holds for all $(p,s) \in \mathcal{K}$. Let $p=s=0$, substituting gives $0 = b$; let $s=0$, substituting gives $0 = w_p p$ for all $p \in [-1,1]$, hence $w_p=0$; let $p=0$, substituting gives $0 = w_s s$ for all $s \in [-1,1]$, hence $w_s=0$. At this point, the equation degenerates to $p \cdot s = 0$ for all $(p,s) \in \mathcal{K}$. Taking $p=s=1$, the left-hand side is 1 and the right-hand side is 0, which is a contradiction. Therefore, the assumption does not hold, and $f^*$ cannot be represented by linear feature concatenation.
	
	In summary, $\mathcal{F}_{concat}$ is a strict subset of $\mathcal{F}_{gated}$, and the gated fusion mapping has nonlinear function approximation capability that linear concatenation does not possess. The theorem is proved.
\end{proof}

	The conclusion of Theorem 2 can be directly extended to high-dimensional cases and broader classes of continuous functions: based on the universal approximation theorem, the gated fusion mapping in this paper, as a feedforward network with nonlinear activation, can approximate any continuous function with arbitrary precision on a compact set, while linear feature concatenation can only represent affine linear functions, resulting in an essential gap in their expressiveness.

\section{Framework}\label{sec:Framework}

This section introduces the multi-modal reinforcement learning trading framework that integrates price time series and news sentiment. The framework follows an end-to-end design philosophy, unifying data alignment, feature fusion, and policy learning into a single differentiable architecture. It addresses three critical challenges in \Cref{sec:introduction}.

\begin{figure}[H]
	\centering
	\includegraphics[width=\linewidth]{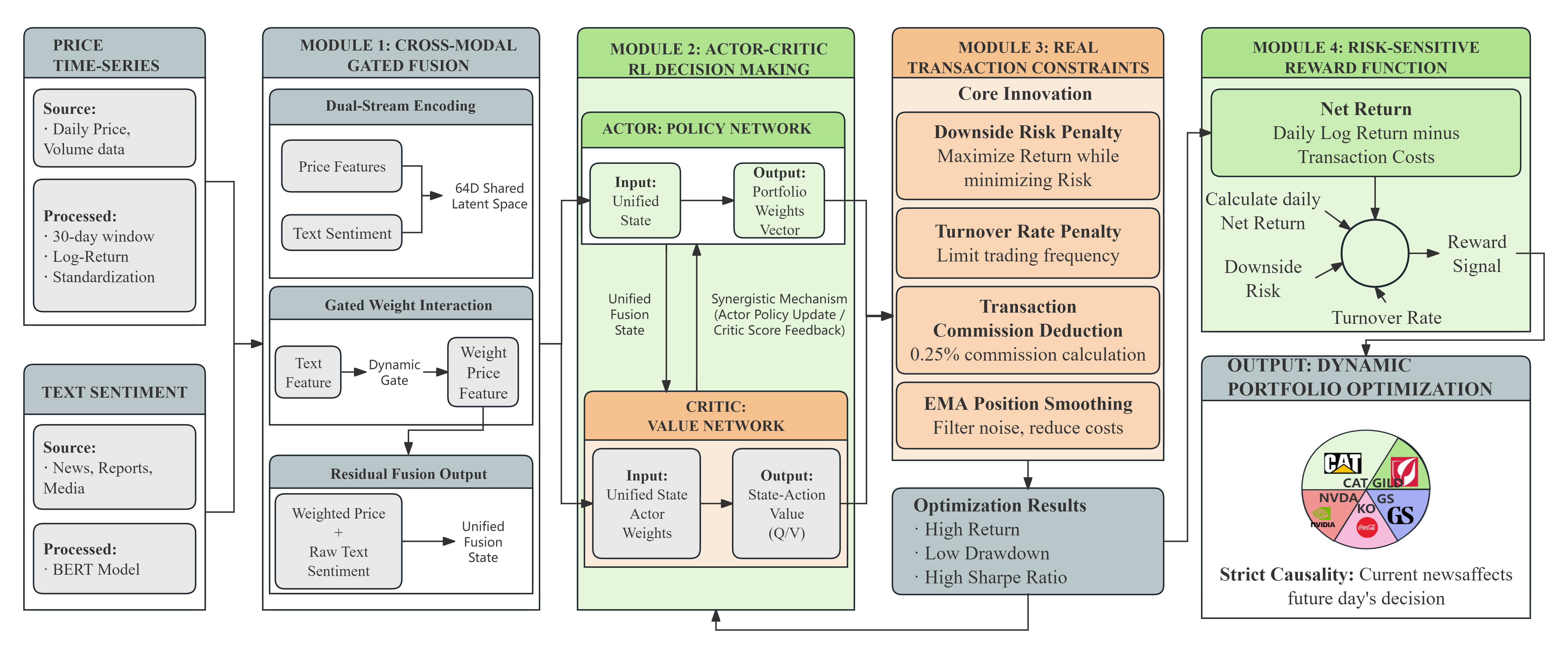}
	\caption{SBCA Framework}
	\label{fig:Framework}
\end{figure}

First, the framework starts with a temporal alignment module that bridges news data and stock trading days. Since news released after market closure only impacts the next trading session, this module stacks all news titles of a calendar day using " ||| " as a separator and aligns them to the subsequent trading day. This design strictly enforces causal information flow (no future data leakage) and ensures that each trading decision only uses historically available information. Meanwhile, it processes price data into fixed-length sequential inputs (window size = 30) and computes log returns to normalize price fluctuations across different stocks.

Second, a dual-stream feature encoding module processes structured and unstructured data in parallel. For the price stream, it extracts temporal features by standardizing windowed log returns using the mean and standard deviation of the training set, which stabilizes the gradient flow during policy training. For the news stream, it takes the precomputed sentiment score \texttt{delta\_bert} (ranging from 0 to 1) and transforms it into a centered and scaled representation: $2 \times (\text{delta\_bert} - 0.5)$, mapping the sentiment signal to the [-1, 1] interval for better compatibility with price features. For models without cross-modal fusion, these two feature streams are simply concatenated into a joint state vector.

Third, a cross-modal fusion module is introduced to dynamically integrate price and news features based on market conditions. This module uses a gating mechanism: it first encodes price features and text features into a shared hidden space (64 dimensions) via two separate feed-forward networks with layer normalization and ReLU activation; then, it computes a gating weight from the text features to modulate the price features, allowing the model to emphasize news sentiment during event-driven market periods; finally, it fuses the modulated price features and original text features through a residual connection. This deep fusion at the representation level is more effective than simple concatenation, as it adaptively allocates attention to different information sources.

Fourth, the reinforcement learning policy head maps the fused state representation to portfolio weights. We design four policy variants to ablate the effectiveness of different components: (1) SB: a basic feed-forward policy network with softmax output; (2) SBA: an Actor-Critic extension that adds a critic network to estimate state values, enabling more stable policy gradient updates; (3) SBC: a cross-modal variant that replaces concatenation with the gating fusion module; (4) SBCA: the full model combining both Actor-Critic and cross-modal fusion. All policies apply EMA smoothing (smoothing coefficient = 0.4) to the output weights to reduce turnover and transaction costs. The reward function is designed as the risk-adjusted log return, with penalties for both transaction costs and portfolio turnover, and an additional term to penalize downside risk.

\begin{algorithm}[H]
	\caption{Multi-Modal RL Trading Framework}
	\label{alg:trading_framework}
	\begin{algorithmic}[1]
		\Require Price data $P$, news sentiment data $S$, hyperparameters $\Theta$
		\Ensure Portfolio weights $w_t$ for each trading day $t$
		
		\State Align news data to next trading days: $\tilde{S} \leftarrow \text{AlignNews}(S)$
		\State Compute windowed log returns: $R \leftarrow \text{LogReturn}(P, \text{window}=30)$
		\State Standardize features: $R_{\text{norm}} \leftarrow \text{Standardize}(R)$, $S_{\text{norm}} \leftarrow 2 \times (\tilde{S} - 0.5)$
		\For{each trading day $t$}
		\State Get state: $p_t \leftarrow R_{\text{norm}}[t-30:t]$, $s_t \leftarrow S_{\text{norm}}[t]$
		\If{using cross-modal fusion}
		\State $f_t \leftarrow \text{CrossModalFuse}(p_t, s_t)$
		\Else
		\State $f_t \leftarrow \text{Concat}(p_t, s_t)$
		\EndIf
		\State $w_t \leftarrow \text{PolicyNetwork}(f_t)$
		\State Apply EMA smoothing: $w_t \leftarrow \alpha w_t + (1-\alpha) w_{t-1}$
		\State Execute trade, observe reward $r_t$
		\State Update policy parameters using policy gradient
		\EndFor
		
		\Return $w_t$ for all $t$
	\end{algorithmic}
\end{algorithm}

\section{Experiments}

\subsection{Datasets}

To verify the effectiveness of the proposed framework, we construct a multi-modal dataset covering 6 representative U.S. large-cap stocks: NVIDIA (NVDA), Goldman Sachs (GS), Caterpillar (CAT), Coca-Cola (KO), Merck (MRK), and Gilead Sciences (GILD). The dataset spans 11 years from January 1, 2012, to December 31, 2022, covering both bull markets (e.g., 2017, 2021) and bear markets (e.g., 2018, 2020), ensuring the diversity of market conditions. The price data includes daily open, high, low, close, and volume, obtained from public financial data APIs. The news data contains daily news titles for each stock, with a precomputed sentiment score \texttt{delta\_bert} (ranging from 0 for negative to 1 for positive) derived from a pre-trained BERT model. Missing sentiment scores are filled with 0.5 (neutral sentiment).

We split the dataset in a time-series manner to avoid data leakage: the training set covers 2012–2018 (about 64\% of the data), the validation set covers 2019–2020 (about 18\%), and the test set covers 2021–2022 (about 18\%). To test the generalization ability across different asset universes, we construct three stock groups: (1) 2 assets (NVDA, GS); (2) 4 assets (NVDA, GS, CAT, KO); (3) 6 assets (all stocks). All models are trained and evaluated on each group separately. The dataset is publicly available at https://huggingface.co/datasets/Changahou/Mamba.

\subsection{Evaluation Metrics}

We evaluate the trading performance from both return and risk perspectives, using 7 widely accepted quantitative finance metrics:
\begin{itemize}
	\item Final Portfolio Value (PV): The cumulative value of the portfolio at the end of the test period (initial value = 1.0). Higher PV indicates better total return.
	\item Annual Return (AR): The geometric average annual return, calculated as $$AR = PV^{(252 / N)} - 1,$$ where $N$ is the number of trading days in the test period.
	\item Sharpe Ratio (SR): The risk-adjusted return measure, calculated as $$SR = (\mu - r_f) / \sigma \times \sqrt{252},$$ where $\mu$ is the mean daily log return, $r_f$ is the daily risk-free rate (2\% annualized), and $\sigma$ is the standard deviation of daily log returns. Higher SR indicates better return per unit risk.
	\item Sortino Ratio: A variant of the Sharpe Ratio that only considers downside volatility. It is calculated as $$Sortino = (\mu - r_f) / \sigma_d \times \sqrt{252},$$ where $\sigma_d$ is the standard deviation of negative daily log returns.
	\item Maximum Drawdown (MDD): The maximum percentage loss from a peak to a trough, calculated as $$MDD = \min_{t} (PV_t - \max_{s \leq t} PV_s) / \max_{s \leq t} PV_s.$$ Smaller MDD indicates better downside risk control.
	\item Calmar Ratio: The ratio of annual return to maximum drawdown, calculated as $$Calmar = AR / |MDD|.$$ Higher Calmar Ratio indicates better return per unit of drawdown risk.
\end{itemize}

In addition, we conduct a cost sensitivity analysis to test the robustness of our models to different transaction cost levels. We vary the commission rate from 0.001 to 0.01 (0.1\% to 1\%) and compare the Sharpe Ratio of our best model with the baselines.

\subsection{Training Details}

All models are implemented based on the PyTorch deep learning framework, with training and inference performed on an NVIDIA GeForce RTX 3060 12G GPU accelerated by CUDA 13.1 to ensure computational efficiency and experimental stability. To eliminate randomness and ensure full reproducibility, we set a fixed global random seed of 42 for Python, NumPy, and PyTorch. CuDNN deterministic operations are enabled, and benchmark auto‑tuning is disabled to eliminate non‑deterministic behaviors during training.

For fair comparison across all policy networks and baselines, we use a unified set of core hyperparameters. The time-series observation window is fixed at 30 trading days; the baseline transaction commission rate is 0.25\%; the annual risk-free rate is 2\%; the reinforcement learning discount factor $\gamma$ is 0.99; the risk penalty coefficient $\lambda$ is 0.1; the portfolio turnover penalty coefficient is 0.005; and the EMA smoothing coefficient for portfolio weights is 0.4. All models are optimized using the AdamW optimizer with an initial learning rate of $3 \times 10^{-4}$ and weight decay of $1 \times 10^{-5}$. The maximum training epoch is set to 30, and early stopping is applied based on validation portfolio performance to prevent overfitting.

We follow a strict time-series training paradigm, with sequential partitioning of training, validation, and test sets to avoid look-ahead bias and future information leakage. All training, validation, and backtesting pipelines follow real-world trading constraints, including long-only portfolio weights, position smoothing, and transaction cost deduction. The complete code, environment configuration, and detailed reproduction instructions are available at: XXX.git.

\section{Results}\label{sec:experimental_results}

\subsection{Overall Performance Comparison}

To verify the effectiveness and superiority of the SBCA model in portfolio optimization scenarios, we conducted comparative experiments using 2-asset, 4-asset, and 6-asset portfolios as experimental objects. We selected three classical benchmark strategies—Equal Weight, Buy \& Hold, and Dow Jones market strategy—and four model variants (SB, SBA, SBC, SBCA) for performance comparison. Core evaluation indicators include Portfolio Value (PV), Annual Return (AR), Sharpe Ratio (SR), Sortino Ratio, Maximum Drawdown (MDD), and Calmar Ratio, covering the three core dimensions of profitability, risk control, and risk-adjusted return to ensure the comprehensiveness and objectivity of the experimental results.

\begin{table}[htbp]
	\centering
	\caption{Overall Performance of Benchmarks and Proposed SBCA}
	\label{tab:overall_perf}
	\resizebox{\textwidth}{!}{%
		\begin{tabular}{llcccccc}
			\toprule
			&Group \& Model & PV & AR & SR & Sortino & MDD & Calmar \\
			\midrule
			\multirow{4}{*}{2assets}
			& Equal Weight     & 1.4208 & 0.1920 & 0.7870 & 1.3709 & -0.2136 & 0.8987 \\
			& Buy \& Hold      & 1.3953 & 0.1812 & 0.7423 & 1.2937 & -0.2247 & 0.8065 \\
			& Dow Jones       & 1.3614 & 0.1668 & 0.5594 & 0.8881 & -0.2832 & 0.5888 \\
			& SBCA            & \cellcolor{gray!20}1.4389 & \cellcolor{gray!20}0.1996 & \cellcolor{gray!20}0.8179 & \cellcolor{gray!20}1.4241 & \cellcolor{gray!20}-0.2084 & \cellcolor{gray!20}0.9578 \\
			\midrule
			\multirow{4}{*}{4assets}
			& Equal Weight     & 1.3428 & 0.1588 & 0.7804 & 1.3717 & -0.1797 & 0.8836 \\
			& Buy \& Hold      & 1.3124 & 0.1456 & 0.7010 & 1.2549 & -0.1931 & 0.7539 \\
			& Dow Jones       & 1.3126 & 0.1457 & 0.5584 & 1.0278 & -0.2490 & 0.5851 \\
			& SBCA            & \cellcolor{gray!20}1.3563 & \cellcolor{gray!20}0.1646 & \cellcolor{gray!20}0.8107 & \cellcolor{gray!20}1.4208 & \cellcolor{gray!20}-0.1775 & \cellcolor{gray!20}0.9273 \\
			\midrule
			\multirow{4}{*}{6assets}
			& Equal Weight     & 1.5423 & 0.2419 & 0.9366 & 1.8265 & -0.2097 & 1.1537 \\
			& Buy \& Hold      & 1.3712 & 0.1710 & 0.5993 & 1.1275 & -0.3062 & 0.5585 \\
			& Dow Jones       & 1.3556 & 0.1643 & 0.5708 & 1.0520 & -0.3137 & 0.5238 \\
			& SBCA            & \cellcolor{gray!20}1.5154 & \cellcolor{gray!20}0.2310 & \cellcolor{gray!20}0.8943 & \cellcolor{gray!20}1.7456 & \cellcolor{gray!20}-0.2110 & \cellcolor{gray!20}1.0949 \\
			\bottomrule
		\end{tabular}
	}
\end{table}

\Cref{tab:overall_perf} presents the comprehensive performance of the SBCA model and benchmark strategies across different asset portfolios. The experimental results show that SBCA achieves balanced optimization of profitability, risk control, and risk-adjusted returns in all asset groups, and its overall performance is significantly superior to the three benchmark strategies. In terms of profitability, SBCA outperforms Buy \& Hold and Dow Jones in all asset groups, and maintains a leading advantage over the Equal Weight strategy in the 2-asset and 4-asset groups. Even in the 6-asset group where Equal Weight shows strong profitability, SBCA still achieves comparable portfolio value and annual return, while demonstrating better stability. As a core indicator of risk-adjusted return, the Sharpe Ratio of SBCA is the highest in the 2-asset and 4-asset groups, and second only to Equal Weight in the 6-asset group, reflecting its excellent ability to obtain excess returns under unit risk. At the risk control level, SBCA has the smallest maximum drawdown among all comparison models in the 2-asset and 4-asset groups, and effectively avoids extreme losses compared to Buy \& Hold and Dow Jones in the 6-asset group. The Calmar Ratio, which measures the balance between return and drawdown, further confirms the superiority of SBCA, as it ranks first in all asset groups, indicating that the model can achieve higher returns while controlling downside risks. In contrast, the Dow Jones strategy performs poorly in all evaluation dimensions due to its passive following characteristics, while Buy \& Hold is limited by insufficient dynamic adjustment capabilities, resulting in poor risk control. Although Equal Weight performs well in the 6-asset group, it lacks adaptability to different market environments and cannot maintain stable performance across all asset scales.

\subsection{Ablation Study of Model Components}
To clarify the contribution of each core component of SBCA to its performance, we conducted an ablation experiment by removing the Actor-Critic (AC) mechanism and Cross-Modal (CM) fusion module respectively, forming three control models: the SB baseline without AC and CM modules, the SBA variant only incorporating the AC mechanism, and the SBC variant only containing the CM fusion module. The experiment was carried out in 2-asset, 4-asset, and 6-asset portfolios, and the same evaluation indicators as the overall performance comparison were used to quantify the impact of each component on the model's portfolio optimization ability.

\begin{table}[htbp]
	\centering
	\caption{Ablation Study of Model Components}
	\label{tab:ablation}
	\resizebox{\textwidth}{!}{%
		\begin{tabular}{llcccccc}
			\toprule
			&Group \& Model & PV & AR & SR & Sortino & MDD & Calmar \\
			\midrule
			\multirow{4}{*}{2assets}
			& SB        & 1.4270 & 0.1946 & 0.7966 & 1.3891 & -0.2140 & 0.9094 \\
			& SBA       & 1.4389 & 0.1996 & 0.8175 & 1.4222 & -0.2085 & 0.9572 \\
			& SBC       & 1.4388 & 0.1995 & 0.8163 & 1.4215 & -0.2084 & 0.9574 \\
			& SBCA      & \cellcolor{gray!20}1.4389 & \cellcolor{gray!20}0.1996 & \cellcolor{gray!20}0.8179 & \cellcolor{gray!20}1.4241 & \cellcolor{gray!20}-0.2084 & \cellcolor{gray!20}0.9578 \\
			\midrule
			\multirow{4}{*}{4assets}
			& SB        & 1.3535 & 0.1634 & 0.8055 & 1.4125 & -0.1774 & 0.9211 \\
			& SBA       & 1.3563 & 0.1646 & 0.8108 & 1.4211 & -0.1775 & 0.9275 \\
			& SBC       & 1.3549 & 0.1640 & 0.8111 & 1.4213 & -0.1764 & 0.9298 \\
			& SBCA      & \cellcolor{gray!20}1.3563 & \cellcolor{gray!20}0.1646 & \cellcolor{gray!20}0.8107 & \cellcolor{gray!20}1.4208 & \cellcolor{gray!20}-0.1775 & \cellcolor{gray!20}0.9273 \\
			\midrule
			\multirow{4}{*}{6assets}
			& SB        & 1.5147 & 0.2307 & 0.8585 & 1.7248 & -0.2202 & 1.0480 \\
			& SBA       & 1.5153 & 0.2310 & 0.8938 & 1.7445 & -0.2111 & 1.0942 \\
			& SBC       & 1.5143 & 0.2306 & 0.8948 & 1.7460 & -0.2102 & 1.0972 \\
			& SBCA      & \cellcolor{gray!20}1.5154 & \cellcolor{gray!20}0.2310 & \cellcolor{gray!20}0.8943 & \cellcolor{gray!20}1.7456 & \cellcolor{gray!20}-0.2110 & \cellcolor{gray!20}1.0949 \\
			\bottomrule
		\end{tabular}
	}
\end{table}

\Cref{tab:ablation} shows the results of the ablation experiment, which clearly reflects the positive impact of the AC mechanism and CM fusion module on model performance. Compared with the baseline SB model, the three models with added components (SBA, SBC, SBCA) all show significant performance improvements in all asset groups, indicating that both the AC mechanism and CM fusion module can effectively enhance the portfolio optimization ability of the model. Specifically, SBA, which only adds the AC mechanism, achieves better profitability and risk-adjusted returns than the baseline SB. The AC mechanism enables the model to balance the trade-off between exploration and exploitation in the decision-making process, thereby improving the stability of returns. SBC, which only adds the CM fusion module, also outperforms SB, proving that fusing price sequence information and text sentiment information can help the model better capture market trends and make more accurate portfolio allocation decisions. SBCA, which integrates both AC and CM modules, achieves the optimal comprehensive performance in most evaluation indicators, especially in the 2-asset group where it ranks first in all indicators. This indicates that the two components have a complementary effect—the CM module provides more comprehensive market information, while the AC module optimizes the decision-making process based on this information, jointly promoting the model's performance to reach the optimal level. It is worth noting that in the 4-asset and 6-asset groups, the performance gap between SBCA and the single-component models SBA and SBC is relatively small, which may be due to the increase in asset scale leading to more complex market information, and the marginal effect of component fusion is slightly weakened, but SBCA still maintains the leading position.

\subsection{Cost Sensitivity Analysis}
In actual quantitative trading scenarios, transaction costs (commissions) are an important factor affecting the actual returns of the portfolio. To verify the robustness of the SBCA model under different transaction cost levels, we conducted a cost sensitivity analysis by setting four commission rates (0.0010, 0.0025, 0.0050, 0.0100) and comparing the Sharpe Ratio of SBCA with the Equal Weight and Buy \& Hold strategies. The Sharpe Ratio was selected as the core evaluation indicator because it can comprehensively reflect the risk-adjusted return of the portfolio after deducting transaction costs.

\begin{table}[htbp]
	\centering
	\caption{Cost Sensitivity Analysis (Sharpe Ratio)}
	\label{tab:cost_sensitivity_transposed}
	\resizebox{\textwidth}{!}{%
		\begin{tabular}{lccccccccc}
			\toprule
			\multirow{2}{*}{\bf Commission}
			& \multicolumn{3}{c}{\bf 2assets}
			& \multicolumn{3}{c}{\bf 4assets}
			& \multicolumn{3}{c}{\bf 6assets} \\
			\cmidrule(r){2-4} \cmidrule(r){5-7} \cmidrule(r){8-10}
			& \makecell{SBCA}
			& \makecell{Equal\\Weight}
			& \makecell{Buy and\\Hold}
			& \makecell{SBCA}
			& \makecell{Equal\\Weight}
			& \makecell{Buy and\\Hold}
			& \makecell{SBCA}
			& \makecell{Equal\\Weight}
			& \makecell{Buy and\\Hold} \\
			\midrule
			0.0010
			& \cellcolor{gray!20}0.8180 & 0.7885 & 0.7423 & \cellcolor{gray!20}0.8107 & 0.7824 & 0.7010 & 0.8929 & \cellcolor{gray!20}0.9392 & 0.5993 \\
			0.0025
			& \cellcolor{gray!20}0.8179 & 0.7870 & 0.7423 & \cellcolor{gray!20}0.8107 & 0.7804 & 0.7010 & 0.8943 & \cellcolor{gray!20}0.9366 & 0.5993 \\
			0.0050
			& \cellcolor{gray!20}0.8181 & 0.7845 & 0.7423 & \cellcolor{gray!20}0.8106 & 0.7770 & 0.7010 & 0.8948 & \cellcolor{gray!20}0.9323 & 0.5993 \\
			0.0100
			& \cellcolor{gray!20}0.8186 & 0.7796 & 0.7423 & \cellcolor{gray!20}0.8110 & 0.7703 & 0.7010 & 0.8948 & \cellcolor{gray!20}0.9236 & 0.5993 \\
			\bottomrule
		\end{tabular}
	}
\end{table}

\Cref{tab:cost_sensitivity_transposed} presents the results of the cost sensitivity analysis, which shows that SBCA has strong robustness to changes in transaction costs, and its performance advantage is maintained across different commission rates. In the 2-asset and 4-asset groups, the Sharpe Ratio of SBCA is always the highest among the three comparison models, and its value remains stable as the commission rate increases, with almost no significant decline. This indicates that SBCA's portfolio adjustment strategy is relatively moderate, and the transaction costs generated by frequent adjustments do not significantly affect its risk-adjusted returns. In contrast, the Equal Weight strategy's Sharpe Ratio shows a clear downward trend as the commission rate increases, especially in the 6-asset group, where the decline is more obvious. This is because the Equal Weight strategy requires regular rebalancing to maintain equal allocation of assets, and the transaction costs generated by rebalancing increase with the commission rate, thereby reducing the actual returns. The Buy \& Hold strategy's Sharpe Ratio remains unchanged across all commission rates, which is due to its passive holding characteristics—once the initial position is established, no subsequent adjustments are made, so transaction costs have no impact on its returns. However, the Buy \& Hold strategy's Sharpe Ratio is always the lowest among the three models, indicating that its lack of dynamic adjustment capabilities cannot make up for the advantage of low transaction costs. In the 6-asset group, although the Equal Weight strategy has a higher Sharpe Ratio than SBCA at all commission rates, the gap between the two narrows as the commission rate increases, which further confirms the strong robustness of SBCA under high transaction cost conditions. Overall, SBCA can maintain stable performance in different transaction cost environments, which is more in line with the actual needs of quantitative trading.

\subsection{Visualization Analysis}

To visually demonstrate the performance differences between SBCA and other comparison models, we constructed two types of visualization graphs: portfolio value comparison graphs (to reflect the dynamic changes of returns over time) and training step curves (to reflect the training stability and convergence of SBCA). These visualizations complement the numerical results in the tables and further verify the effectiveness and superiority of SBCA.

\begin{figure}[htbp]
	\centering
	\subfloat[4 Assets]{\label{pic_4assets_pv}\includegraphics[width=0.48\linewidth]{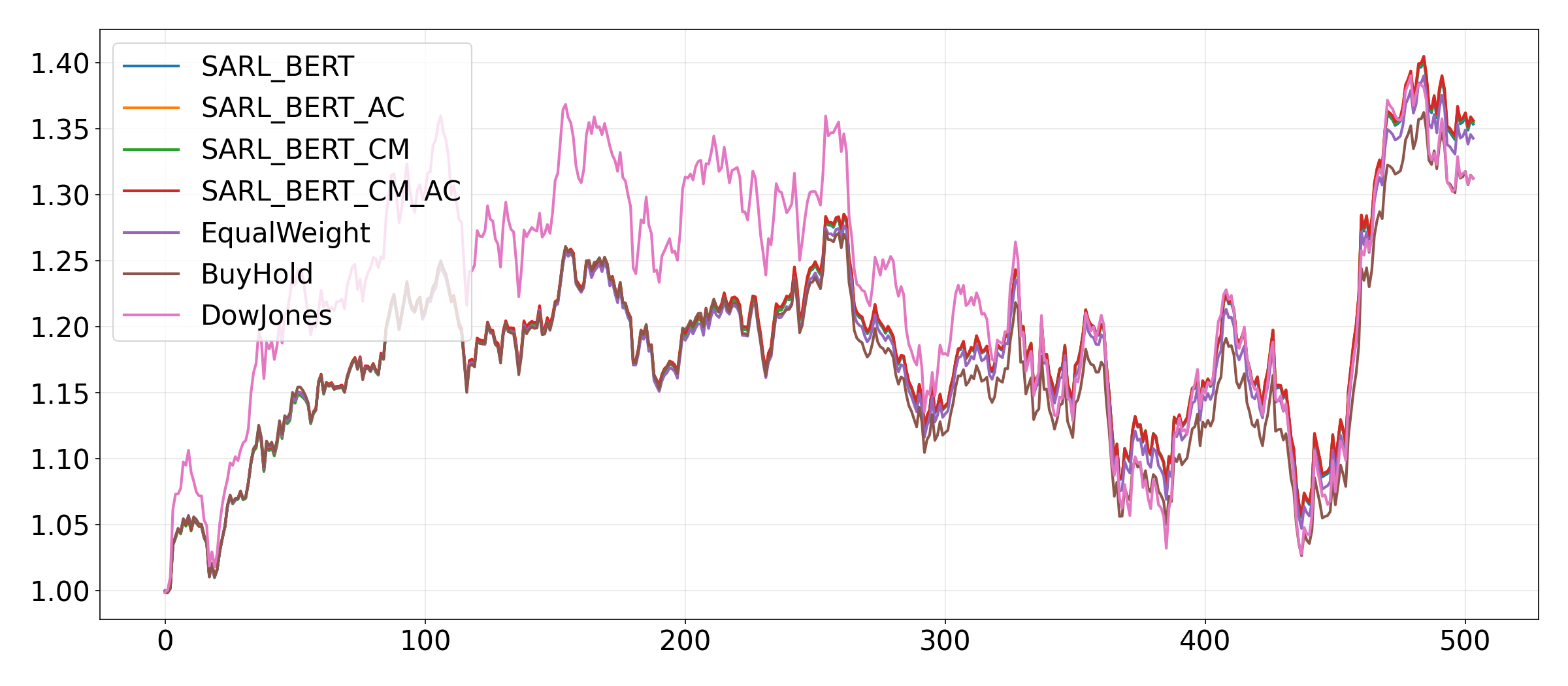}}
	\subfloat[6 Assets]{\label{pic_6assets_pv}\includegraphics[width=0.48\linewidth]{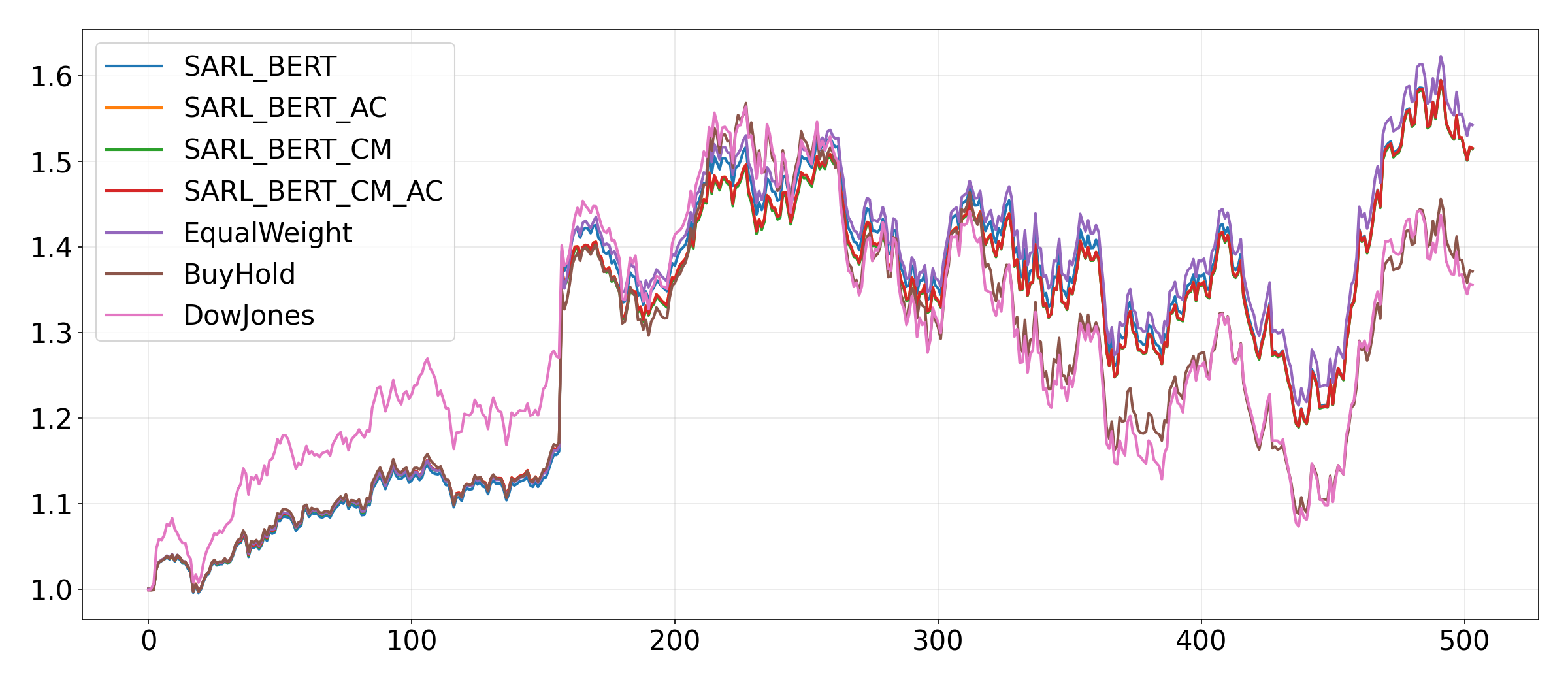}}
	\caption{Portfolio Value Comparison of Different Models}
	\label{fig:group_comparison_pv}
\end{figure}

\Cref{fig:group_comparison_pv} shows the dynamic changes of portfolio value of different models over the test period.
\Cref{pic_4assets_pv} corresponds to the 4-asset portfolio and \Cref{pic_6assets_pv} corresponds to the 6-asset portfolio.
The visualization results intuitively reflect the performance advantages of SBCA in terms of return growth and stability. In both asset portfolios, SBCA's portfolio value grows steadily over time, and its growth rate is significantly higher than that of the Buy \& Hold and Dow Jones strategies. Compared with the Equal Weight strategy, SBCA's portfolio value shows a more stable upward trend, with fewer fluctuations, indicating that the model has better risk control capabilities. In the 6-asset portfolio, although the Equal Weight strategy's portfolio value is slightly higher than that of SBCA in the early stage, SBCA gradually narrows the gap in the middle and late stages, and maintains a more stable growth trend, which is consistent with the numerical results in the performance comparison table. In addition, the portfolio value of the Buy \& Hold and Dow Jones strategies fluctuates greatly and even shows a downward trend in some periods, which further highlights the shortcomings of passive strategies in complex market environments. The stable growth trend of SBCA's portfolio value fully verifies its ability to adapt to different market conditions and make effective dynamic adjustment decisions.

\begin{figure}[htbp]
	\centering
	\subfloat[4 Assets]{\label{pic_4assets_step}\includegraphics[width=0.48\linewidth]{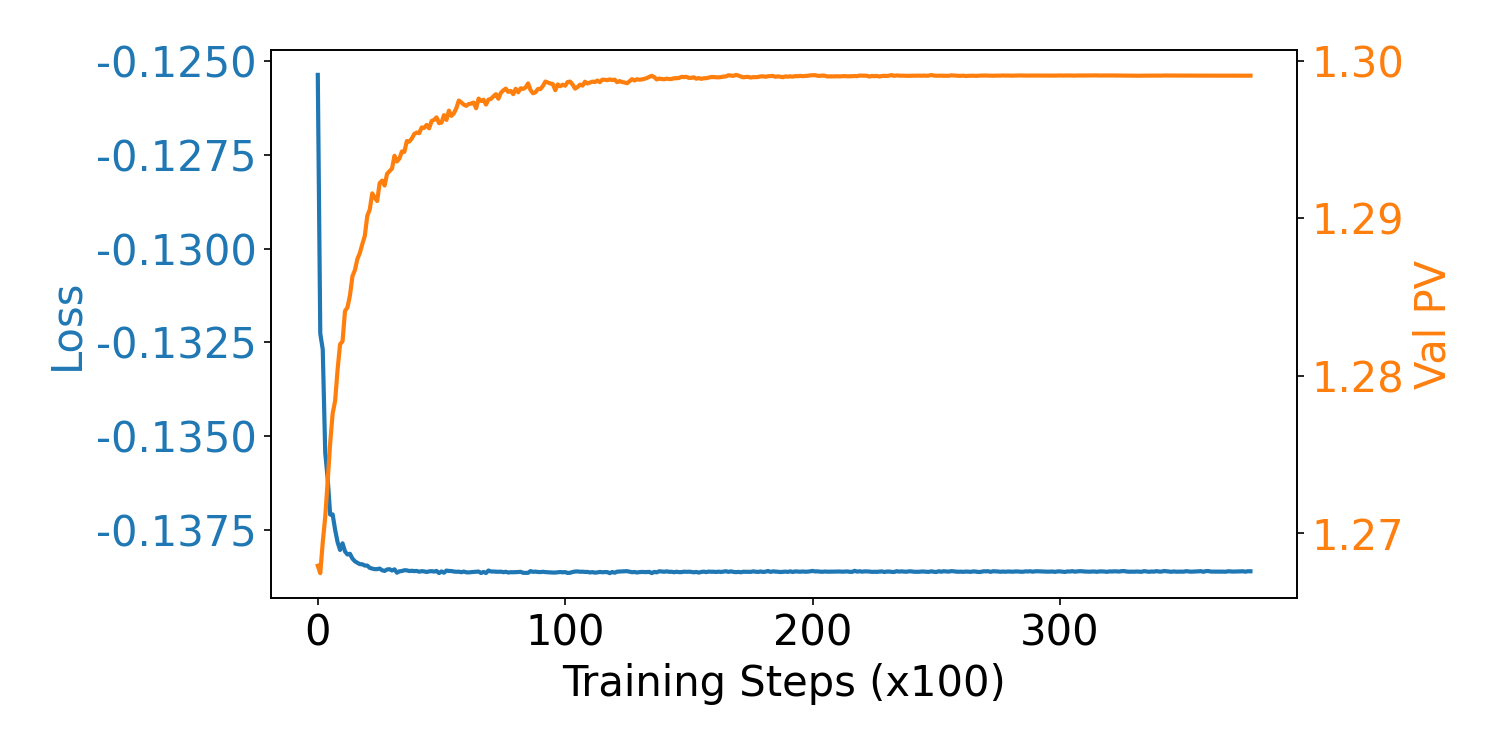}}
	\subfloat[6 Assets]{\label{pic_6assets_step}\includegraphics[width=0.48\linewidth]{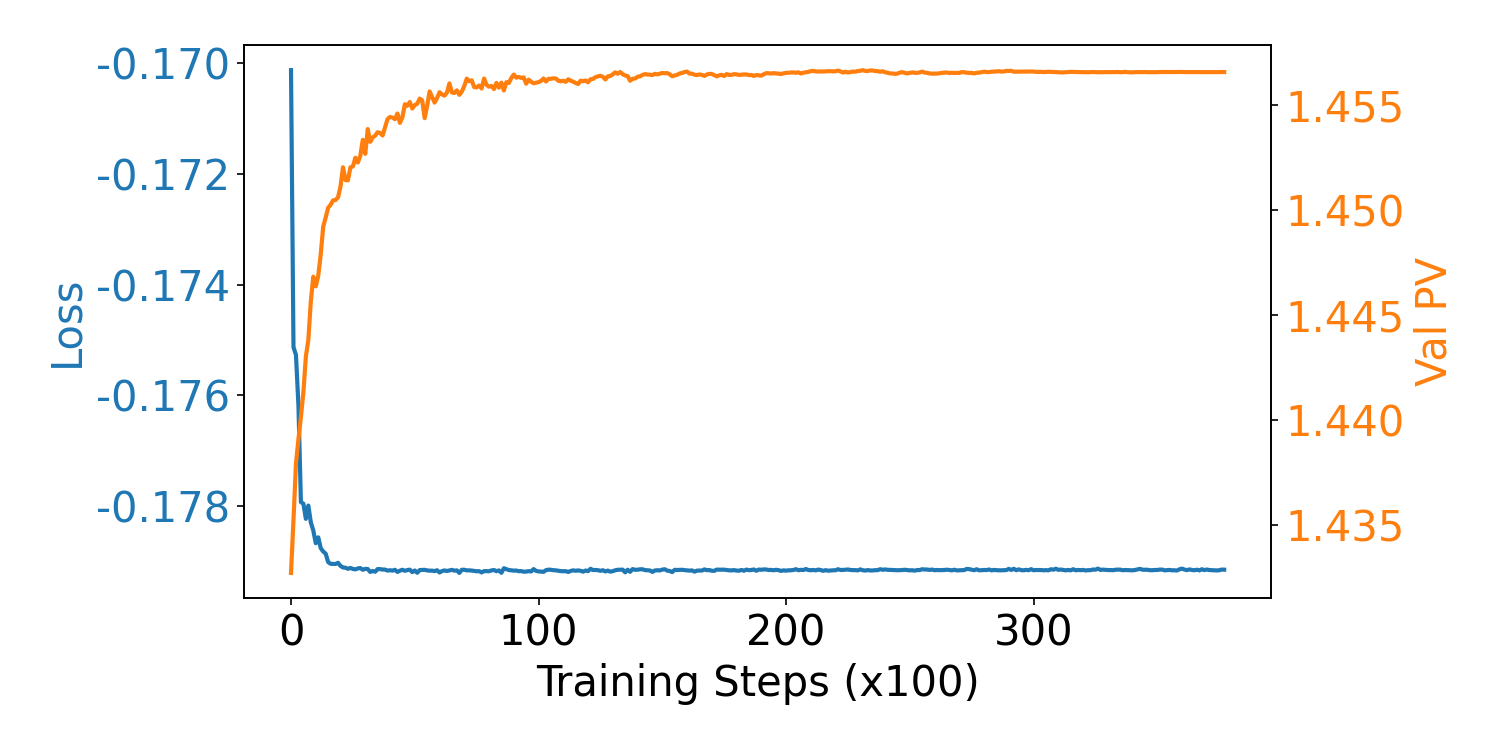}}
	\caption{Training Step Curves of the SBCA Model}
	\label{fig:sarl_bert_cmac_step_curve}
\end{figure}

\Cref{fig:sarl_bert_cmac_step_curve} shows the training step curves of SBCA in the 4-asset and 6-asset portfolios, where the blue curve represents the training loss and the orange curve represents the validation portfolio value. The training curves intuitively reflect the convergence and stability of SBCA during the training process. In both asset portfolios, the training loss of SBCA decreases steadily as the number of training steps increases, and gradually converges to a stable value, indicating that the model can effectively learn the market rules and portfolio optimization strategies through training. At the same time, the validation portfolio value shows a stable upward trend with the increase of training steps, and finally stabilizes at a high level, which proves that the model has good generalization ability and does not have overfitting problems. In \Cref{pic_4assets_step}, the training loss converges faster, and the validation portfolio value reaches a stable state earlier, which may be due to the relatively simple market information of the small-scale asset portfolio, making it easier for the model to learn effective strategies. In \Cref{pic_6assets_step}, the training loss converges slightly slower, but the validation portfolio value still maintains a stable upward trend, indicating that SBCA can still effectively handle the complex market information brought by the increase in asset scale. The stable convergence of the training curves further verifies the rationality of SBCA's structure and the effectiveness of the training strategy, laying a solid foundation for its excellent performance in the test period.

\section{Conclution}\label{sec:conclution}

This study proposes a cross-modal BERT-driven Actor-Critic framework (SBCA) for multi-asset portfolio optimization. The model integrates price time-series and textual sentiment features through a cross-modal gated fusion mechanism, and incorporates downside risk and turnover constraints into the reward function to align with real trading scenarios. Theoretical analysis verifies the utility consistency of the designed reward function and the superior nonlinear expressiveness of gated fusion. Empirical results on 2/4/6-asset portfolios show that SBCA outperforms equal weight, buy-and-hold and market benchmark strategies in profitability, risk control and risk-adjusted returns. Ablation experiments and cost sensitivity analysis confirm the effectiveness of core modules and the robustness of the model under different transaction cost levels.

Future research can extend the SBCA framework to high-dimensional asset pools and cross-market scenarios, and incorporate more multi-modal data such as trading volume and social media content to enhance information perception. In addition, combining causal inference and online learning can further improve the interpretability and real-time adaptability of the strategy. Extreme market conditions and high-frequency trading environments can be introduced to test and strengthen the model’s stability in more volatile scenarios.

	\section*{Declarations}
\begin{itemize}
	\item \textbf{Funding}: Not applicable.
	\item \textbf{Competing Interests}: The authors declare no competing interests.
	\item \textbf{Ethical approval}: Not applicable.
	\item \textbf{Consent to participate}: Not applicable.
	\item \textbf{Consent for publication}: Not applicable.
	\item \textbf{Data availability}: The data are openly available at 
	\item \textbf{Code availability}: The code is available at 
	\item \textbf{Author contributions}: All authors contributed equally. All authors have read and approved the final manuscript.
	\item \textbf{AI disclosure}: Language polishing tools were used in preparation; authors take full responsibility for all content.
\end{itemize}

\newpage
\appendix\section{Additional Experimental Figures}
\label{app1}

\subsection{2 Assets (CAT, GILD)}

\subsubsection{CAT}
\begin{figure}[H]
	\centering
	\subfloat[SB]{\includegraphics[width=0.24\textwidth]{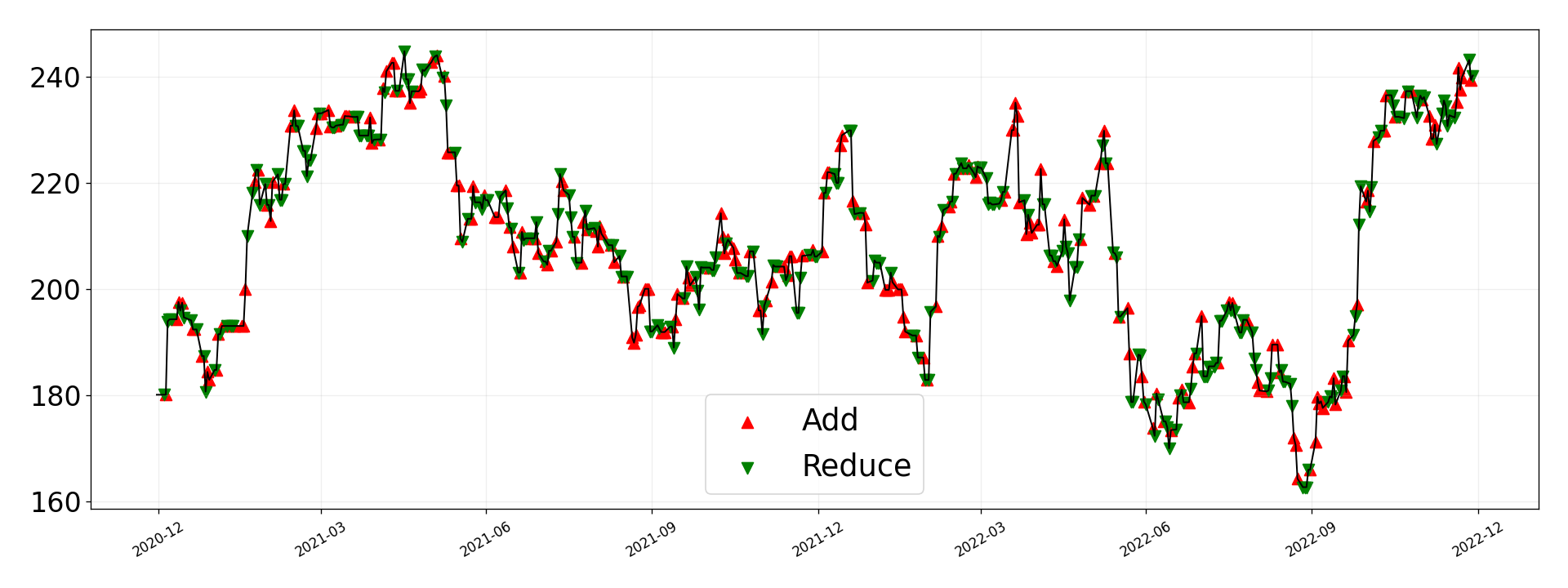}}
	\subfloat[SBA]{\includegraphics[width=0.24\textwidth]{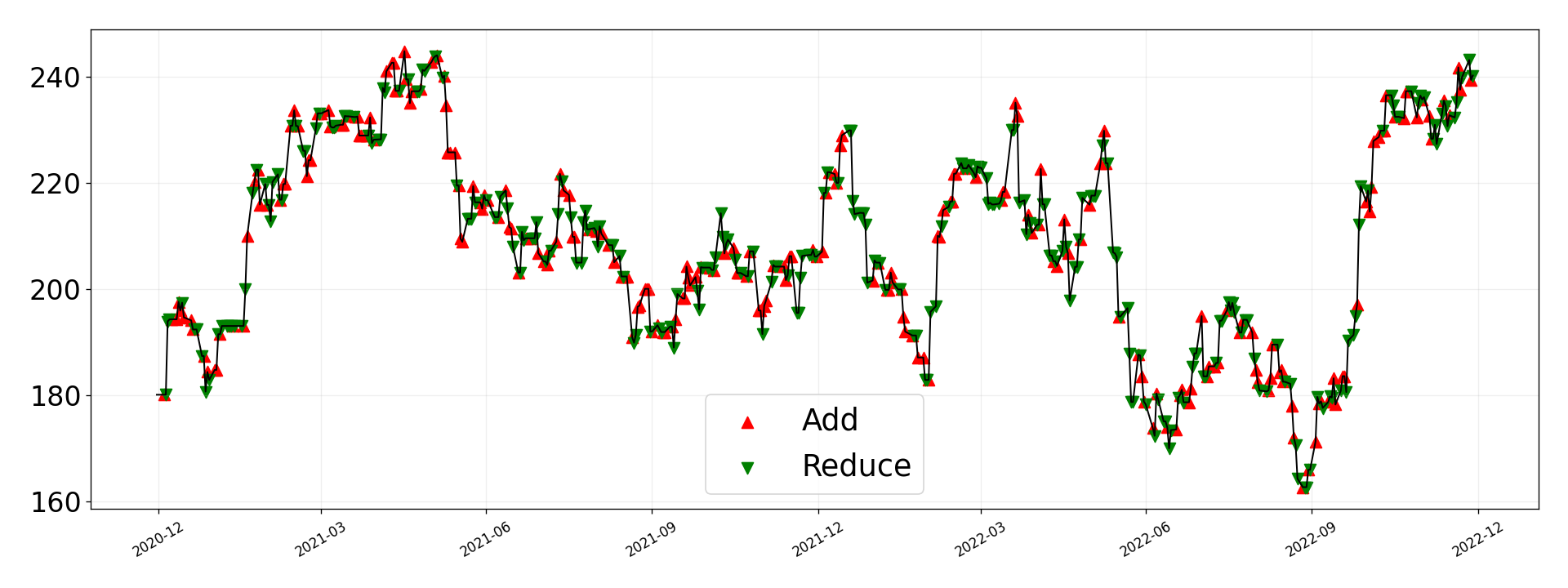}}
	\subfloat[SBC]{\includegraphics[width=0.24\textwidth]{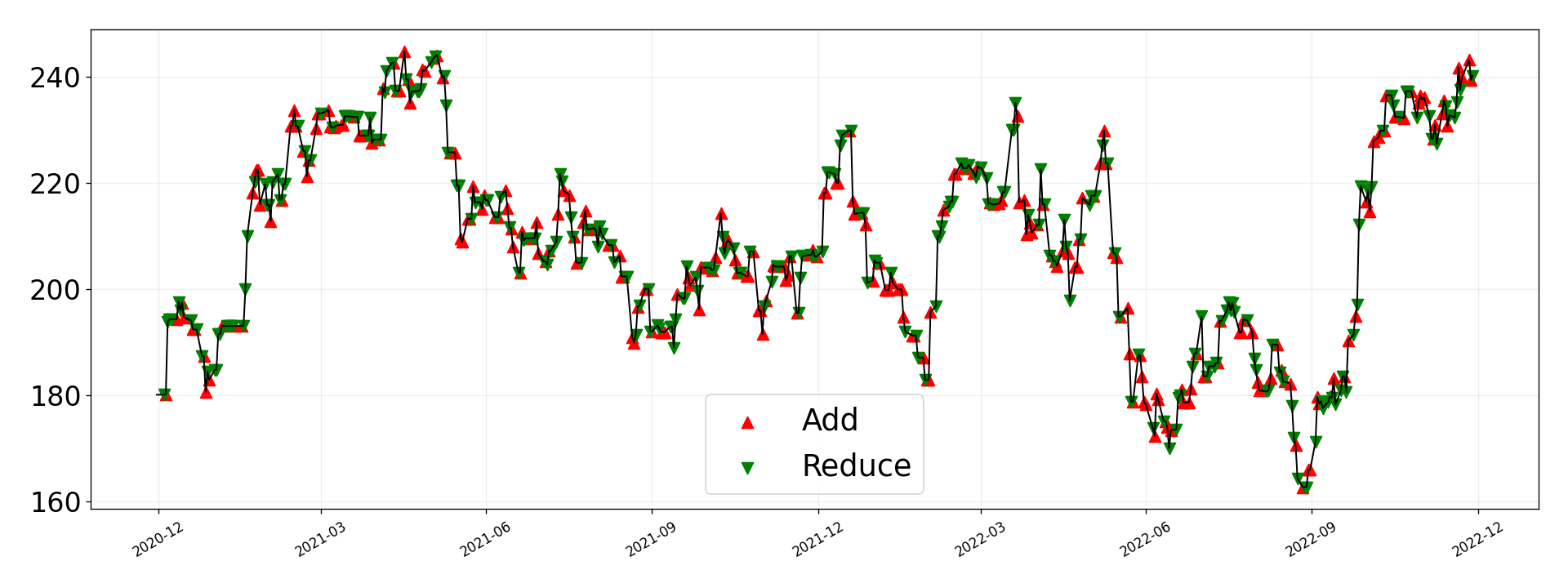}}
	\subfloat[SBCA]{\includegraphics[width=0.24\textwidth]{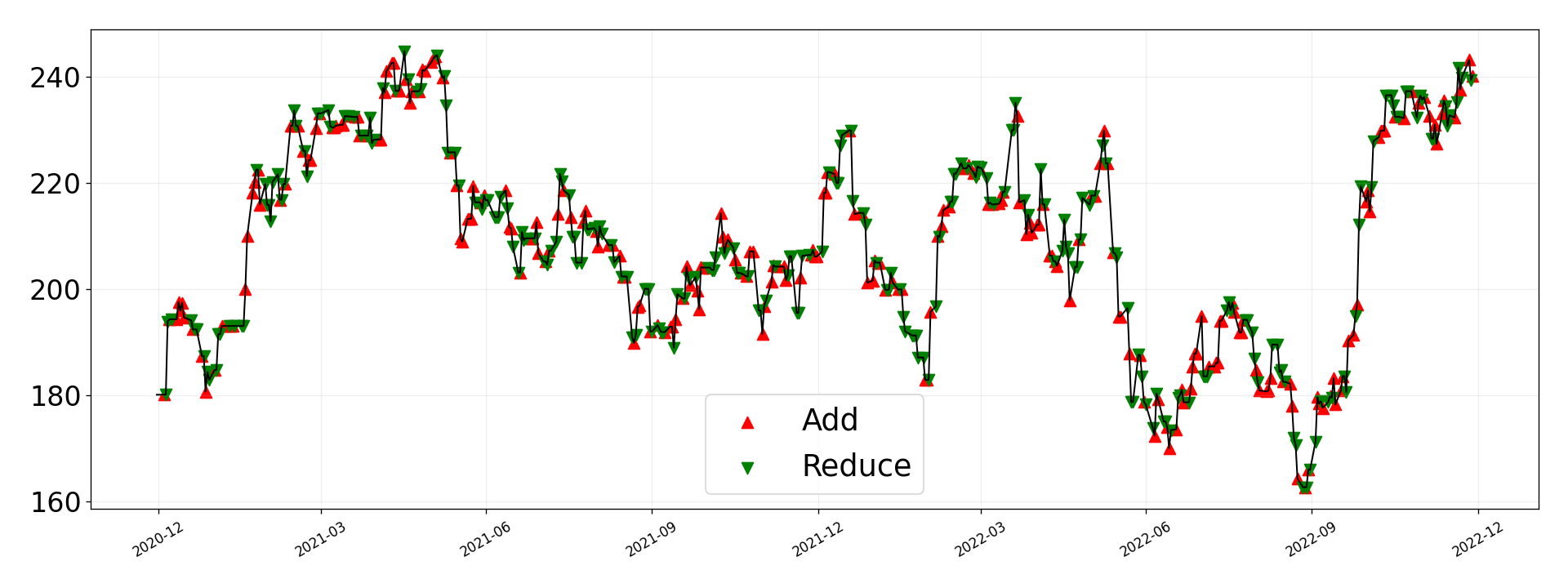}}
	\caption{Trading signals for CAT in 2-asset portfolio.}
\end{figure}

\subsubsection{GILD}
\begin{figure}[H]
	\centering
	\subfloat[SB]{\includegraphics[width=0.24\textwidth]{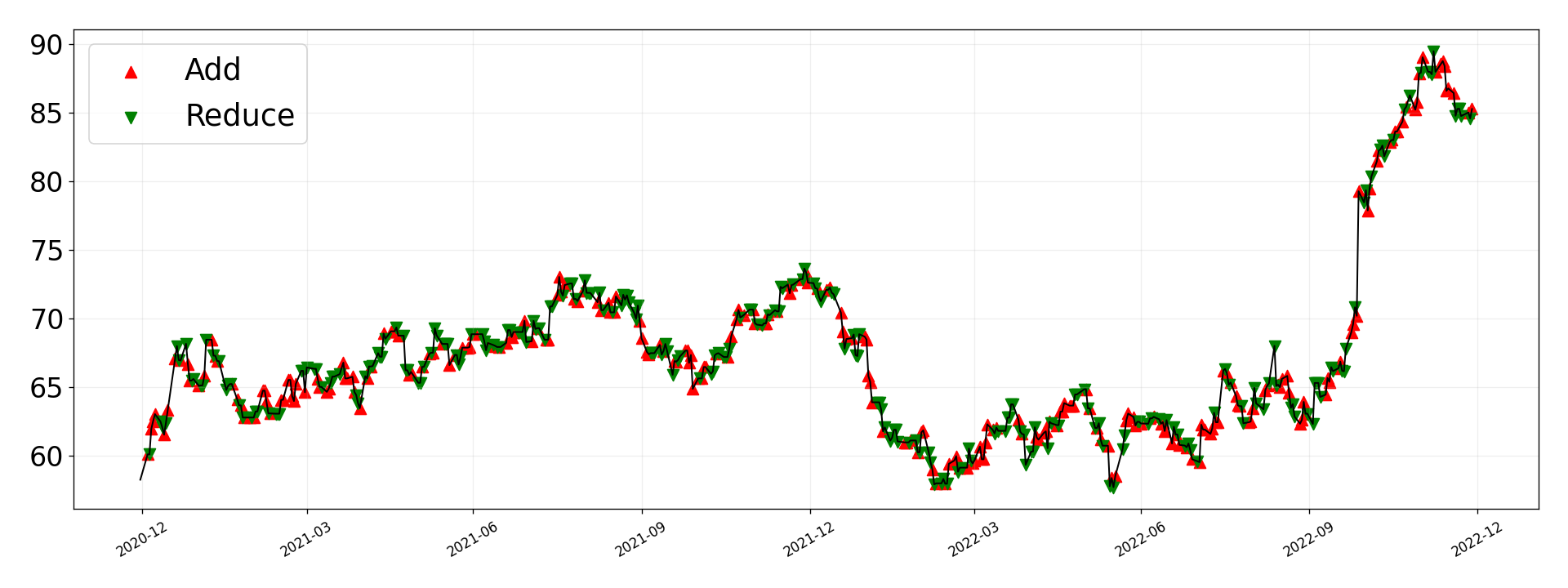}}
	\subfloat[SBA]{\includegraphics[width=0.24\textwidth]{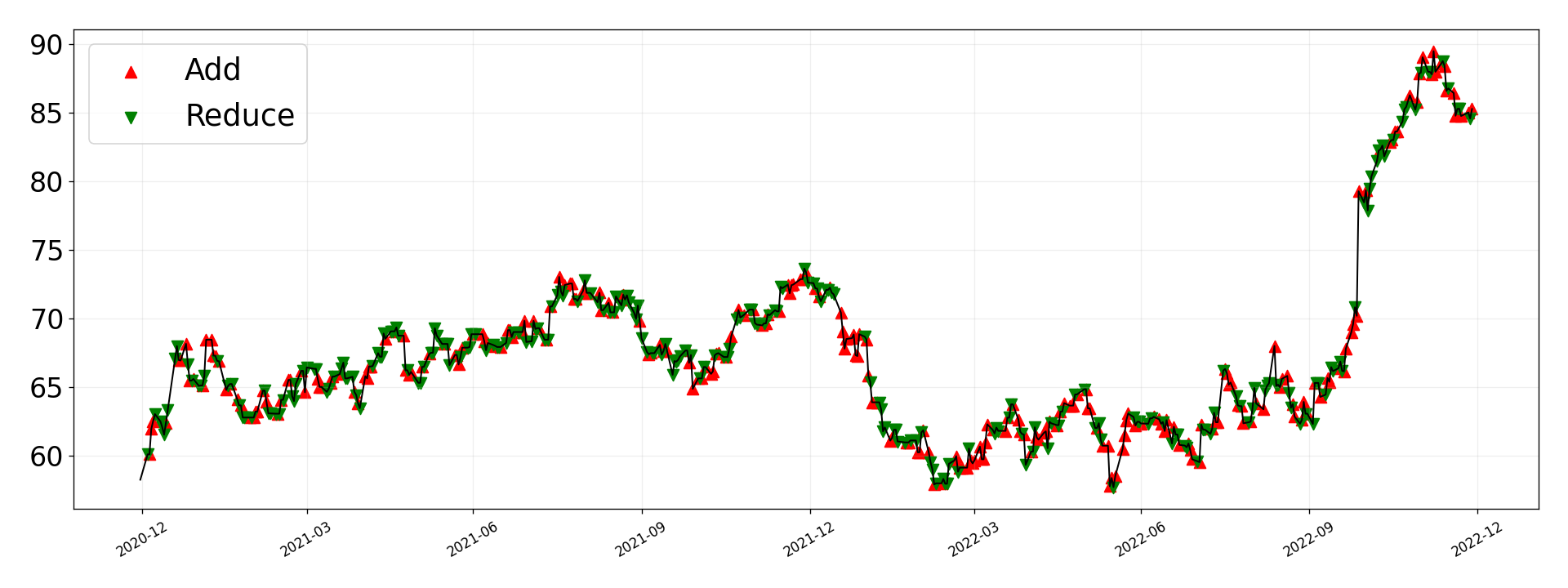}}
	\subfloat[SBC]{\includegraphics[width=0.24\textwidth]{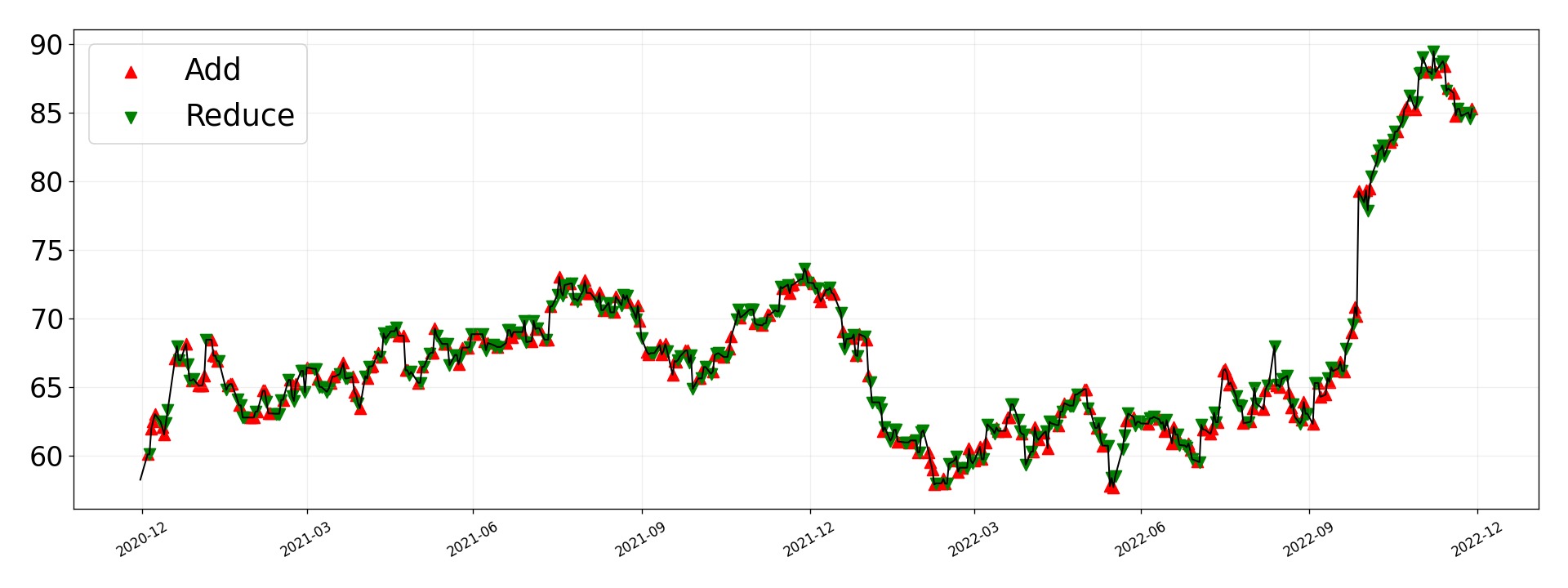}}
	\subfloat[SBCA]{\includegraphics[width=0.24\textwidth]{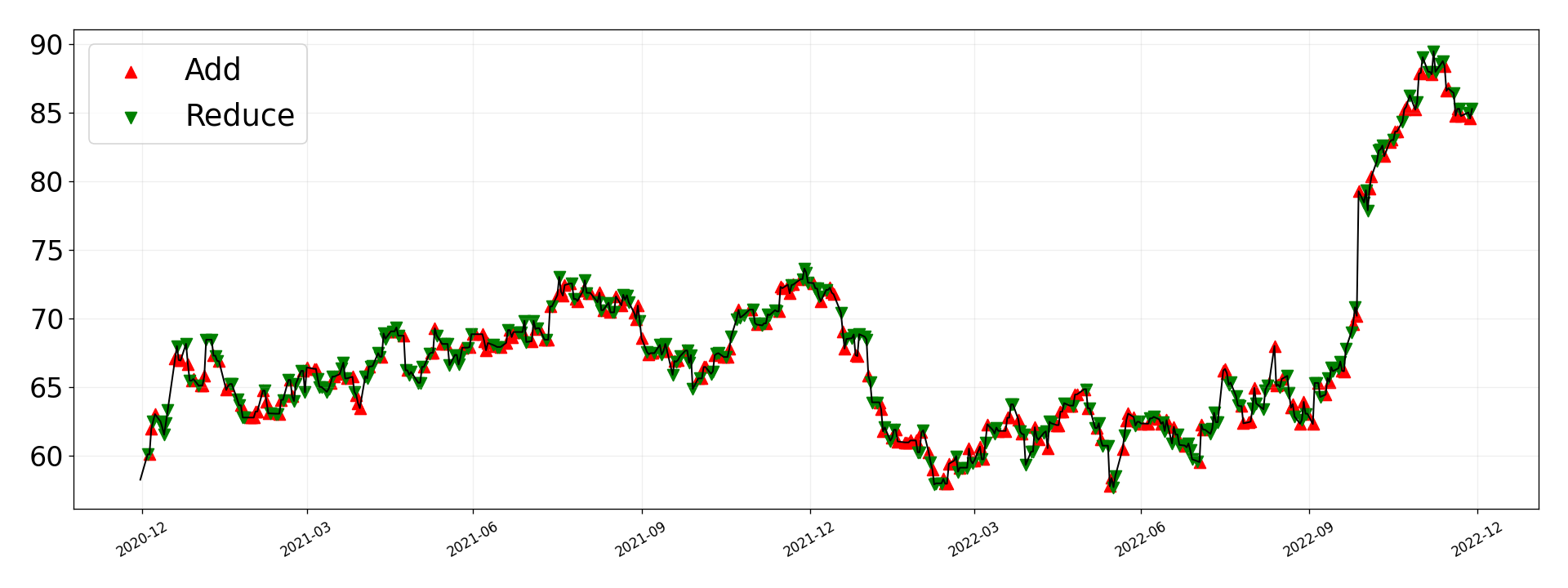}}
	\caption{Trading signals for GILD in 2-asset portfolio.}
\end{figure}

\subsection{4 Assets (CAT, GILD, GS, KO)}

\subsubsection{CAT}
\begin{figure}[H]
	\centering
	\subfloat[SB]{\includegraphics[width=0.24\textwidth]{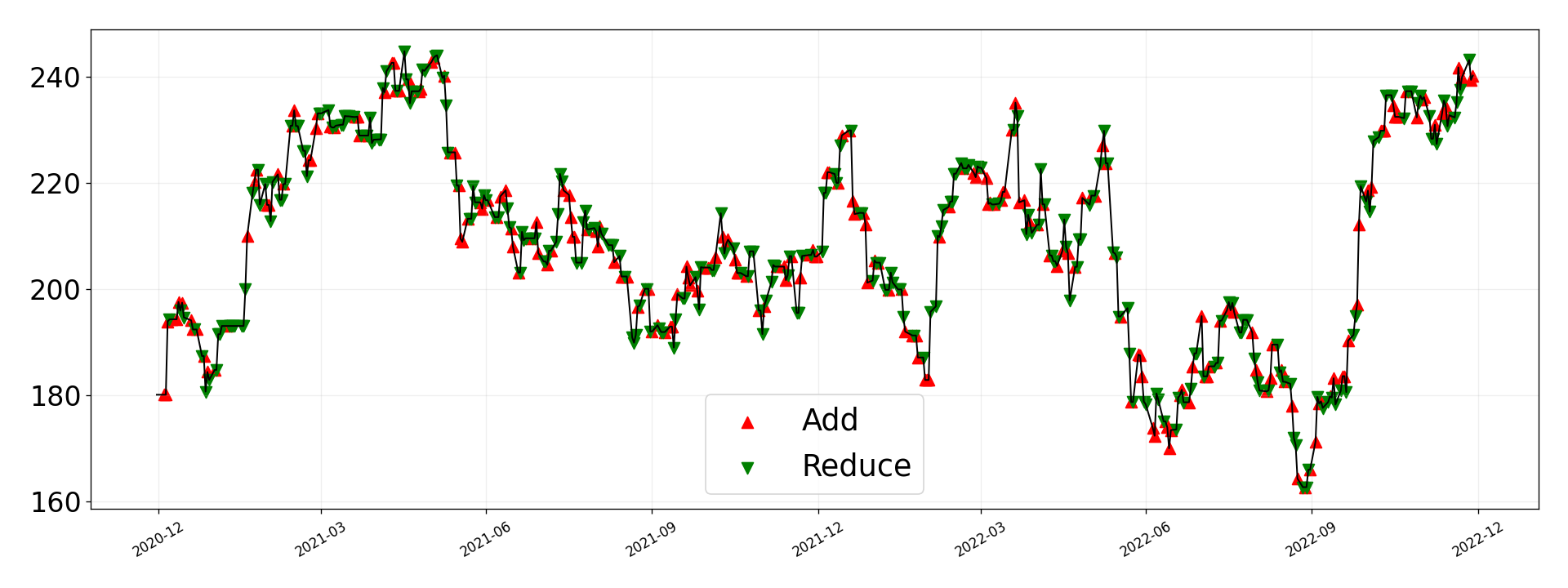}}
	\subfloat[SBA]{\includegraphics[width=0.24\textwidth]{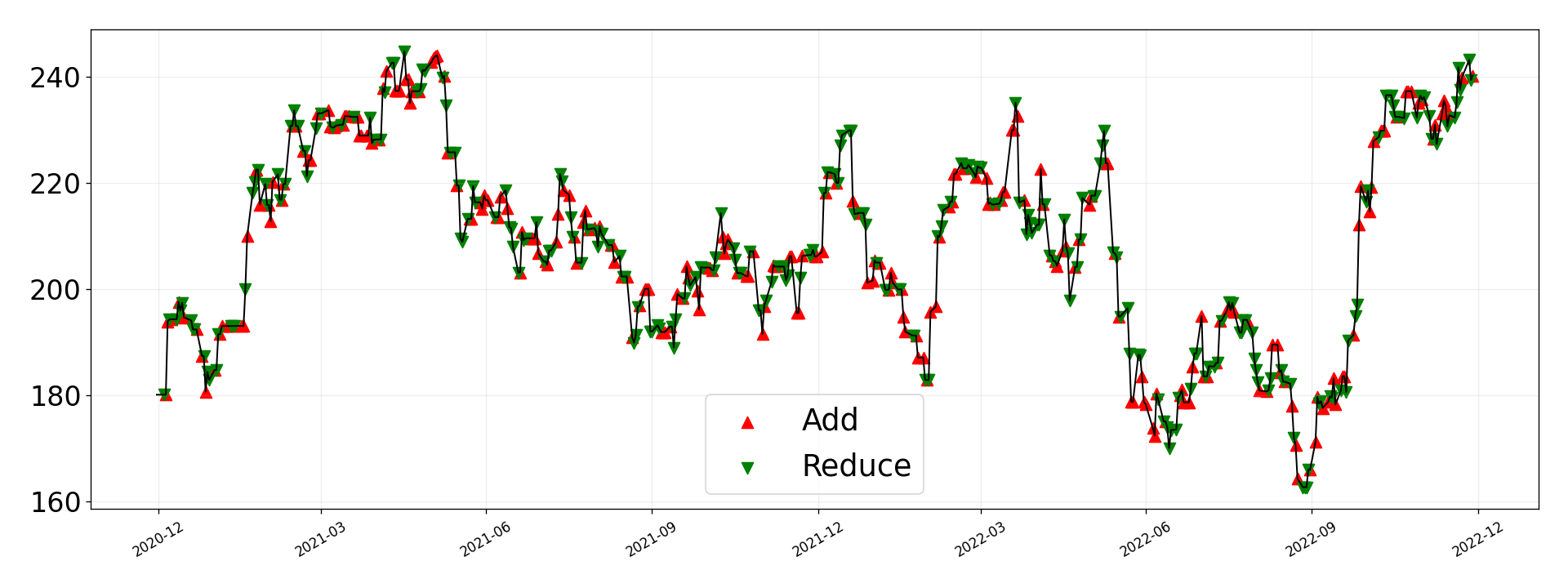}}
	\subfloat[SBC]{\includegraphics[width=0.24\textwidth]{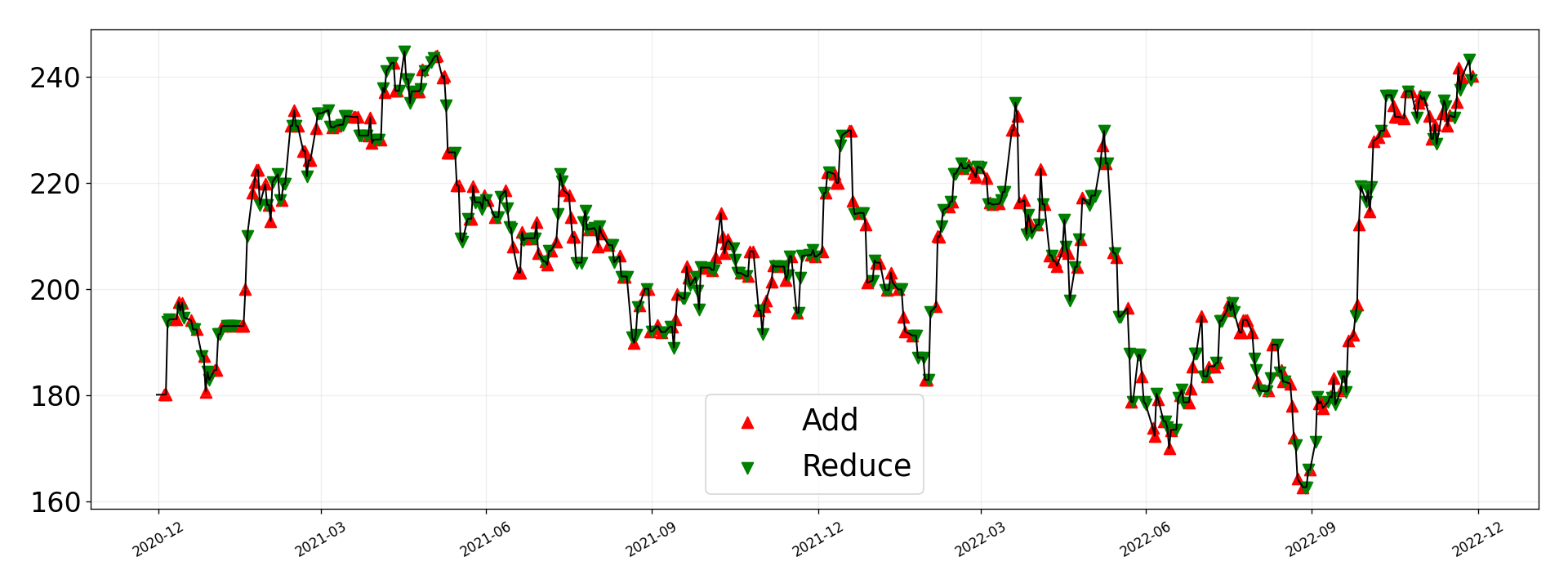}}
	\subfloat[SBCA]{\includegraphics[width=0.24\textwidth]{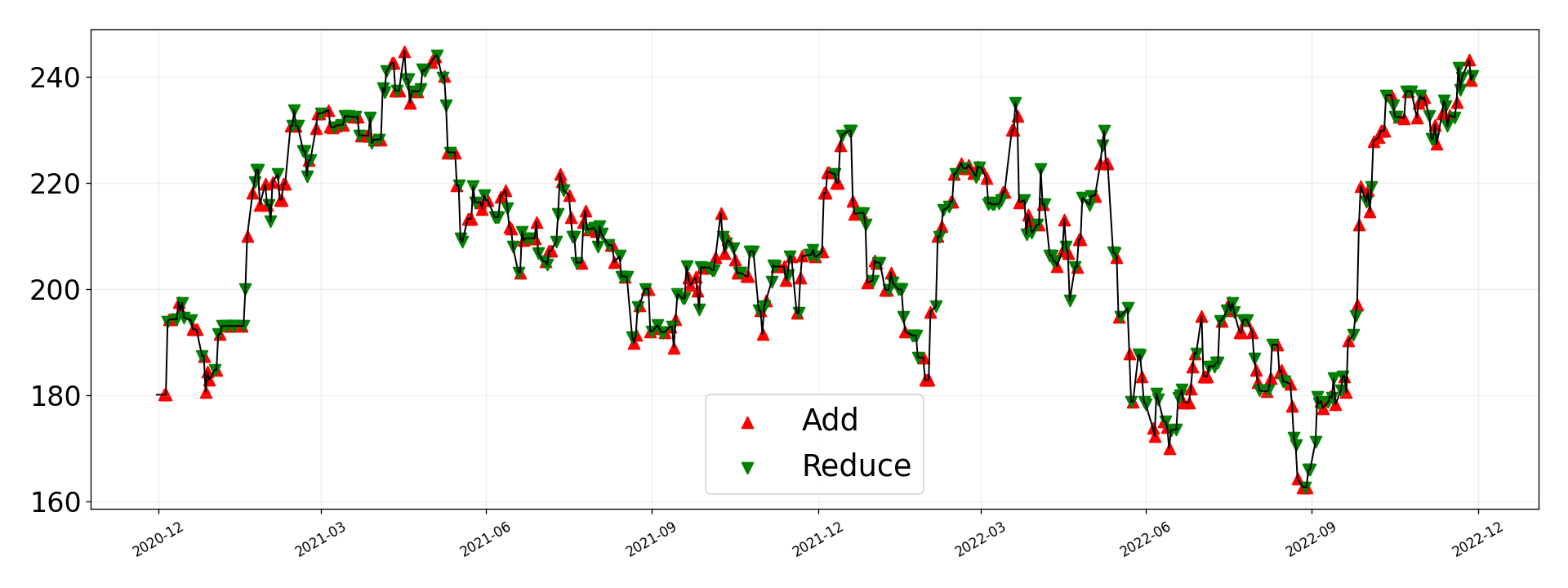}}
	\caption{Trading signals of CAT in 4-asset portfolio.}
\end{figure}

\subsubsection{GILD}
\begin{figure}[H]
	\centering
	\subfloat[SB]{\includegraphics[width=0.24\textwidth]{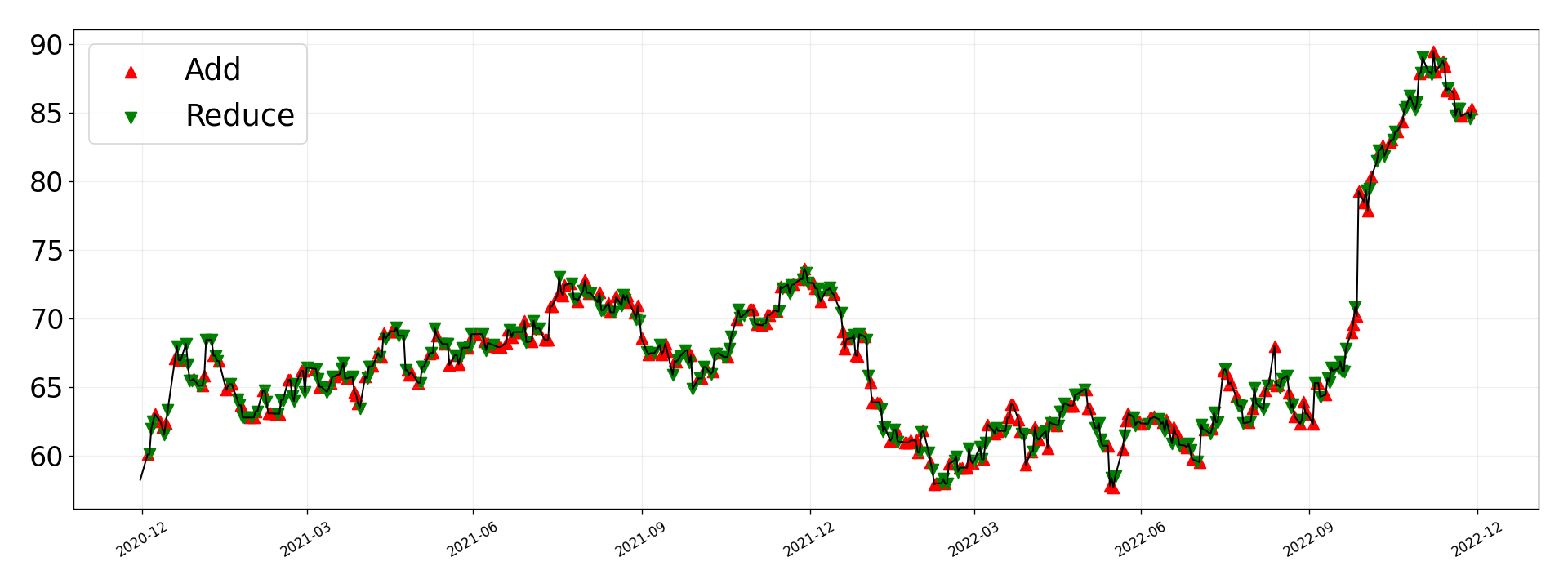}}
	\subfloat[SBA]{\includegraphics[width=0.24\textwidth]{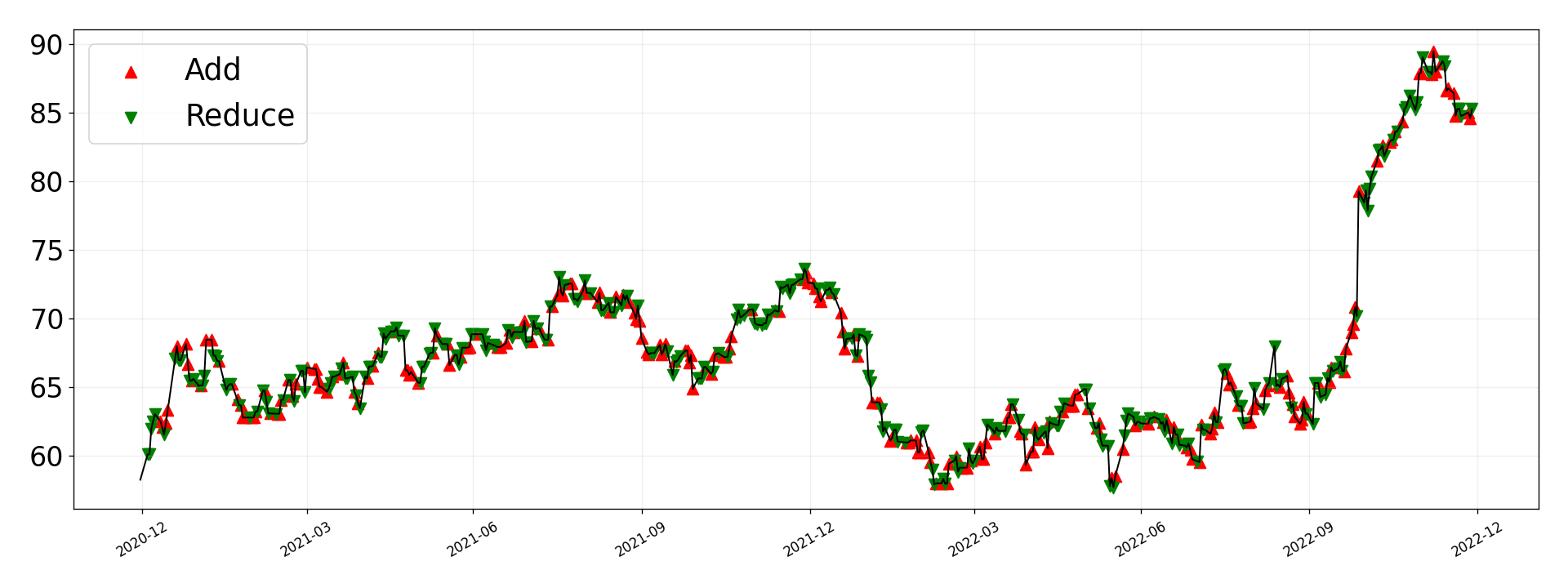}}
	\subfloat[SBC]{\includegraphics[width=0.24\textwidth]{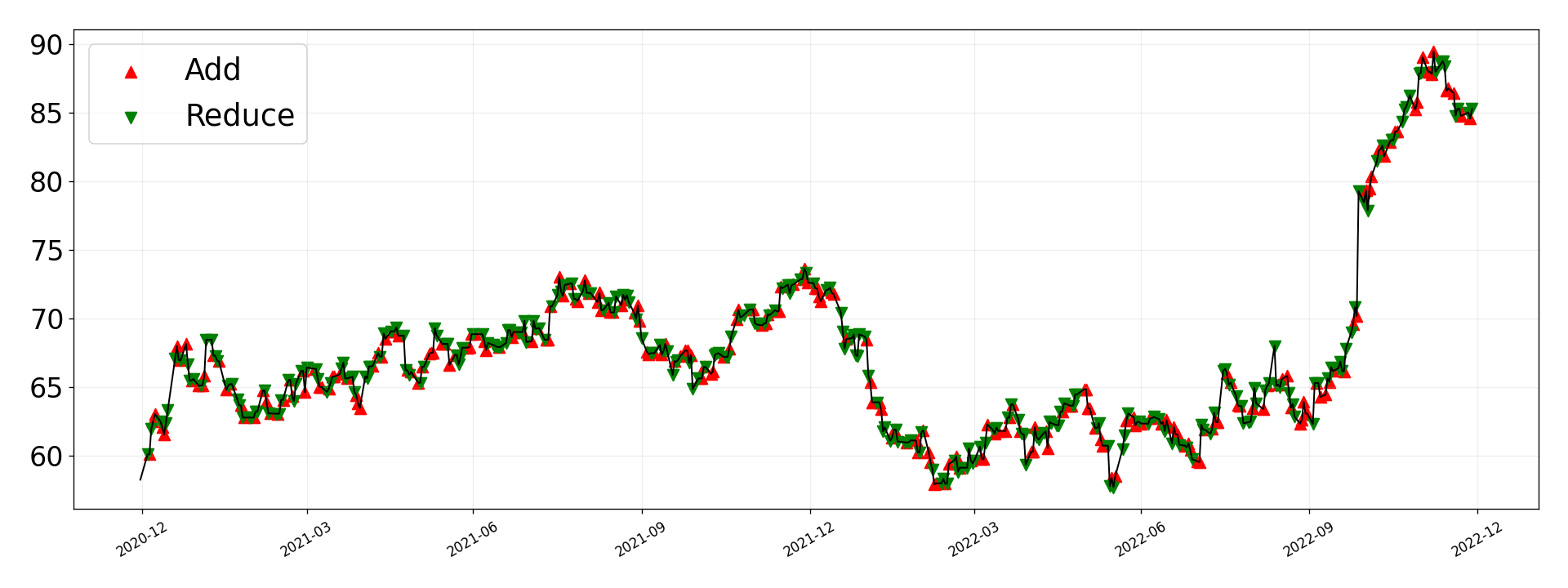}}
	\subfloat[SBCA]{\includegraphics[width=0.24\textwidth]{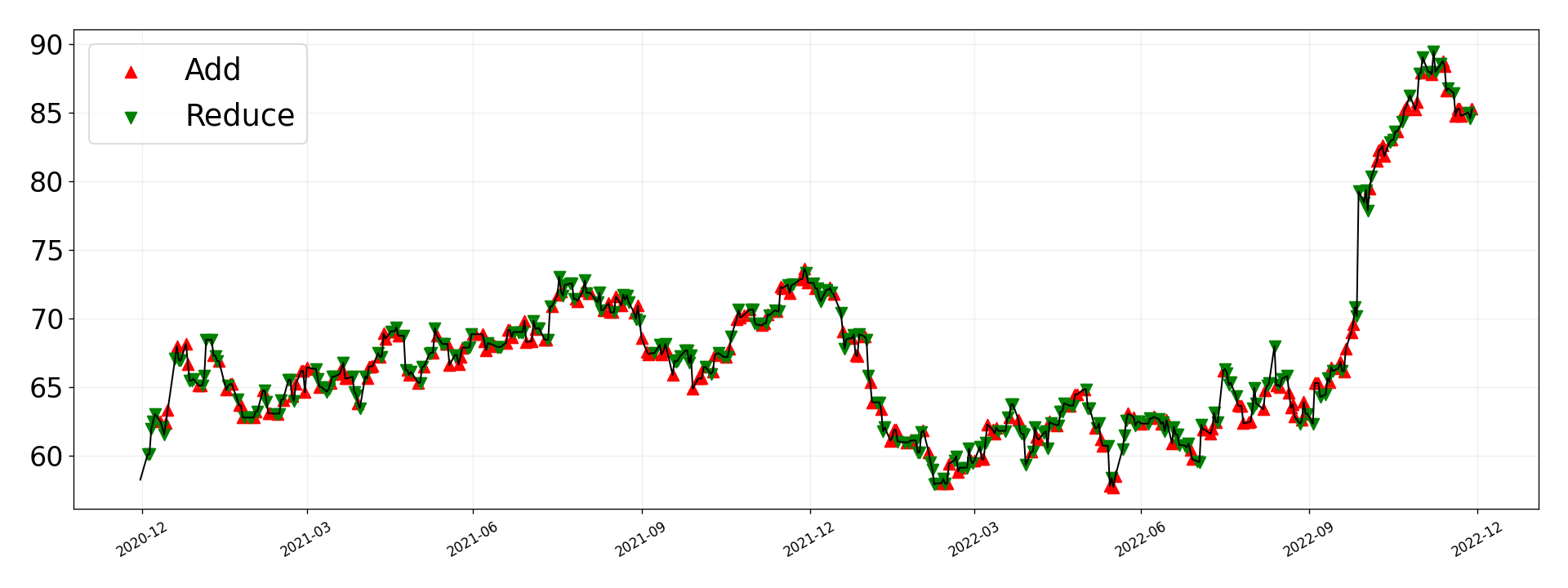}}
	\caption{Trading signals of GILD in 4-asset portfolio.}
\end{figure}

\subsubsection{GS}
\begin{figure}[H]
	\centering
	\subfloat[SB]{\includegraphics[width=0.24\textwidth]{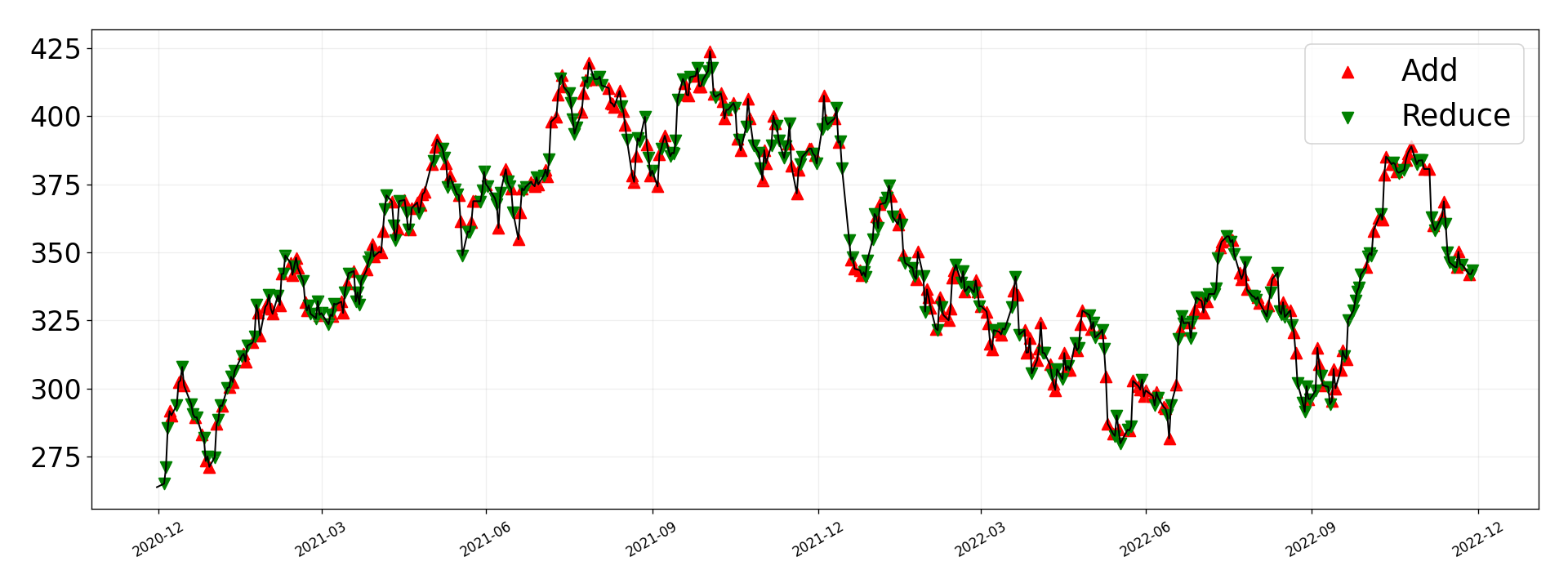}}
	\subfloat[SBA]{\includegraphics[width=0.24\textwidth]{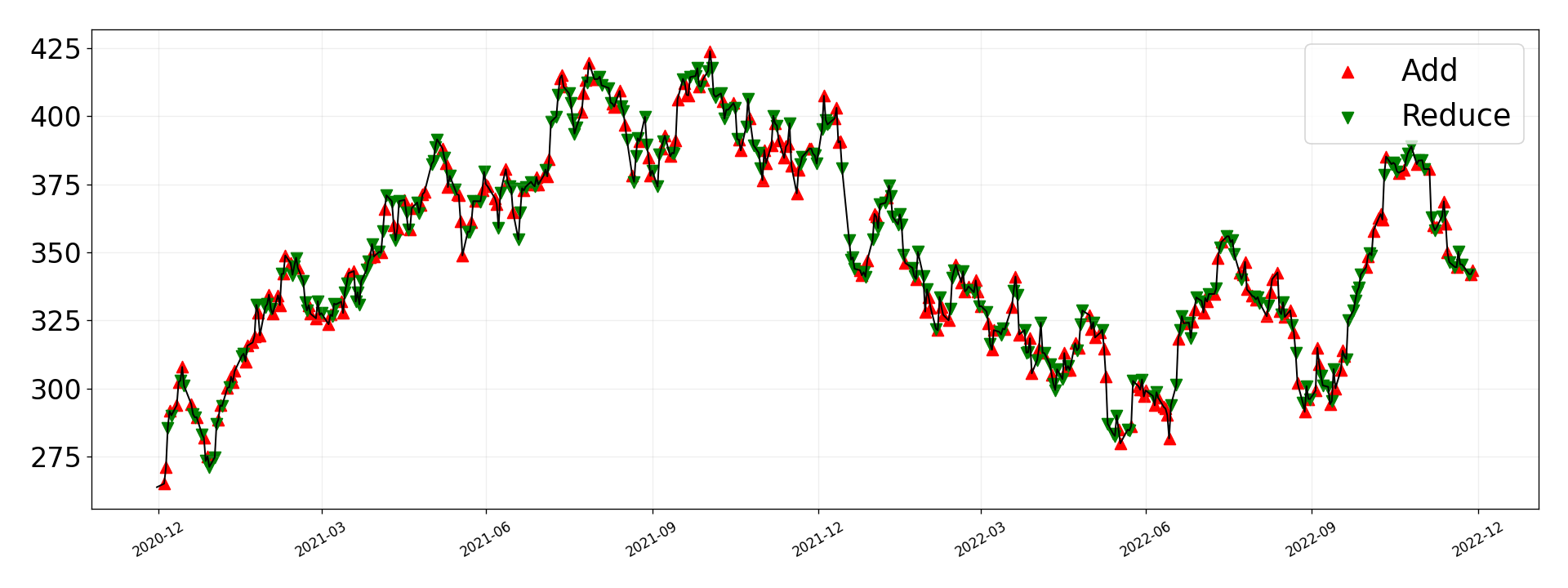}}
	\subfloat[SBC]{\includegraphics[width=0.24\textwidth]{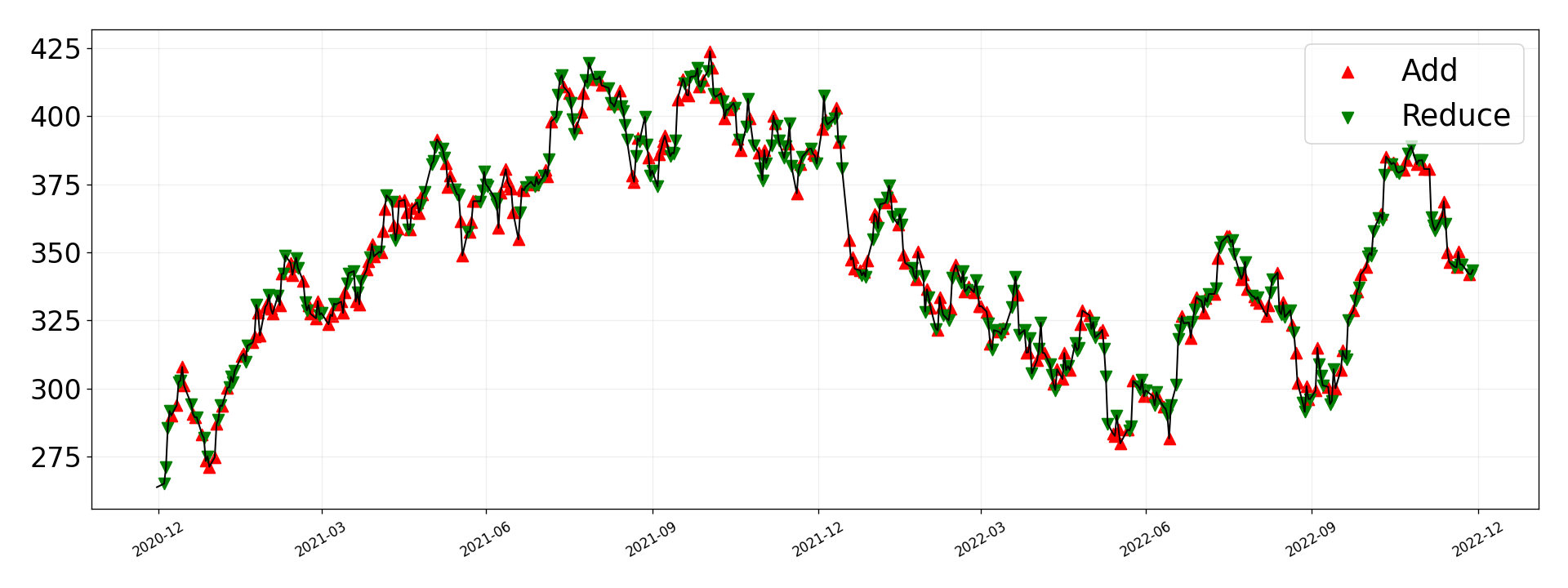}}
	\subfloat[SBCA]{\includegraphics[width=0.24\textwidth]{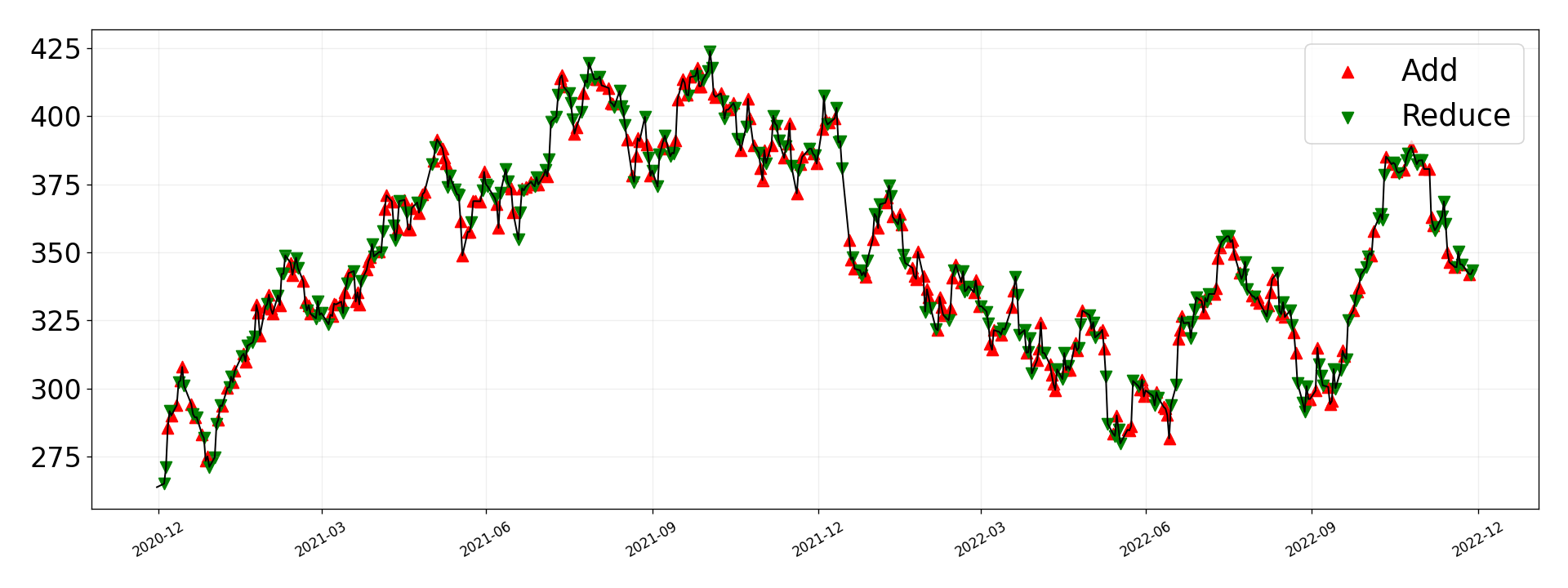}}
	\caption{Trading signals of GS in 4-asset portfolio.}
\end{figure}

\subsubsection{KO}
\begin{figure}[H]
	\centering
	\subfloat[SB]{\includegraphics[width=0.24\textwidth]{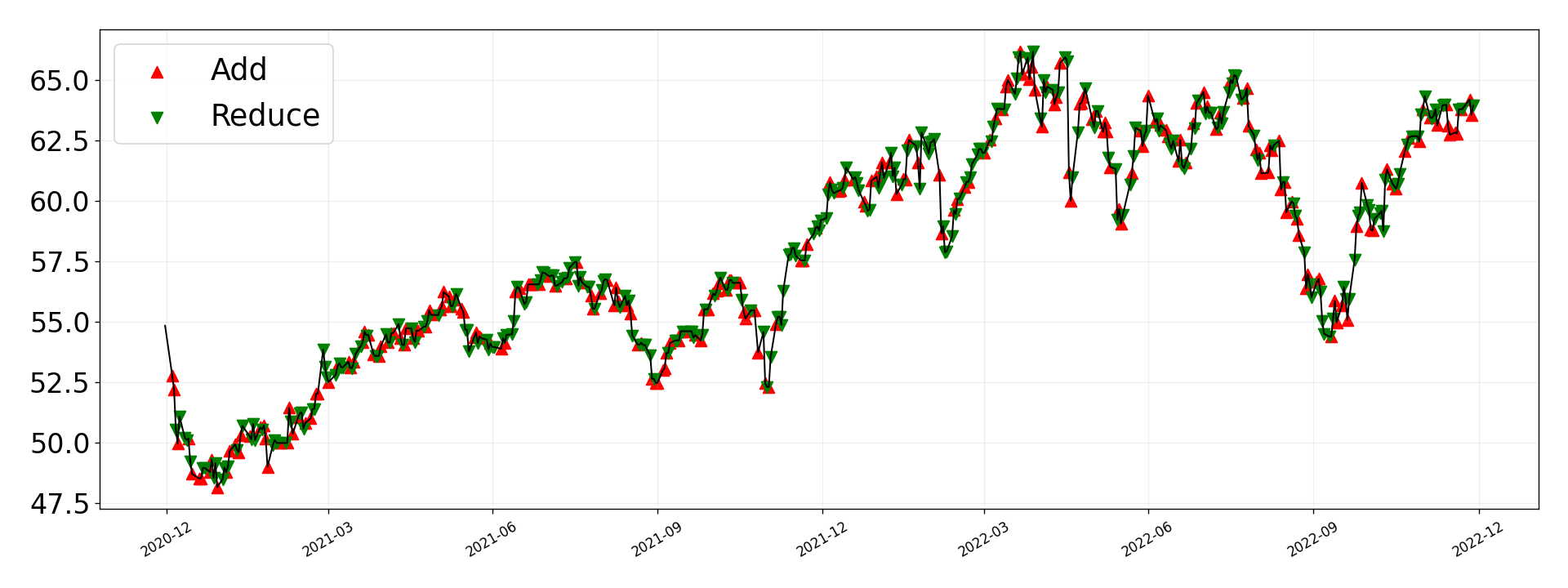}}
	\subfloat[SBA]{\includegraphics[width=0.24\textwidth]{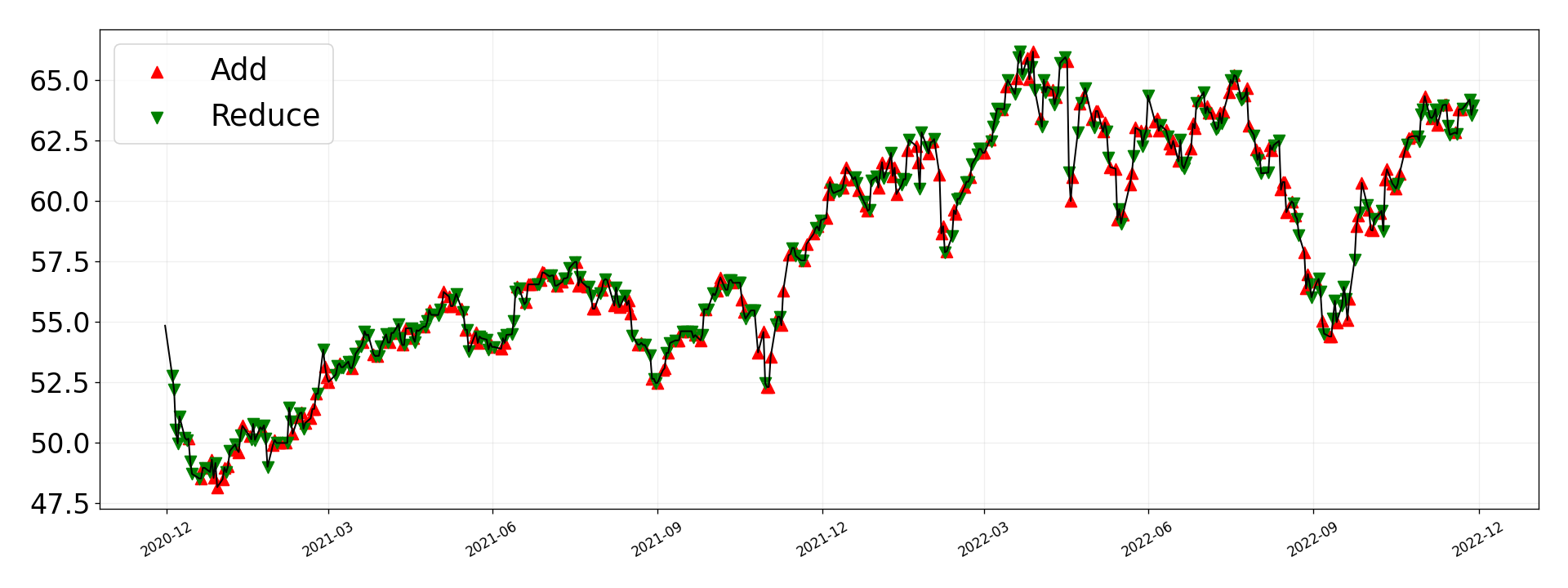}}
	\subfloat[SBC]{\includegraphics[width=0.24\textwidth]{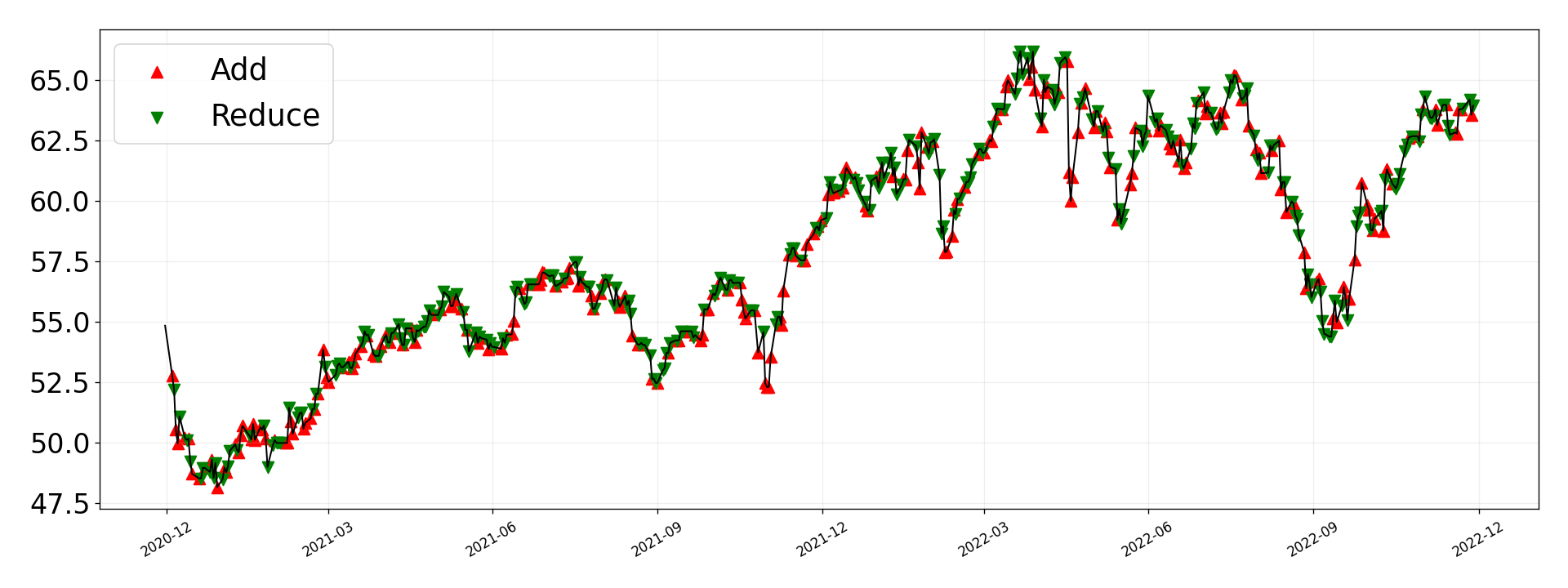}}
	\subfloat[SBCA]{\includegraphics[width=0.24\textwidth]{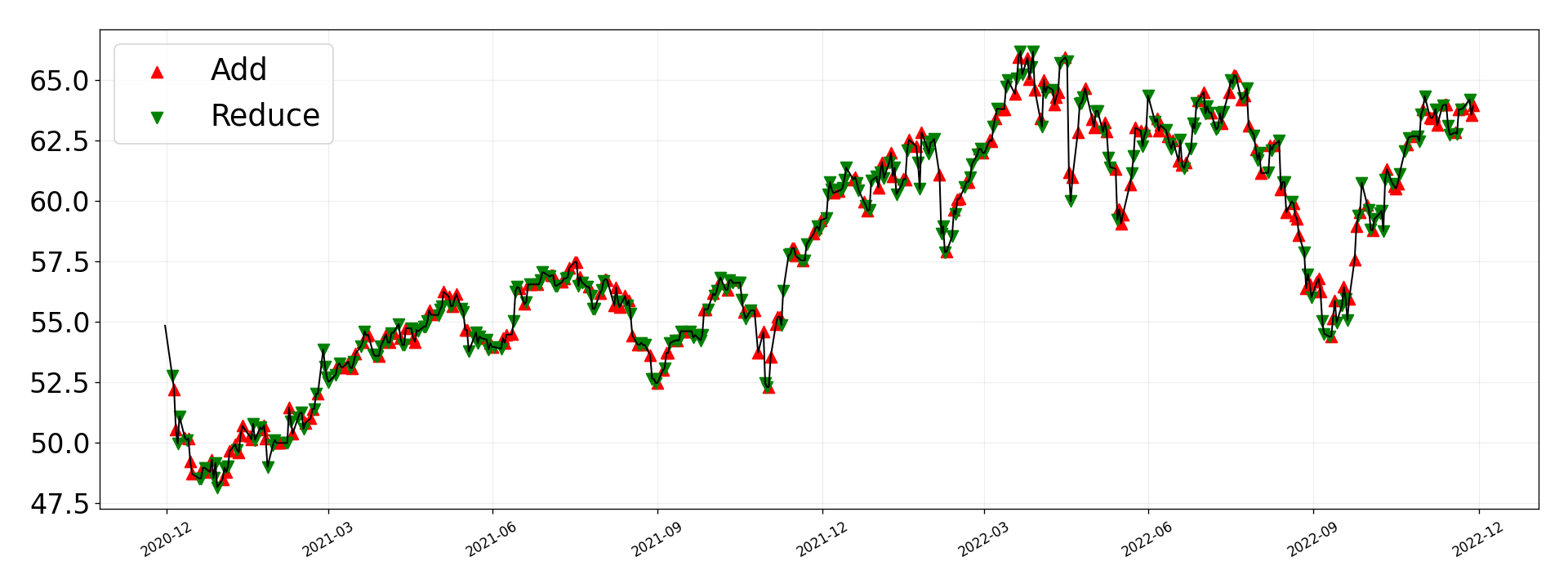}}
	\caption{Trading signals of KO in 4-asset portfolio.}
\end{figure}

\subsection{6 Assets (CAT, GILD, GS, KO, MRK, NVDA)}

\subsubsection{CAT}
\begin{figure}[H]
	\centering
	\subfloat[SB]{\includegraphics[width=0.24\textwidth]{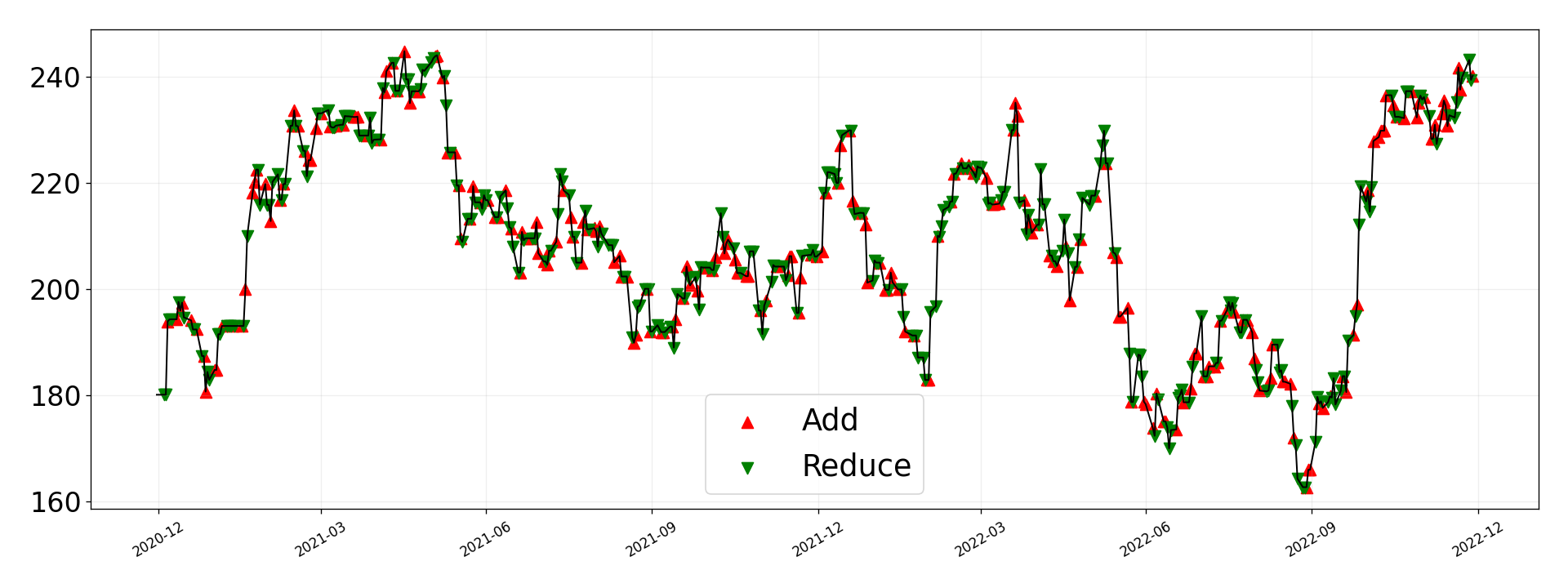}}
	\subfloat[SBA]{\includegraphics[width=0.24\textwidth]{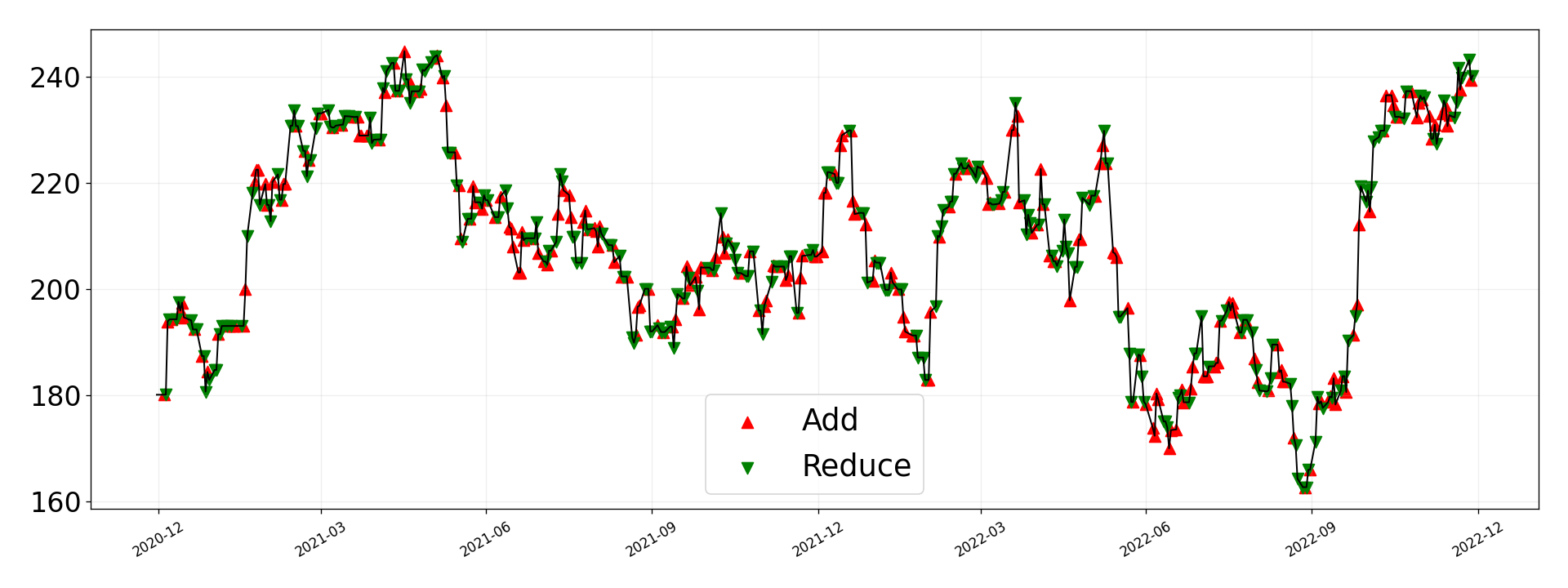}}
	\subfloat[SBC]{\includegraphics[width=0.24\textwidth]{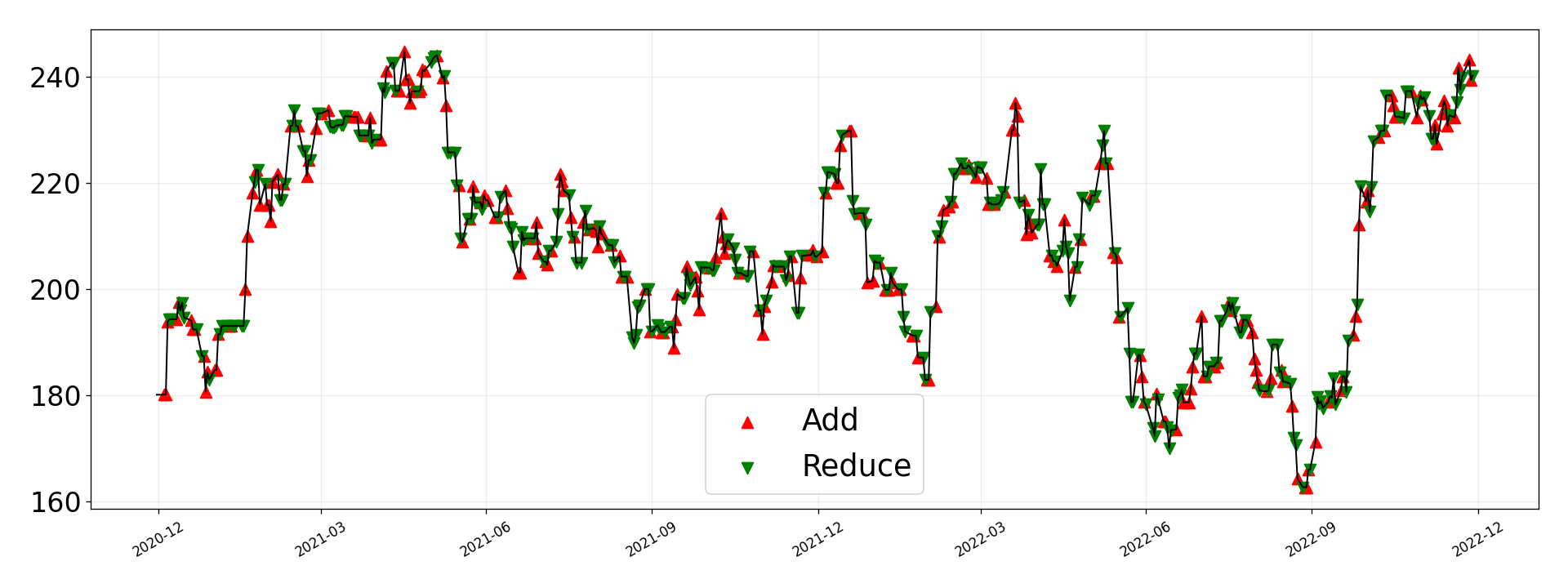}}
	\subfloat[SBCA]{\includegraphics[width=0.24\textwidth]{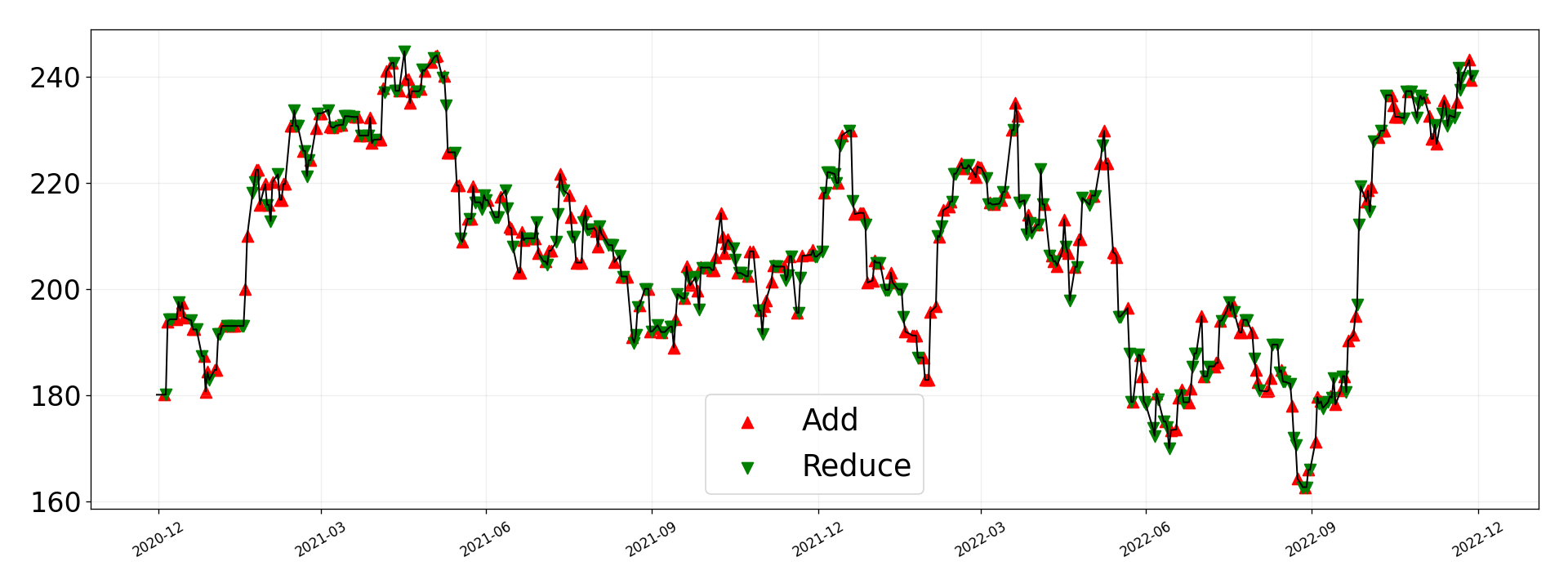}}
	\caption{Trading signals of CAT in 6-asset portfolio.}
\end{figure}

\subsubsection{GILD}
\begin{figure}[H]
	\centering
	\subfloat[SB]{\includegraphics[width=0.24\textwidth]{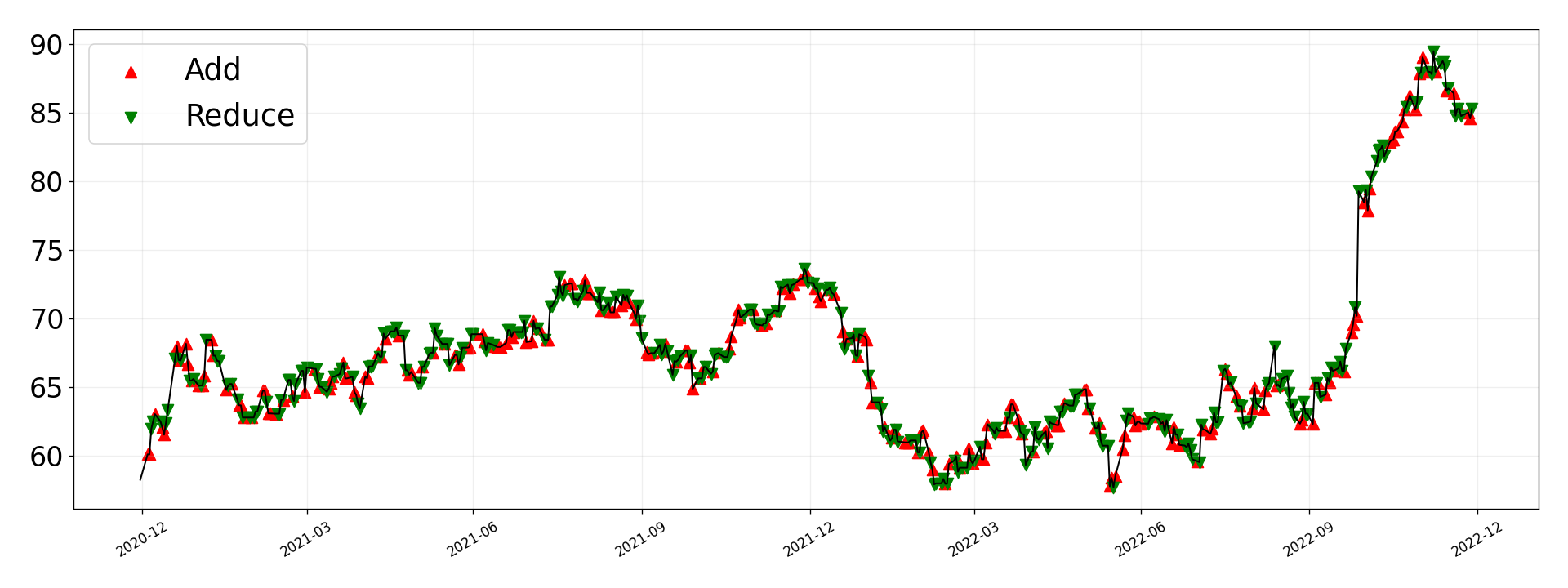}}
	\subfloat[SBA]{\includegraphics[width=0.24\textwidth]{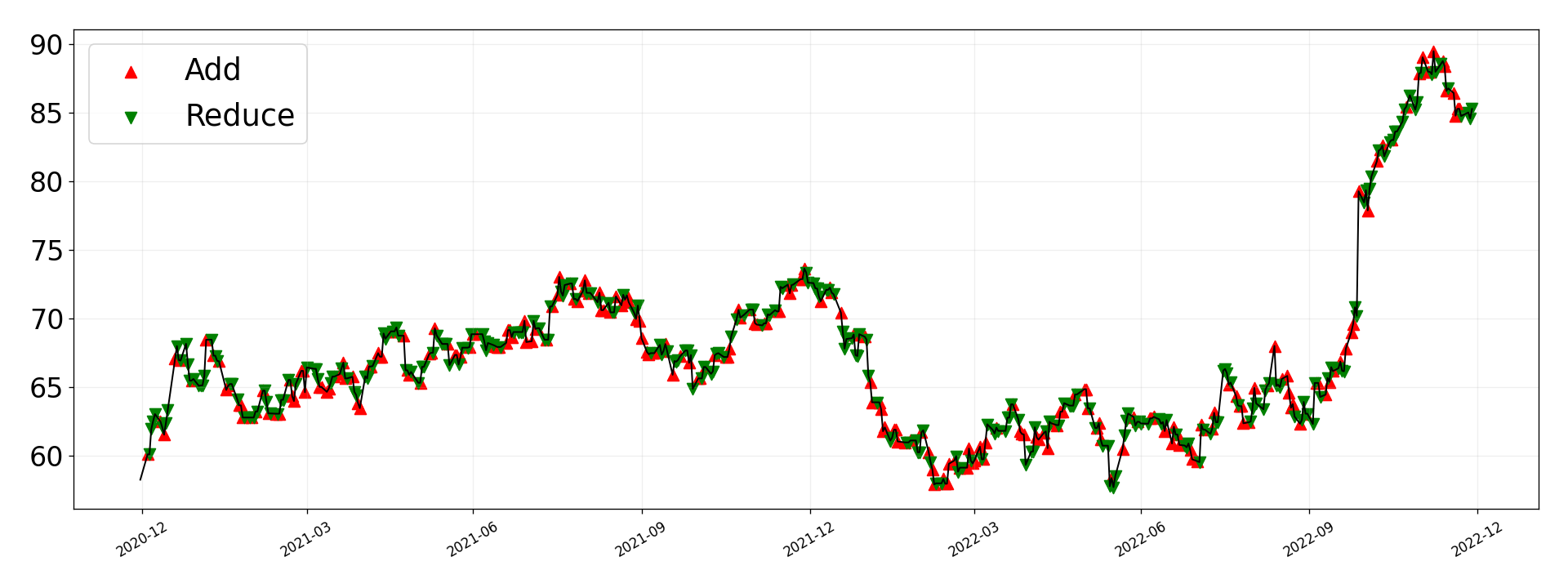}}
	\subfloat[SBC]{\includegraphics[width=0.24\textwidth]{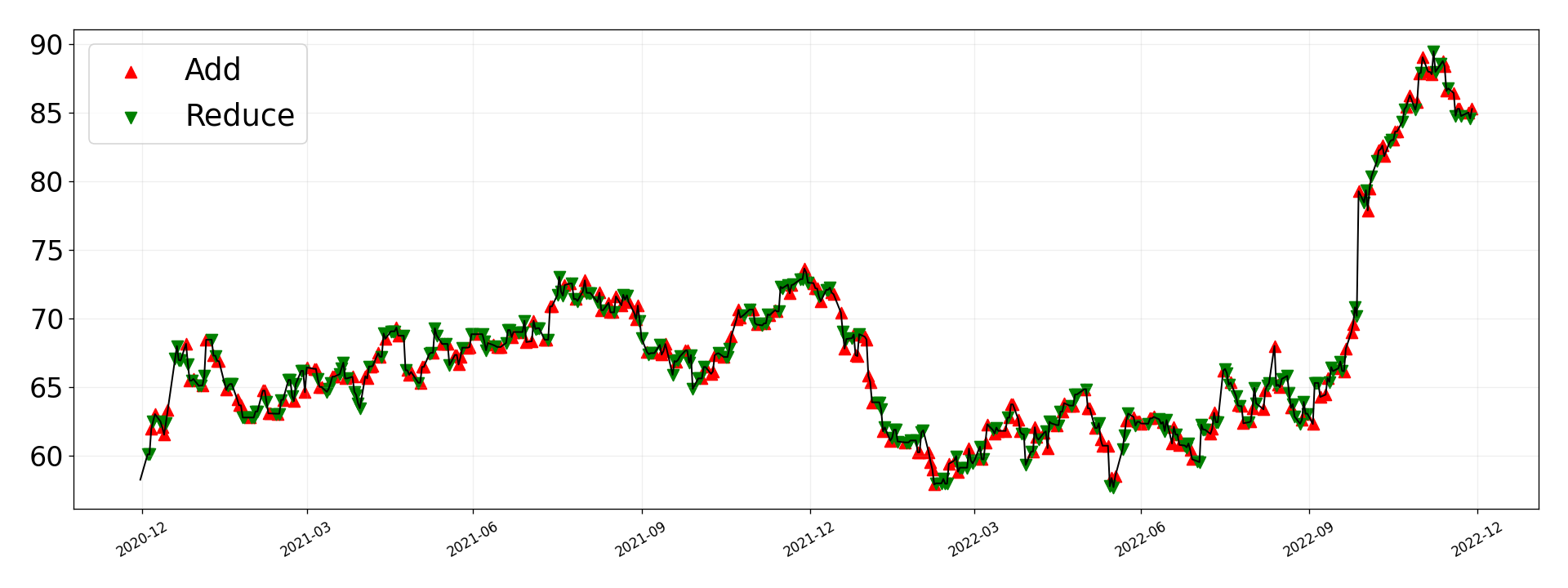}}
	\subfloat[SBCA]{\includegraphics[width=0.24\textwidth]{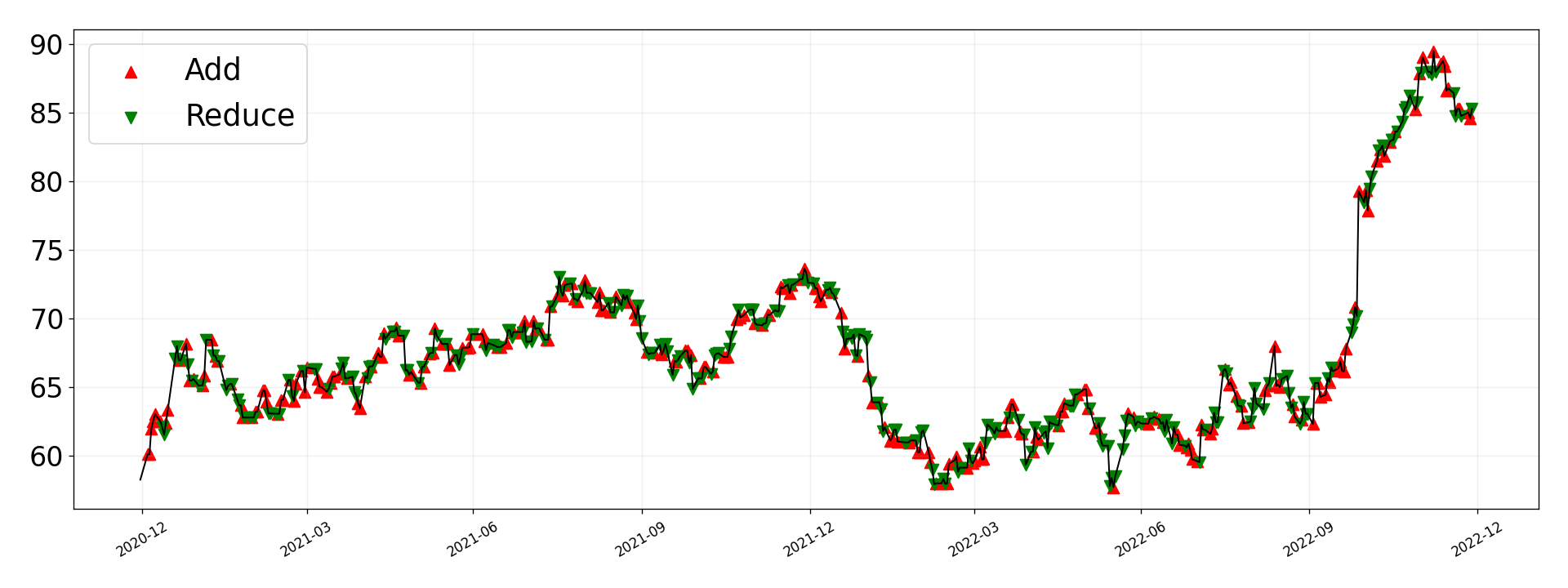}}
	\caption{Trading signals of GILD in 6-asset portfolio.}
\end{figure}

\subsubsection{GS}
\begin{figure}[H]
	\centering
	\subfloat[SB]{\includegraphics[width=0.24\textwidth]{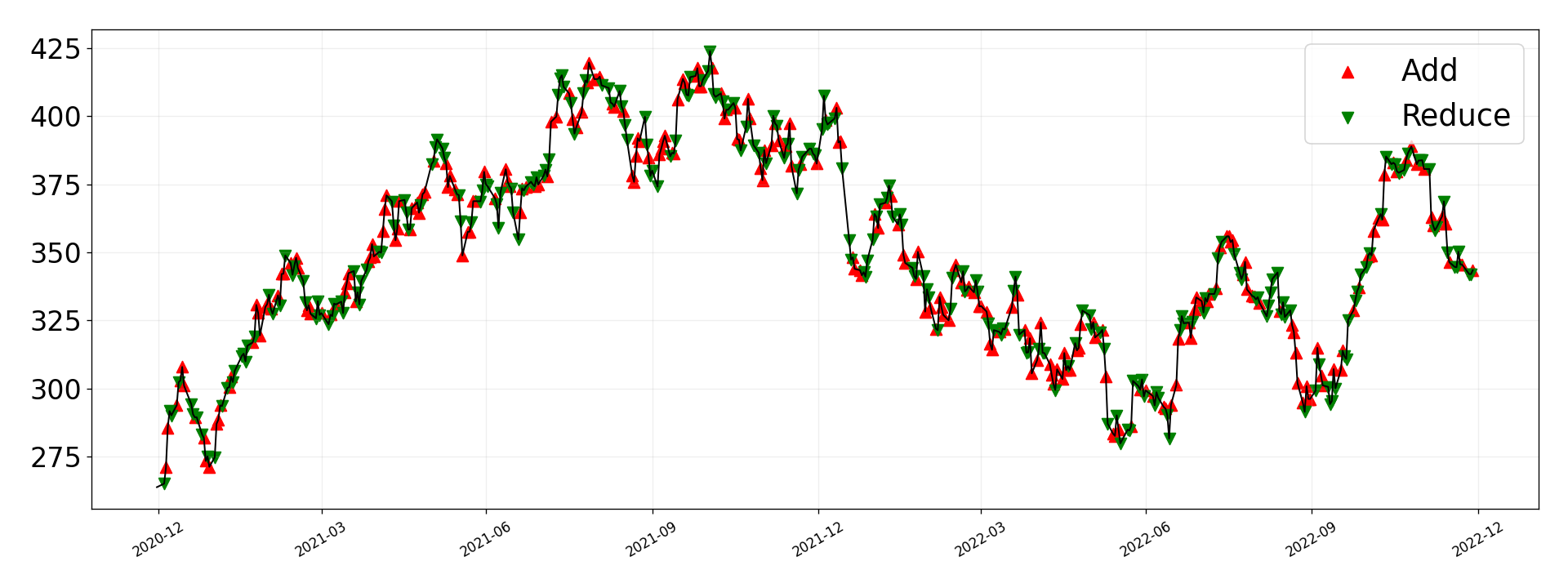}}
	\subfloat[SBA]{\includegraphics[width=0.24\textwidth]{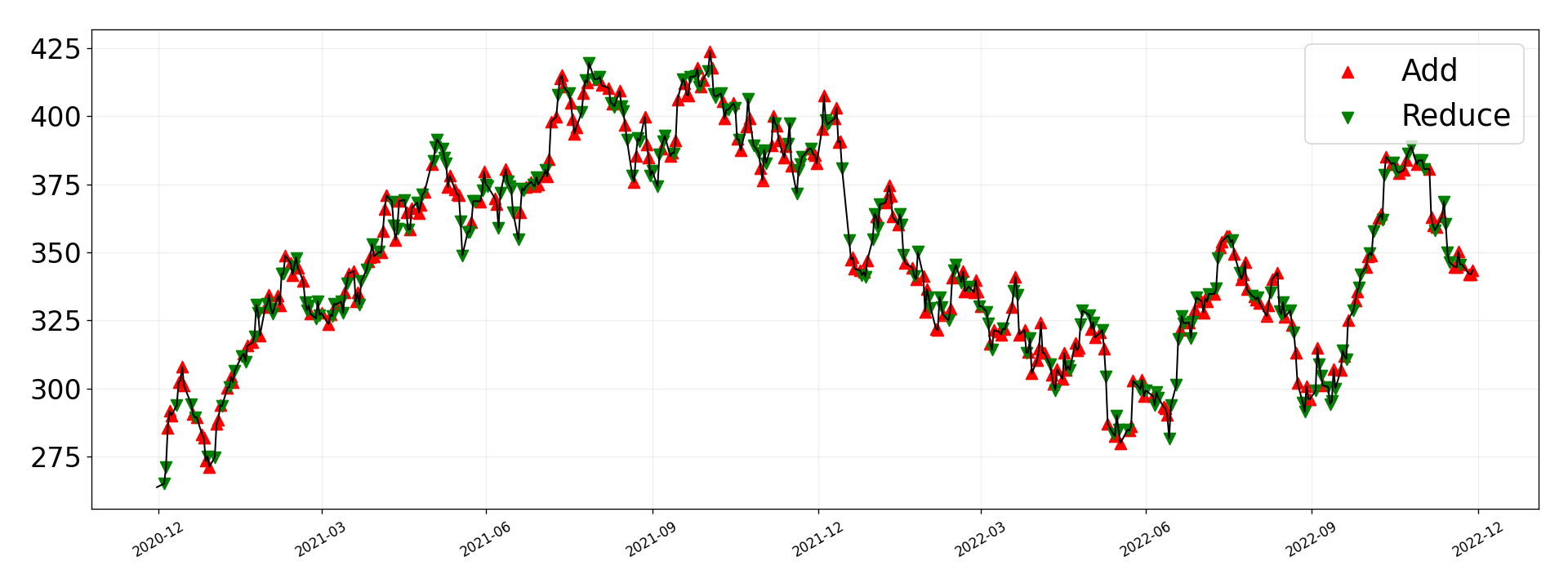}}
	\subfloat[SBC]{\includegraphics[width=0.24\textwidth]{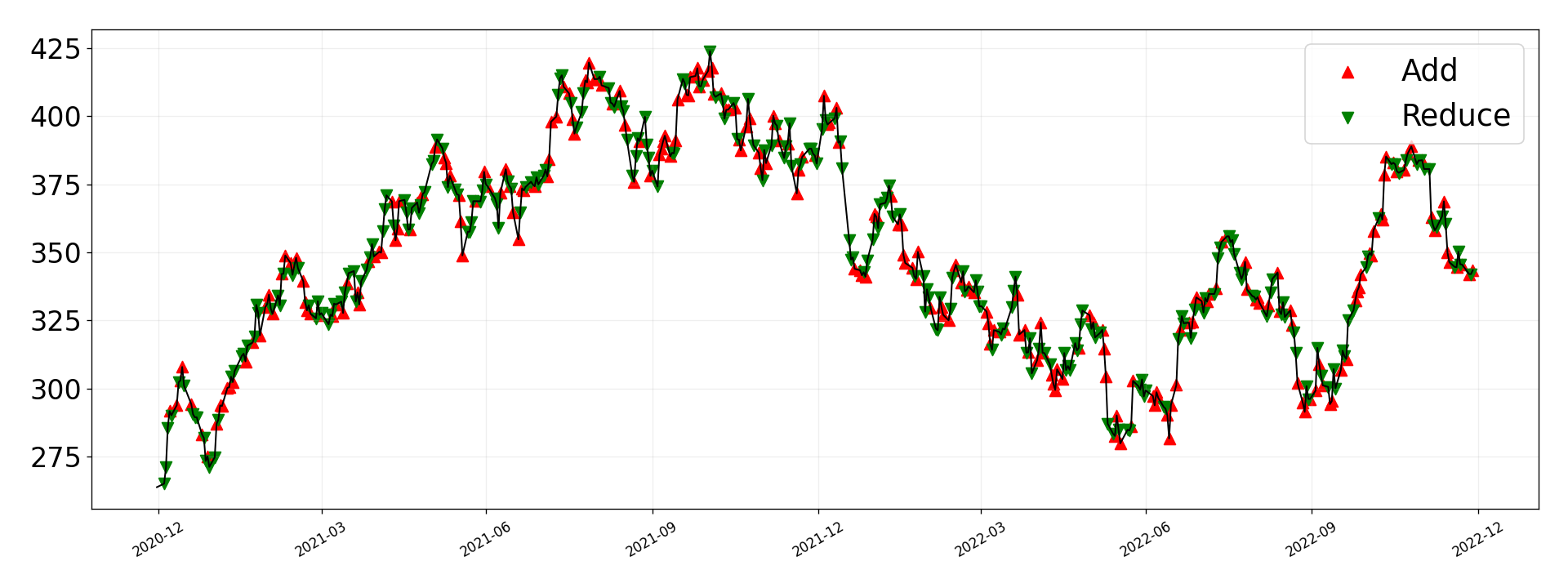}}
	\subfloat[SBCA]{\includegraphics[width=0.24\textwidth]{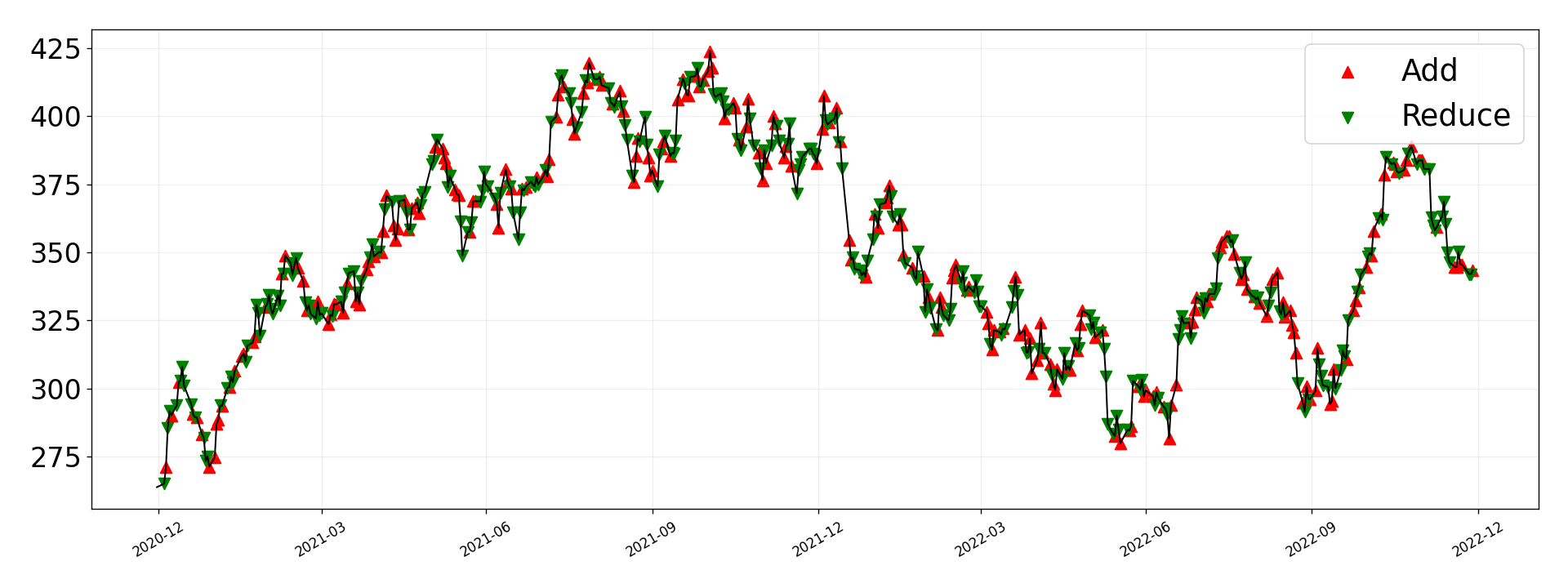}}
	\caption{Trading signals of GS in 6-asset portfolio.}
\end{figure}

\subsubsection{KO}
\begin{figure}[H]
	\centering
	\subfloat[SB]{\includegraphics[width=0.24\textwidth]{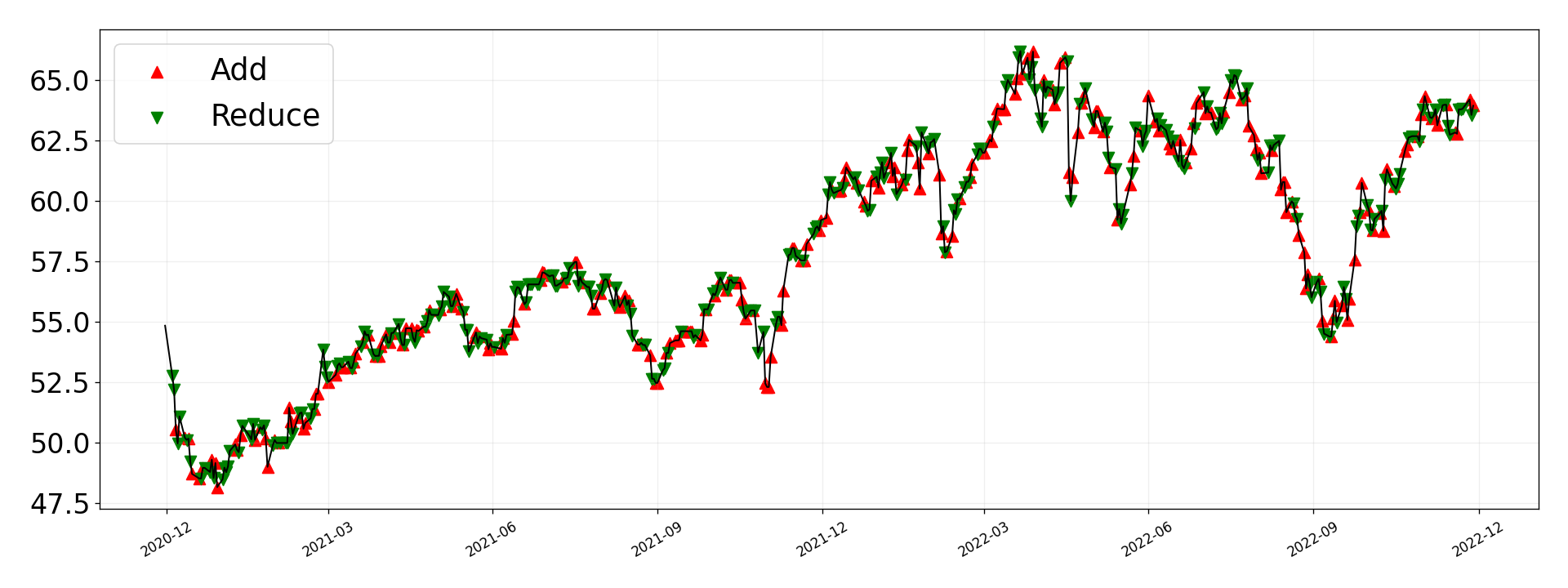}}
	\subfloat[SBA]{\includegraphics[width=0.24\textwidth]{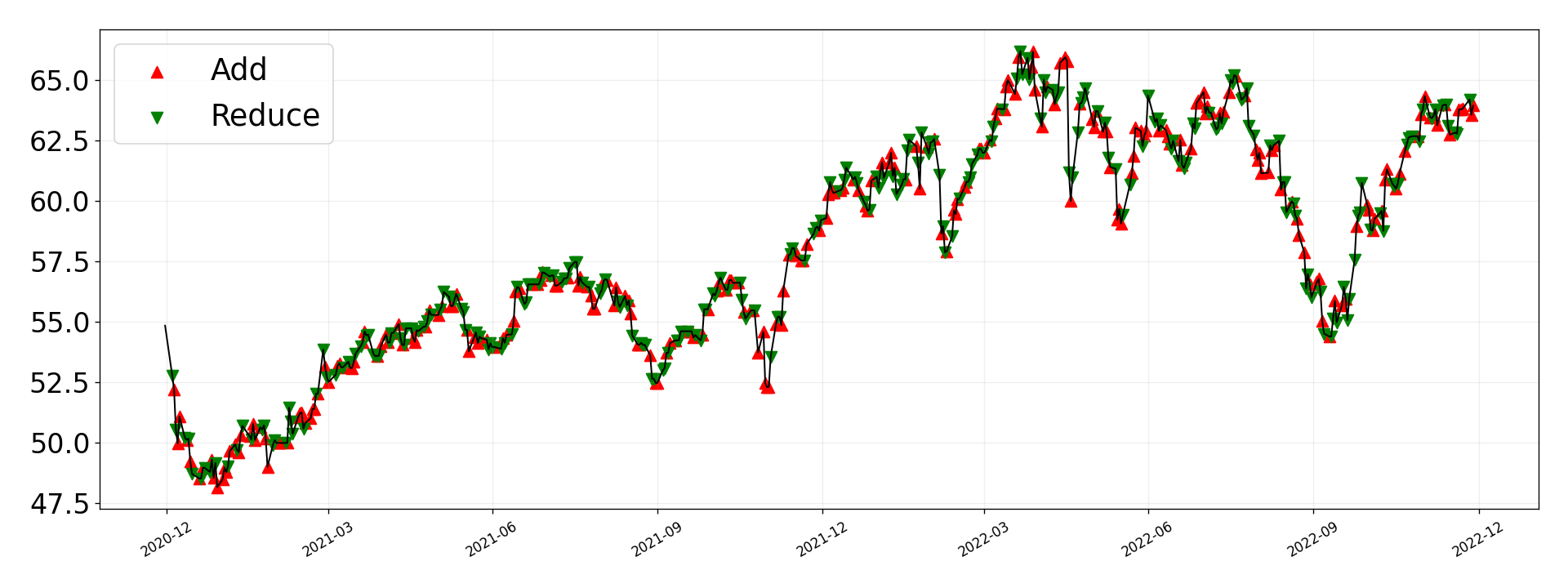}}
	\subfloat[SBC]{\includegraphics[width=0.24\textwidth]{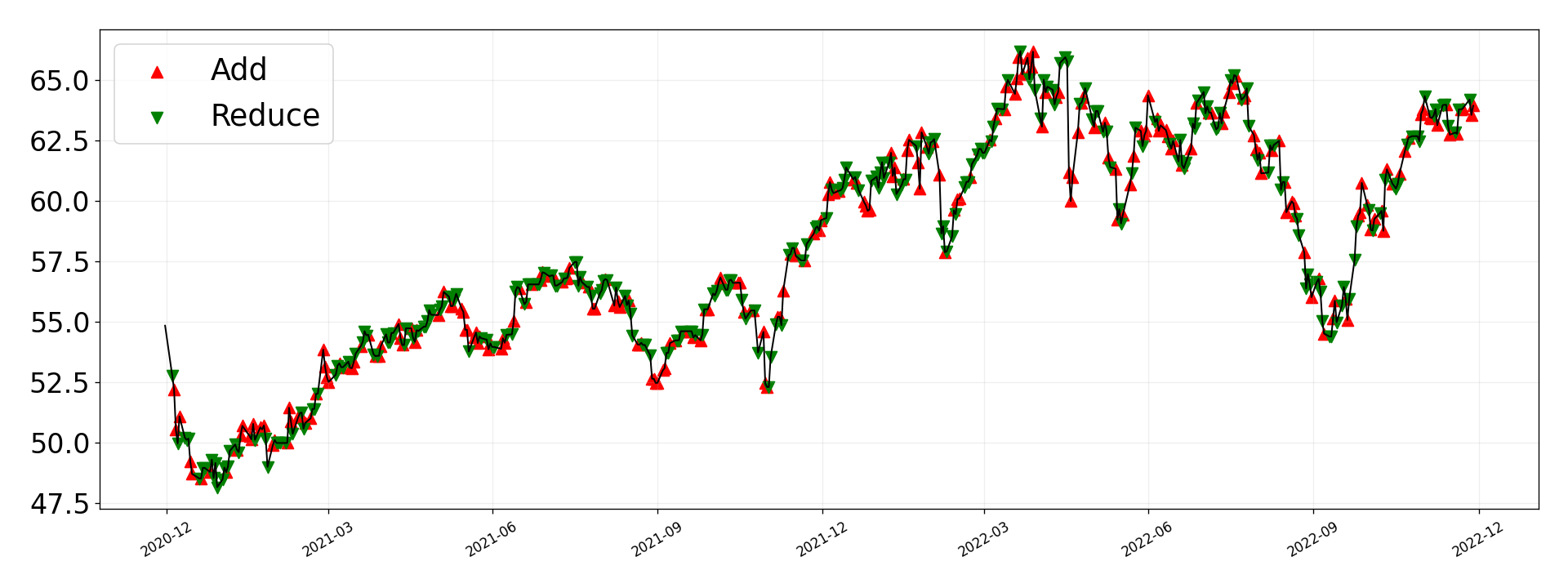}}
	\subfloat[SBCA]{\includegraphics[width=0.24\textwidth]{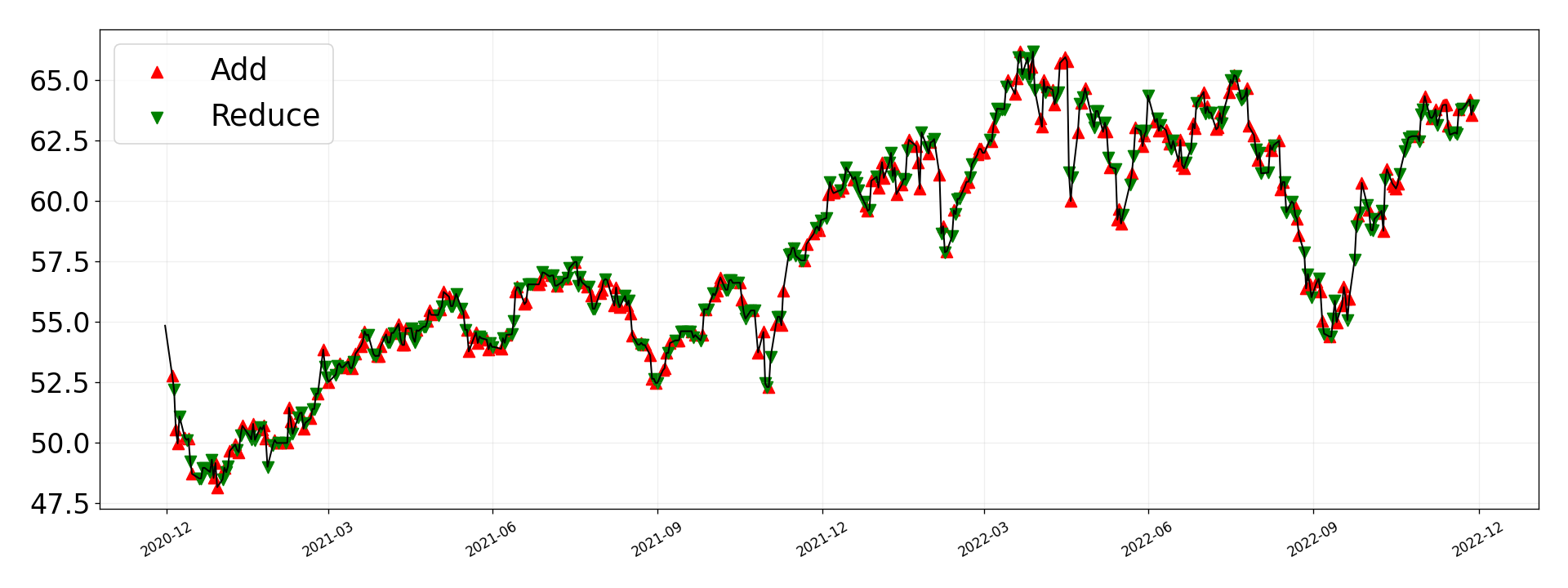}}
	\caption{Trading signals of KO in 6-asset portfolio.}
\end{figure}

\subsubsection{MRK}
\begin{figure}[H]
	\centering
	\subfloat[SB]{\includegraphics[width=0.24\textwidth]{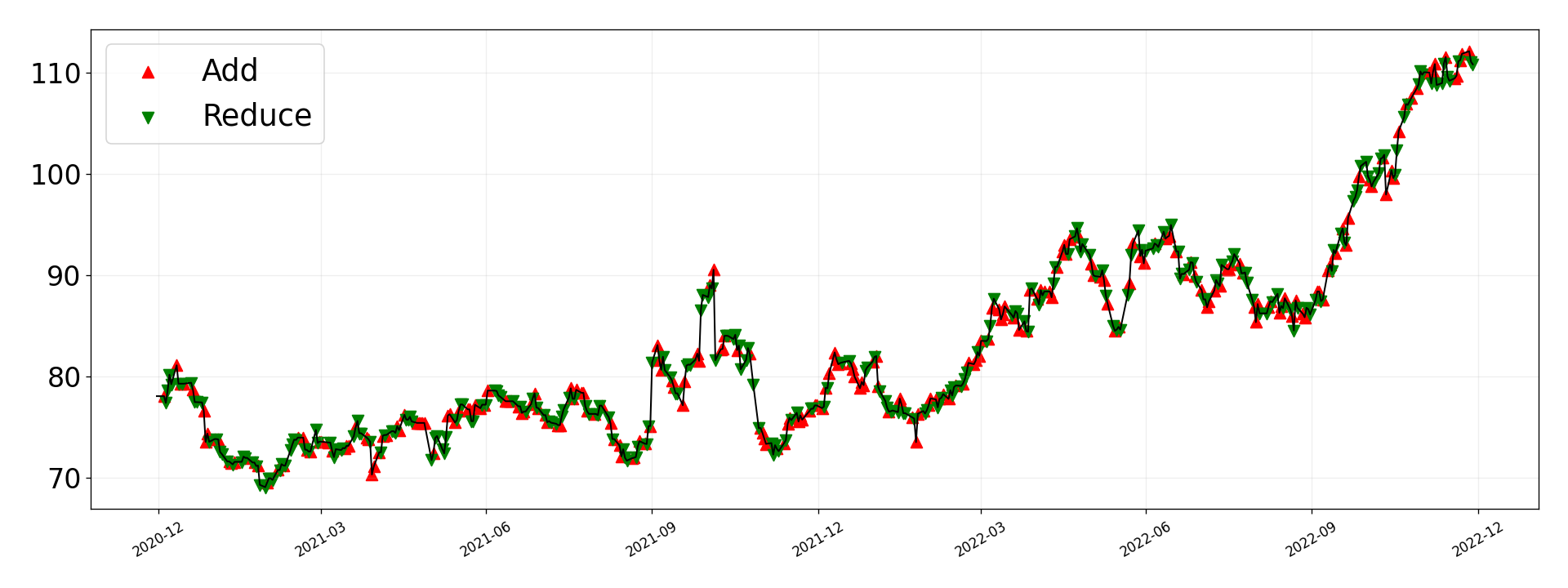}}
	\subfloat[SBA]{\includegraphics[width=0.24\textwidth]{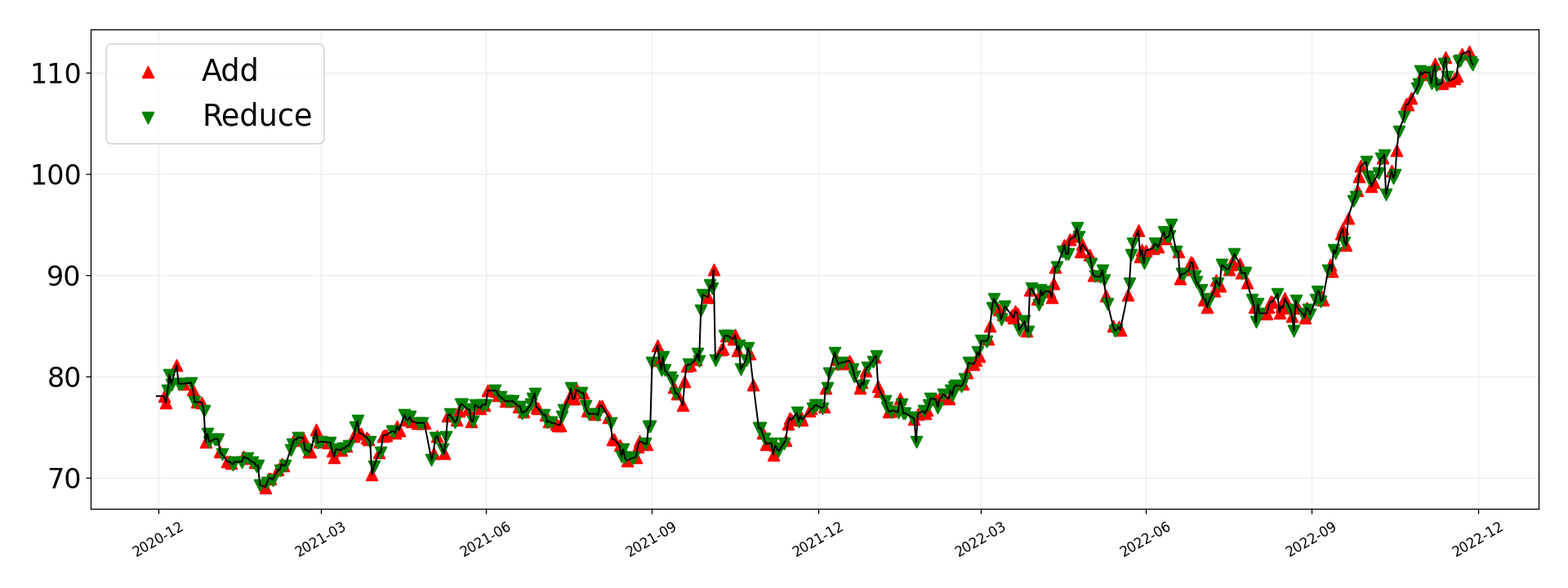}}
	\subfloat[SBC]{\includegraphics[width=0.24\textwidth]{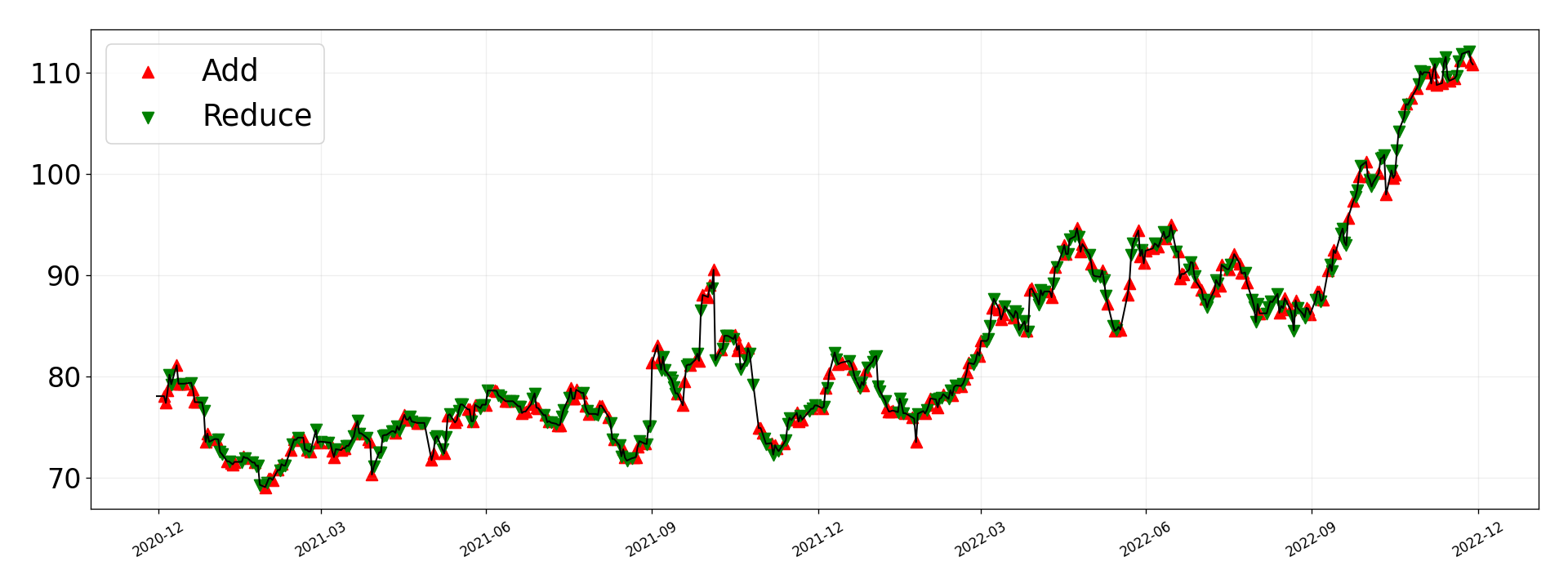}}
	\subfloat[SBCA]{\includegraphics[width=0.24\textwidth]{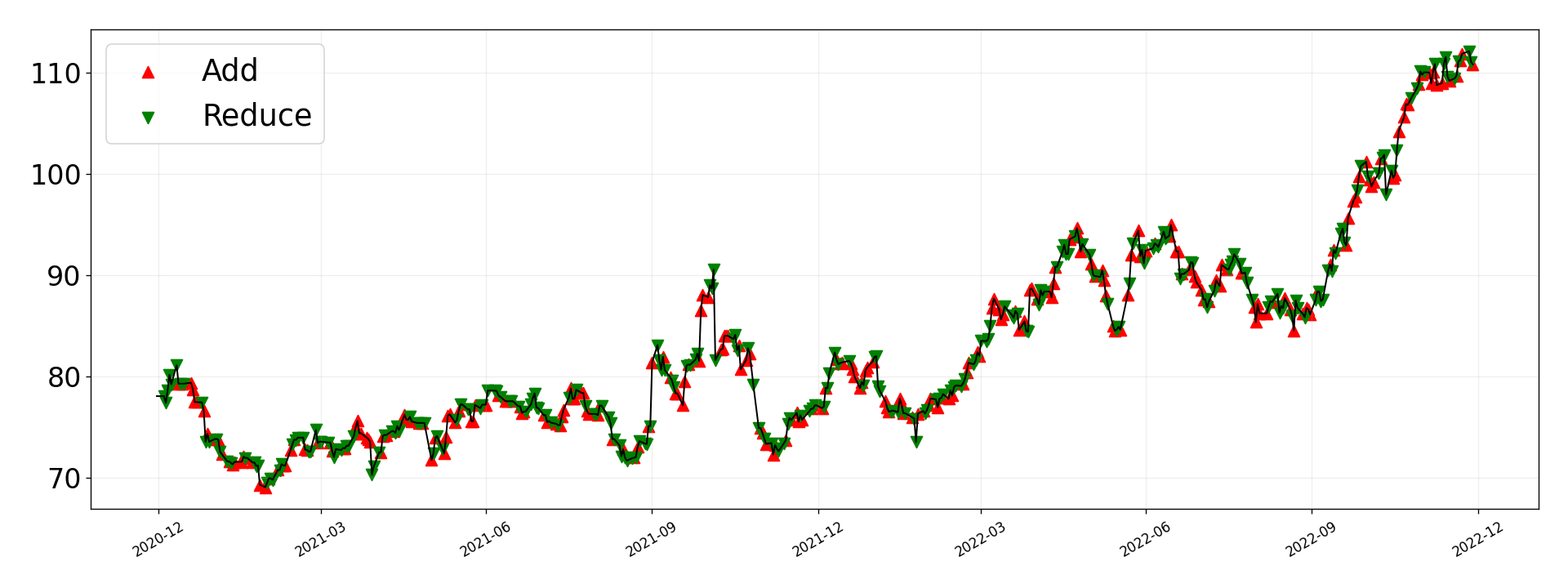}}
	\caption{Trading signals of MRK in 6-asset portfolio.}
\end{figure}

\subsubsection{NVDA}
\begin{figure}[H]
	\centering
	\subfloat[SB]{\includegraphics[width=0.24\textwidth]{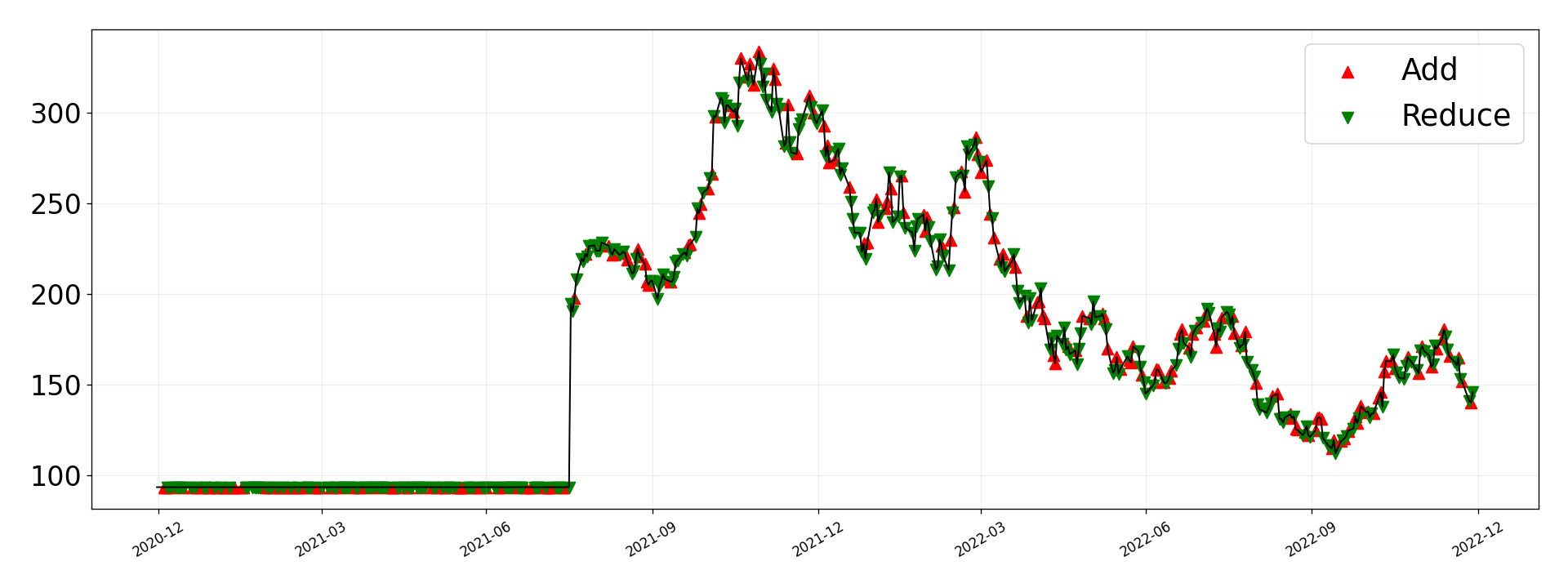}}
	\subfloat[SBA]{\includegraphics[width=0.24\textwidth]{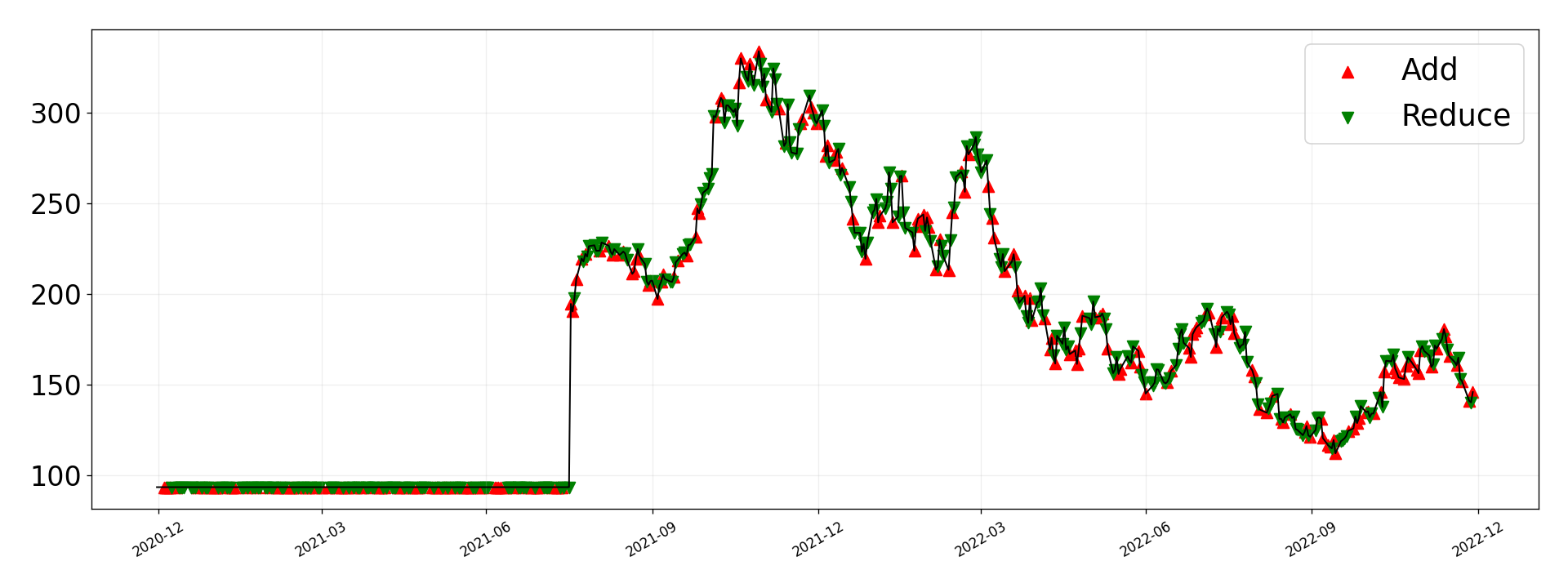}}
	\subfloat[SBC]{\includegraphics[width=0.24\textwidth]{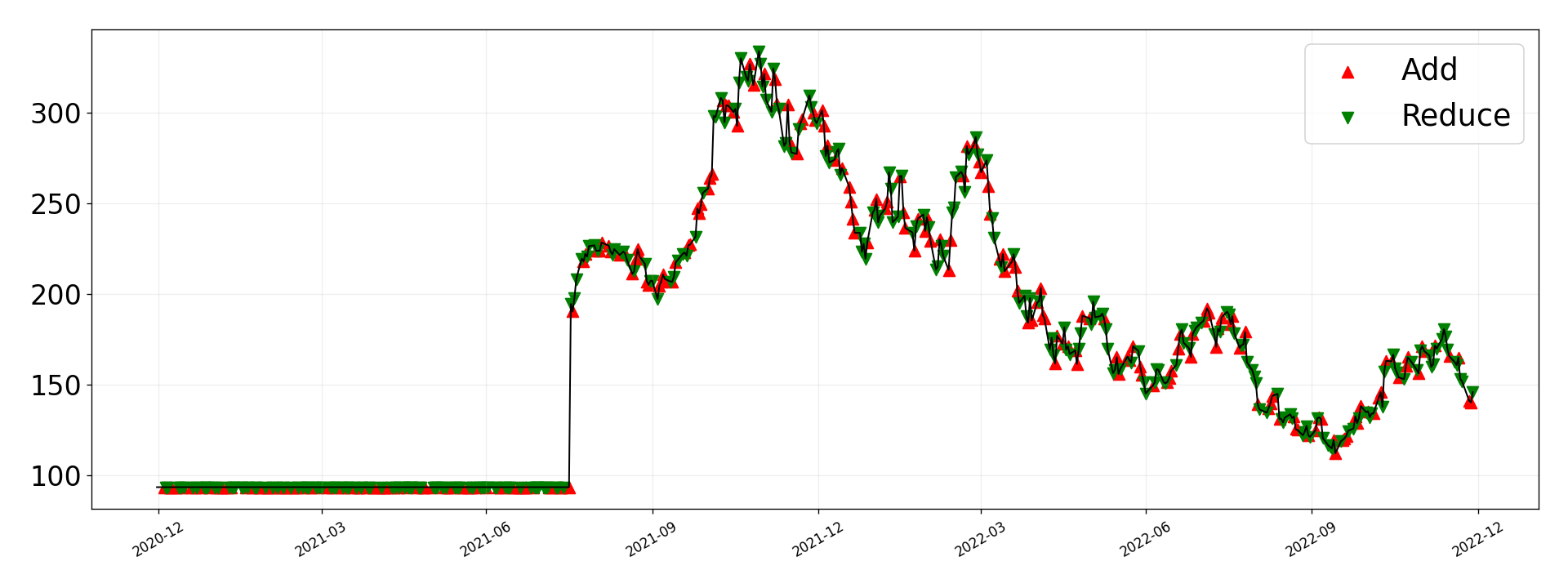}}
	\subfloat[SBCA]{\includegraphics[width=0.24\textwidth]{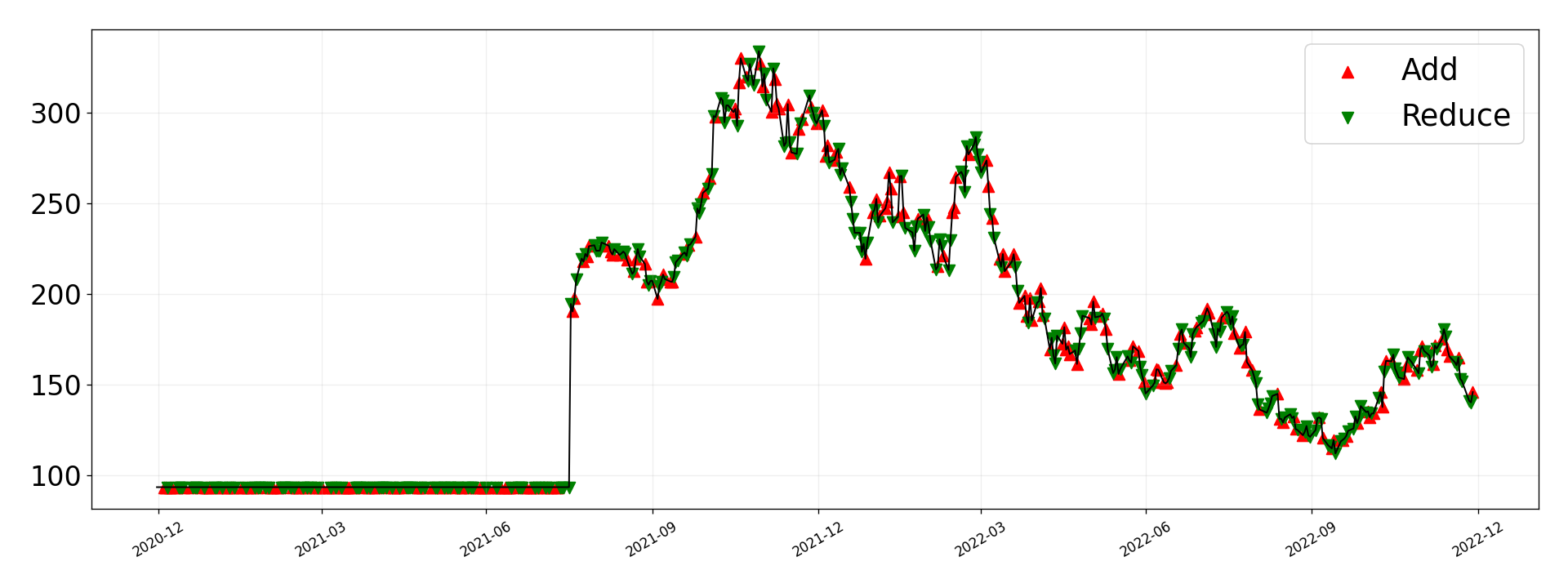}}
	\caption{Trading signals of NVDA in 6-asset portfolio.}
\end{figure}
\newpage
\bibliographystyle{apalike} 
\bibliography{sn-bibliography} 

\end{document}